\RequirePackage{etex}
\documentclass[11pt,letterpaper]{article}

\usepackage[margin=1in]{geometry}
\usepackage{amsmath,amssymb,amsfonts,amsthm}
\usepackage{mathtools}
\usepackage{microtype}
\usepackage{graphicx}
\usepackage[T1,OT1]{fontenc}
\usepackage{booktabs,array}
\usepackage{float}
\usepackage[dvipsnames,table]{xcolor}
\usepackage{tikz}
\usetikzlibrary{arrows.meta,positioning,fit,calc}
\usepackage{enumitem}
\usepackage{xspace}
\usepackage{natbib}
\bibpunct{[}{]}{,}{a}{,}{,}
\usepackage{url}
\usepackage{authblk}

\usepackage{tcolorbox}
\usepackage{hyperref}
\usepackage{cleveref}
\newcommand{\declarecolor}[2]{\definecolor{#1}{RGB}{#2}\expandafter\newcommand\csname #1\endcsname[1]{\textcolor{#1}{##1}}}
\declarecolor{White}{255, 255, 255}
\declarecolor{Black}{0, 0, 0}
\declarecolor{Maroon}{128, 0, 0}
\declarecolor{Coral}{255, 127, 80}
\declarecolor{Red}{182, 21, 21}
\declarecolor{LimeGreen}{50, 205, 50}
\declarecolor{DarkGreen}{0, 80, 0}
\declarecolor{Purple}{146, 42, 158}
\declarecolor{Navy}{0, 0, 128}
\declarecolor{LightBlue}{84, 101, 202}
\definecolor{mydarkblue}{rgb}{0,0.08,0.45}

\usepackage{tabularx}
\usepackage{xcolor}	

\definecolor{CardinalRed}{HTML}{C41E3A}
\definecolor{Dartmouth}{HTML}{00693E}
\definecolor{SapphireBlue}{HTML}{0F52BA}

\colorlet{MyRed}{CardinalRed}
\colorlet{MyGreen}{Dartmouth}

\colorlet{MyLightRed}{MyRed!25}
\colorlet{MyLightGreen}{MyGreen!25}

\colorlet{AlertColor}{MyRed}	
\colorlet{BadColor}{MyRed}	
\colorlet{FocusColor}{MyRed}	
\colorlet{GoodColor}{MyGreen}	
\colorlet{MacroColor}{MyRed}	

\hypersetup{ %
    hypertexnames=false,
    pdftitle={},
    pdfkeywords={},
    pdfborder=0 0 0,
    pdfpagemode=UseNone,
    colorlinks=true,
    linkcolor=MyRed,
    citecolor=MyGreen,
    filecolor=Purple,
    urlcolor=blue,
}

\usepackage{tcolorbox}
\definecolor{TheoremBlueFrame}{RGB}{25,92,170}
\definecolor{LemmaGreenFrame}{RGB}{46,125,50}
\definecolor{TheoremBlueBack}{RGB}{245,249,255}
\definecolor{LemmaGreenBack}{RGB}{247,252,247}

\usepackage{dsfont}

\tcbuselibrary{breakable}

\newenvironment{lemmabox}
  {\begin{tcolorbox}[
    breakable,
    colframe=green!5,
    colback=green!5,
    top=2pt,
    left=2pt,
    right=2pt,
    bottom=2pt
  ]}
  {\end{tcolorbox}}

\newenvironment{definitionbox}
  {\begin{tcolorbox}[
    breakable,
    colframe=red!5,
    colback=red!5,
    top=2pt,
    left=2pt,
    right=2pt,
    bottom=2pt
  ]}
  {\end{tcolorbox}}

\newenvironment{theorembox}
  {\begin{tcolorbox}[
    breakable,
    colframe=TheoremBlueBack,
    colback=TheoremBlueBack,
    top=2pt,
    left=2pt,
    right=2pt,
    bottom=2pt
  ]}
  {\end{tcolorbox}}

\allowdisplaybreaks
\setlist[itemize]{leftmargin=2em,itemsep=2pt,topsep=4pt}
\setlist[enumerate]{leftmargin=2em,itemsep=2pt,topsep=4pt}

\theoremstyle{plain}
\newtheorem{theorem}{Theorem}[section]
\newtheorem{lemma}[theorem]{Lemma}

\newtheorem{corollary}[theorem]{Corollary}

\theoremstyle{definition}
\newtheorem{definition}[theorem]{Definition}

\theoremstyle{remark}
\newtheorem{remark}[theorem]{Remark}

\newcommand{\PPAD}{\ensuremath{\mathsf{PPAD}}\xspace}
\newcommand{\PPADS}{\ensuremath{\mathsf{PPADS}}\xspace}
\newcommand{\PPA}{\ensuremath{\mathsf{PPA}}\xspace}
\newcommand{\PLS}{\ensuremath{\mathsf{PLS}}\xspace}
\newcommand{\SOPL}{\ensuremath{\mathsf{SOPL}}\xspace}
\newcommand{\EOPL}{\ensuremath{\mathsf{EOPL}}\xspace}

\newcommand{\FIXP}{\ensuremath{\mathsf{FIXP}}\xspace}
\newcommand{\LinFIXP}{\textnormal{\textsc{Linear-FIXP}}\xspace}
\newcommand{\SoL}{\textnormal{\textsc{Sink-of-Line}}\xspace}
\newcommand{\PIN}{\textnormal{\textsc{Positive-Index Nash}}\xspace}
\newcommand{\PINNR}{\textnormal{\textsc{Positive-Index or Nonregular Nash}}\xspace}
\newcommand{\NE}{\ensuremath{\mathrm{NE}}}
\newcommand{\supp}{\operatorname{supp}}
\newcommand{\BR}{\operatorname{BR}}
\newcommand{\ind}{\operatorname{ind}}
\newcommand{\sign}{\operatorname{sign}}

\newcommand{\one}{\mathbf{1}}
\newcommand{\zero}{\mathbf{0}}

\newcommand{\poly}{\operatorname{poly}}
\newcommand{\transpose}{\mathsf{T}}

\title{\fontsize{18}{22}\selectfont Finding a Positive Index Nash Equilibrium is PPADS-Complete}

\author[1,2]{Andreas Kontogiannis}
\author[3]{Ioannis Panageas}
\author[2,4]{Vasilis Pollatos}
\author[3]{Jingming Yan}

\affil[1]{National Technical University of Athens}
\affil[2]{Archimedes, Athena Research Center, Greece}
\affil[3]{University of California, Irvine}
\affil[4]{National and Kapodistrian University of Athens}

\date{}

\begin{document}
\maketitle

\begin{abstract}
Every nondegenerate bimatrix game has a Nash equilibrium of Shapley index $+1$, since all equilibria are isolated, have index $\pm1$, and their indices sum to $+1$. We prove that the following promise search problem is $\PPADS$-complete: given a rational bimatrix game promised to be nondegenerate, find an exact Nash equilibrium of index $+1$. The nondegeneracy promise can be removed if nonregular equilibria are also accepted: finding either a regular Nash equilibrium of index $+1$ or a nonregular Nash equilibrium in an arbitrary rational bimatrix game is also $\PPADS$-complete. To our knowledge, this is the first $\PPADS$-complete equilibrium search problem whose instances are explicit rational normal form payoff matrices, rather than succinct circuits or Turing machines, and thereby addresses an open question posed by
\cite{Daskalakis2019}.
\end{abstract}

\newpage

\tableofcontents

\newpage

\section{Introduction}
Nash equilibrium is the standard solution concept in noncooperative game
theory \citep{Nash1951}. In a bimatrix game, two players choose probability
distributions over finite sets of actions, and a Nash equilibrium is a pair of
strategies from which neither player can profitably deviate. Nash's theorem
guarantees that an equilibrium always exists, but it does not specify which
equilibrium should be selected when there are several. For nondegenerate
bimatrix games, a classical way to distinguish among equilibria is the
\emph{index} introduced by \cite{Shapley1974}; see also the
determinant formulation of \citet[Definition~11 and
Theorem~13]{vonStengel2021}. Every equilibrium then has index either $+1$ or
$-1$, and the sum of the indices of all equilibria is $+1$
\citep{Shapley1974,vonStengel2021}. In particular,
every nondegenerate game has at least one equilibrium of index $+1$.
Positive index is also an established equilibrium selection criterion with
independent game theoretic interpretations
\citep{VonSchemdeVonStengel2008,GovindanLarakiPahl2023}.

The index is closely connected to the geometry of the Lemke--Howson algorithm
\citep{LemkeHowson1964}. After adjoining the artificial point $(0,0)$ and
assigning it orientation $-1$, the complementary-pivot graph associated with each
fixed missing label forms paths and cycles, and the two endpoints of every
path have opposite index \citep{Shapley1974}. Thus the index
is not merely an algebraic label attached to an equilibrium: it determines the
direction in which the classical equilibrium path enters or exits that
equilibrium. This raises a basic computational question:

\begin{center}
    \emph{How hard is it to compute a Nash equilibrium with positive index?}
\end{center}

The unrestricted problem of computing a Nash equilibrium of a bimatrix game
is \PPAD-complete \citep{ChenDengTeng2009}. The class \PPAD, introduced by  \cite{Papadimitriou1994}, is based on the following directed
parity argument.
An implicitly represented directed graph with indegree and outdegree at most
one is given together with a distinguished source; the task is to find another
source or a sink. The sink only version of this principle defines the class
\PPADS: given the same input, the output must be a sink. Since every sink is a
valid end of line solution, $\PPAD\subseteq\PPADS$. No reverse inclusion is
known. \cite{BeameEtAl1998} establish generic relativized separations among these
search classes, showing that such a reverse inclusion
cannot be obtained by relativizing techniques.

Despite its long history, \PPADS\ remains considerably less understood
through natural complete problems than \PPAD. Comparatively few complete
problems are known for the class, and its standard representatives such as
\SoL, Positive Sperner, the circuit defined Injective Pigeonhole problem,
and their variants closely expose the underlying directed parity argument
\citep{Papadimitriou1994,BeameEtAl1998,GoosEtAl2024}. Indeed, 
\cite{BeameEtAl1998} explicitly described Positive Sperner as a natural
\PPADS-complete problem. Its computational formulation is nevertheless
succinct: the coloring of an exponentially large triangulation is specified
intensionally rather than listed as explicit finite data. 
\cite{FriedlEtAl2006} gave a \(\PPADS\)-complete locally two-dimensional
Sperner problem on a specially constructed surface. Its input, however,
contains a Turing machine that succinctly labels an exponentially large
triangulation. Similarly, the multiple-source and multiple-sink variants
studied by \cite{HollenderGoldberg2018} remain
succinct path problems. This distinction between complete problems defined
by concise representations and familiar problems on explicitly presented
data is part of the broader representation-aware study of total search
\citep{GoldbergPapadimitriou2018}.
In his ICM lecture, \citet[Open Question~8]{Daskalakis2019} asked for a
natural $\PPADS$-complete problem whose input contains neither a circuit nor
a Turing machine.

We answer this question with a classical equilibrium selection problem.
We prove that computing an exact positive index Nash equilibrium of an
explicitly given nondegenerate bimatrix game is $\PPADS$-complete.
The input consists only of the two rational payoff matrices of an
ordinary bimatrix game, and the sign of an equilibrium plays exactly the role
of the direction of an endpoint: negative index equilibria correspond to
sources, whereas positive index equilibria correspond to sinks.  Restricting the output of
equilibrium computation from an arbitrary Nash equilibrium to one of index
$+1$ therefore changes the canonical parity problem from
\textsc{End-of-Line} to its sink only counterpart.

The nondegeneracy promise is essential to the formulation.  It makes every
equilibrium isolated and regular, and hence gives every equilibrium a point
index in $\{-1,+1\}$.  A degenerate game can instead have an equilibrium
continuum: the zero payoff game, for example, has no isolated equilibrium to
which the point-index definition applies.  The promise is also algorithmically
substantive: deciding degeneracy is NP-complete, or equivalently deciding
nondegeneracy is coNP-complete, even for sparse bimatrix games \citep{Du2013}.
We therefore treat nondegeneracy as part of the input specification.  The hardness reduction outputs a globally
nondegenerate rational game.

It is natural to ask whether the nondegeneracy promise can be removed while
retaining a genuine total search problem.  A standard approach in the
complexity literature is to totalize a promise problem by admitting
\emph{violations} as alternative solutions: besides the intended solution,
the output relation accepts an efficiently verifiable witness that the
promise has failed
\citep{DaskalakisPapadimitriou2011,
FearnleyGordonMehtaSavani2020,
FearnleyGoldbergHollenderSavani2023}.
Following this convention, we define a promise free version of our problem
on arbitrary rational bimatrix games.  A valid output is either a regular
Nash equilibrium of index \(+1\), or a nonregular Nash equilibrium, which
serves as a violation of the nondegeneracy condition relevant to our
indexed formulation.  Our result therefore also implies that every rational
bimatrix game has either a regular Nash equilibrium of index \(+1\) or a
nonregular Nash equilibrium; equivalently, if all of its Nash equilibria are
regular, then at least one has index \(+1\).  This yields a genuine total
search problem without weakening the hardness result: the games produced by
our reduction are globally nondegenerate, so no violation can occur on the
hard instances.

To our knowledge, this is the first $\PPADS$-complete equilibrium search
problem whose instances are explicit rational normal form payoff matrices.
Positive Sperner is also a natural complete problem, but its input
succinctly specifies the coloring of an exponentially large triangulation
\citep{BeameEtAl1998}.

\subsection{Positive index and equilibrium selection}
\label{subsec:motivation-positive-index}

Positive index connects several classical approaches to equilibrium
selection, including strategic stability
\citep{KohlbergMertens1986}, equilibrium tracing
\citep{HarsanyiSelten1988}, and evolutionary dynamics.
The differentiable approach of \citet{Ritzberger1994} relates
equilibrium indices to robustness under payoff perturbations and
to dynamics on mixed strategies. While either nonzero index
ensures local persistence under small payoff perturbations,
asymptotic stability imposes a stronger restriction. Building on
the relationship between evolutionary and strategic stability of the work of
\cite{RitzbergerWeibull1995},
\citet{DeMichelisGermano2000} show that an isolated asymptotically
stable equilibrium must have index $+1$ for dynamics satisfying
their Nash-field hypotheses. The converse need not hold:
positive index does not guarantee attraction under a particular
learning rule, such as replicator dynamics.

Positive index also characterizes \emph{sustainability}: the
possibility of making an equilibrium unique by adding finitely
many strategies that are inferior replies at that equilibrium.
\citet{VonSchemdeVonStengel2008} establish this connection for
nondegenerate bimatrix games, and \citet{GovindanLarakiPahl2023}
prove that an isolated equilibrium of a finite game is sustainable
if and only if its index is $+1$. Our theorem therefore also
establishes that finding a sustainable equilibrium of a
nondegenerate bimatrix game is $\PPADS$-complete.

Quantal response equilibrium provides another connection to
selection. Players choose better actions more often while
allowing decision errors. At zero payoff sensitivity, play is
uniform and the equilibrium is unique, with index $+1$.
Following the logit equilibrium curve from this starting point
selects a limiting Nash equilibrium as decision noise vanishes,
for almost every game
\citep{McKelveyPalfrey1995,Turocy2005}. The positive index at the
starting point determines the orientation of the curve and
forces its limiting Nash equilibrium to have index $+1$.
Our problem keeps this positive-index requirement while
accepting any equilibrium satisfying it, rather than the
particular equilibrium reached by the procedure. Even this
broader selection problem is $\PPADS$-complete.

\subsection{Our results}
We study the following exact search problem.

\medskip
\noindent\PIN
\begin{itemize}
    \item \emph{Input:} Two rational matrices $(A,B)$ defining a bimatrix
    game, under the promise that the game is nondegenerate.
    \item \emph{Output:} An exact Nash equilibrium $(x,y)$ satisfying
    $\ind_{A,B}(x,y)=+1$.
\end{itemize}

\medskip
\noindent
In this paper, we prove the following theorem.

\begin{theorembox}
\begin{theorem}[Main theorem]\label{thm:main-intro}
\PIN\ is $\PPADS$-complete.  Its membership
reduction is required only on promised inputs, and its hardness reduction
always outputs promised, globally nondegenerate games.
\end{theorem}
\end{theorembox}

The nondegeneracy promise can be removed if nonregular equilibria are also
accepted.  Here regularity concerns the returned equilibrium, not the
entire game; its precise definition is given in
Section~\ref{sec:preliminaries}.  We consider the following problem.

\medskip
\noindent\PINNR
\begin{itemize}
    \item \emph{Input:} Two rational matrices $(A,B)$ defining an arbitrary
    bimatrix game.
    \item \emph{Output:} An exact Nash equilibrium $(x,y)$ that is either regular with $\ind_{A,B}(x,y)=+1$, or
    nonregular.
\end{itemize}

\begin{theorembox}
\begin{corollary}\label{cor:positive-or-nonregular}
\PINNR\ is $\PPADS$-complete.
\end{corollary}
\end{theorembox}

The proof is given in Section~\ref{subsec:positive-or-nonregular}.
Hardness follows from the same nondegenerate output games as in
Theorem~\ref{thm:main-intro}.

For membership, we give a polynomial time reduction
\[
   \PINNR\ \leq_p\ \PIN,
\]
whose output game is nondegenerate.  No index is assigned to a nonregular
equilibrium.

Together with the $\PPAD$-completeness of unrestricted bimatrix Nash
equilibrium, this gives a sharp contrast: allowing an arbitrary Nash
equilibrium captures $\PPAD$, whereas requiring a positive index equilibrium
captures the sink only class $\PPADS$.

The output in Theorem~\ref{thm:main-intro} is exact.  This is
important for two reasons.  First, the index is a local orientation and cannot
in general be recovered from a weak approximate equilibrium.  Second, every
regular equilibrium of a rational bimatrix game is rational and has polynomial
encoding length: once its supports are known, its probabilities are obtained
from nonsingular rational linear systems.  Thus exactness does not take the
problem outside the usual discrete search framework.

\begin{table}[t]
\centering
\caption{Directed-parity search and equilibrium computation.}
\label{tab:landscape}
\renewcommand{\arraystretch}{1.12}
\small
\begin{tabular}{@{}>{\raggedright\arraybackslash}p{0.18\linewidth}
                    >{\raggedright\arraybackslash}p{0.23\linewidth}
                    >{\raggedright\arraybackslash}p{0.33\linewidth}l@{}}
\toprule
Search task & Accepted output & Representation / promise & Complexity \\
\midrule
\textsc{End-of-Line} & another source or a sink & succinct successor/predecessor circuits & $\PPAD$-complete \\
Positive Sperner & a positively oriented simplex & succinctly labeled triangulation & $\PPADS$-complete \\
Bimatrix Nash & an arbitrary Nash equilibrium & explicit rational payoff matrices & $\PPAD$-complete \\
\rowcolor{green!8}\PIN & a $+1$ Nash equilibrium & explicit rational payoff matrices; nondegeneracy promise & $\PPADS$-complete \\
\bottomrule
\end{tabular}
\end{table}

Table~\ref{tab:landscape} places the result alongside the canonical directed
parity problems and unrestricted bimatrix equilibrium computation.
Theorem~\ref{thm:main-intro} has two parts with rather different proofs.
Membership follows directly from the oriented Lemke--Howson graph of the
nondegenerate input game.  The hardness proof is sign sensitive: it starts
from \SoL\ and constructs a nondegenerate game in which every equilibrium of
index $+1$ decodes a sink of the input graph.  It is not enough to invoke an
arbitrary reduction from \textsc{End-of-Line} to Nash equilibrium, since a
reduction that merely preserves the existence of a solution may exchange,
merge, or create fixed points of opposite index.  Our reduction instead lets
one follow the determinant defining the local index through every stage.

\begin{remark}[Independent work]
For \textsc{End-of-Line} instances in which no noncanonical vertex
is both a source and a sink (i.e., no $v\neq 0^n$ satisfies both
$P(S(v))\neq v$ and $S(P(v))\neq v$, where $S$ and $P$ are the
successor and predecessor circuits), our reduction is parsimonious:
it gives a bijection between the noncanonical sources and sinks
of the \textsc{End-of-Line} instance and the negative- and
positive-index equilibria, respectively, of the constructed game.
This resolves an open question raised by \citet{ChenDengTeng2009}.  Ghosh, Goldberg, and Hollender independently and before us
obtained this result in their work on computing a unique Nash
equilibrium~\citep{GhoshGoldbergHollender2026}.  Their manuscript was not
publicly available while our proof was developed and we never reviewed their paper. We are aware of their result from a talk Alexandros Hollender gave. Our reduction and its proof
were obtained independently and our result is stronger: beyond parsimony, our reduction preserves endpoint signs as equilibrium
indices and produces a globally nondegenerate game, ensuring that every
positive index equilibrium decodes a sink and thereby yielding
$\PPADS$-hardness.
\end{remark}
\paragraph{Notation.}
For a positive integer $d$, write $[d]:=\{1,\ldots,d\}$.  For
$t\in\mathbb R$, let $[t]_+:=\max\{t,0\}$.  Given $a,b\in\mathbb R$ with
$a\le b$, define the clipping operator
\[
 [t]_a^b:=\min\{b,\max\{a,t\}\}
          =a+[t-a]_+-[t-b]_+.
\]
For $z\in\mathbb R^d$, both operations are applied coordinatewise; thus
$[z]_a^b:=([z_i]_a^b)_{i\in[d]}$.  In particular, $[z]_0^1$ is the
coordinatewise projection of $z$ onto $[0,1]^d$. We write $e_i$ for the $i$th
standard basis vector, $I_d$ for the $d\times d$ identity matrix, and
$\mathbf 0$ and $\mathbf 1$ for the all-zero and all-one vectors of the
dimension indicated by the context.  For a vector $z$ and a matrix $M$, $z_I$
and $M_{I,J}$ denote their restrictions to the indicated coordinates, with
$M_{II}:=M_{I,I}$; the superscript $\transpose$ denotes transpose.  For a
map $F$, $DF(p)$ denotes its Jacobian at $p$, and
$\sign(c)\in\{-1,+1\}$ denotes the sign of a nonzero scalar $c$. A linear complementarity problem (LCP) asks for vectors \(x,w\geq0\) satisfying
\(w=Mx-b\) and \(x_iw_i=0\) for every coordinate. For nonempty sets \(A,B\subseteq\mathbb R^d\), define their
\(\ell_\infty\)-distance by
$\operatorname{dist}_\infty(A,B)
    :=
    \inf\left\{
        \|x-y\|_\infty :
        x\in A,\ y\in B
    \right\}.$

\subsection{Technical overview}

Figure~\ref{fig:overview}
summarizes the reduction chain and the sign carried through its three hardness
stages.

\begin{figure}[t]
\centering
\begin{tikzpicture}[
  node distance=8mm and 7mm,
  box/.style={draw=black!55,rounded corners=2pt,fill=blue!5,
              minimum height=10mm,text width=28mm,align=center,font=\small},
  arr/.style={-{Latex[length=2mm]},thick,draw=black!65},
  sign/.style={font=\footnotesize,align=center,text=black!75}
]
\node[box] (sink) {\SoL};
\node[box,right=of sink] (fixp) {certified signed\\\LinFIXP\ map};
\node[box,right=of fixp] (lcp) {regular and strictly\\complementary LCP};
\node[box,right=of lcp] (game) {nondegenerate\\bimatrix game};
\draw[arr] (sink) -- (fixp);
\draw[arr] (fixp) -- (lcp);
\draw[arr] (lcp) -- (game);
\node[sign,below=5mm of sink] {$v$ is a sink};
\node[sign,below=5mm of fixp] {sink endpoint: $\ind_\Phi(p)=+1$};
\node[sign,below=5mm of lcp] {$\det M_{II}>0$};
\node[sign,below=5mm of game] {$\ind_{\NE}=+1$};
\end{tikzpicture}
\caption{Overview of the hardness reduction and the preserved positive sign.}
\label{fig:overview}
\end{figure}

\paragraph{Membership via Lemke--Howson.}
After adding a sufficiently large constant to every entry of each player's
payoff matrix, which preserves best responses, Nash equilibria, and their
indices, the two best response polyhedra have the standard labeled form.
For a fixed missing label, the complementary bases and pivots form the
Lemke--Howson graph.  By the nondegeneracy promise, this graph is a disjoint
union of paths and cycles.  Its endpoints are completely labeled vertices:
the artificial point \((0,0)\) and vertices corresponding to Nash
equilibria.  Assigning orientation sign \(-1\) to the artificial endpoint,
the Shapley--Todd determinant orientation gives opposite signs to the two
endpoints of every path and directs the path from its negative endpoint to
its positive endpoint
\citep{Shapley1974,Todd1976,VeghVonStengel2015}.

Each node is represented by a complementary basis.  Exact rational pivot and
determinant computations give its predecessor and successor in polynomial
time, yielding a polynomial-time reduction to \(\SoL\).  In particular, every
sink of the resulting directed graph corresponds to a Nash equilibrium of
index \(+1\).  No perturbation is used in this membership reduction.

\paragraph{A certified signed field.}
Section~\ref{sec:signed-fixed-point-unconditional} constructs a continuous
piecewise-affine field \(G\colon[0,1]^2\to\mathbb R^2\), computed by a
rational Linear-$\FIXP$ circuit, and the map
\(\Phi(x)=[x+G(x)/32]_0^1\). Its fixed points correspond to the sinks and
noncanonical sources of a locally represented graph \(H\) on the original
\SoL\ vertices. They are rational, regular, and lie in the interior of the
square. Every max gate in the circuit computing \(\Phi\) is \emph{strict}
at every fixed point, meaning that its two inputs have unequal values. The
index is \(+1\) at a sink and \(-1\) at a noncanonical source. A
polynomial-time decoder recovers a solution of \(\SoL\) from any positive
index fixed point. All circuits have polynomial size, and all rational
numbers have polynomial bit length. These properties form the certified
signed field of Theorem~\ref{thm:axis-unconditional-map}. Strictness will
later imply strict complementarity in the LCP and allows the determinant
sign defining the fixed-point index to be followed through the subsequent
reductions.

\paragraph{Construction of the signed field.}
We work directly with the original \SoL\ vertices. The graph \(H\) keeps a verified edge \(x\to S(x)\) exactly when \(S(x)\ne x\) and \(P(S(x))=x\); thus verified self-loops are discarded. This does not change proper sources or sinks: a verified self-loop has \(S(x)=P(x)=x\), so by functionality it has no verified non-self incident edge. 
The normalization makes \(0^n\) a proper source. Every sink \(u\) of \(H\) satisfies \(P(S(u))\ne u\), and following the unique path from \(0^n\) must reach
such a sink.

Identify the \(n\)-bit vertices with integers
\(u\in\{0,\ldots,L-1\}\), where \(L=2^n\). For each edge \(u\to v\) of
\(H\), we place vertex \(u\) at \((3u,3u)\) and draw the staircase path
\[
 (3u,3u)\longrightarrow(3v,3u)\longrightarrow(3v,3v).
\]
Each row carries at most one outgoing segment, and each column carries
at most one incoming segment, because indegrees and outdegrees are at
most one. Distinct paths may cross. At every crossing we use the planar switch of
\citet{ChenDeng2009}, reconnecting the two incoming pieces to the two
outgoing pieces without creating or destroying an endpoint. Hence the
planar routed graph has exactly the same sources and sinks as \(H\), though
the switch may change which source is connected to which sink.

We enlarge the routed drawing by a factor of three and place it in the
integer grid \(\{0,\ldots,M\}^2\), with
\(M=32L=2^{n+5}\). Route vertices have color zero, points immediately
to the directed left of a route have color one, and all remaining points
have color two, apart from a small modification near the canonical source.
We divide every grid square into the two triangles determined by the
diagonal from its upper-left to its lower-right corner and assign the field
values
\[
 q_0=e_2,\qquad q_1=e_1,\qquad q_2=(-1,-1).
\]
Affine interpolation on each triangle gives a continuous field \(G\).
Because \(q_0+q_1+q_2=0\), a triangle contains a zero exactly when it has
all three colors, and that zero is its barycenter. The discrete coloring
creates exactly one such triangle at every sink and every noncanonical
source, none along ordinary route pieces, turns, or switches, and none at
the distinguished source.


Let \(p\) be the barycenter of a fully colored endpoint triangle. On this triangle, the field has the form \(G(x)=A_p(x-p)\), where \(A_p\) is nonsingular.
If \(p_0,p_1,p_2\) denote the triangle's vertices in color order, then
\[
 E=[p_1-p_0\mid p_2-p_0],\qquad
 A_p=[q_1-q_0\mid q_2-q_0]E^{-1},\qquad
 \det(-A_p)=-\frac{3}{\det E}.
\]
The endpoint geometry then gives \(\det E<0\) at a sink and
\(\det E>0\) at a source, for either possible incident direction.
This therefore implies that \(\det(-A_p)>0\) at sinks and \(\det(-A_p)<0\) at noncanonical sources. 
Since clipping is inactive near \(p\),
\[
\operatorname{ind}_{\Phi}(p)
=\operatorname{sign}\det\!\bigl(I-D\Phi(p)\bigr)
=\operatorname{sign}\det(-A_p).
\]
Thus sinks yield fixed points of index \(+1\), while noncanonical sources yield fixed points of index \(-1\).
The boundary coloring also ensures that clipping creates no additional fixed point.

The remaining challenge is to represent the interpolated field exactly by a
polynomial-size rational Linear-$\FIXP$ circuit. Although only the three
vertices of the containing triangle contribute at any point, identifying
those vertices involves a discrete choice. The natural approach is to
compute $\lfloor Mx_1\rfloor$ and $\lfloor Mx_2\rfloor$ and query the
nearby colors. However, floor is discontinuous, whereas every
Linear-$\FIXP$ circuit computes a continuous function. We must therefore
combine discrete color queries with continuous interpolation without
introducing any error in the field.

The key idea is to arrange that uncertainty in selecting a vertex occurs
only where its color has no effect on the interpolation. We partition
the grid vertices into nine groups according to their coordinates
modulo three. Within each group, the supports of the interpolation
weights are separated by gaps, and at most one vertex contributes at
any point. These gaps give us room to switch between candidate vertices:
we place the transition regions of the bit-extraction circuit entirely
inside them. Wherever a group has a nonzero interpolation weight, the
circuit identifies the contributing vertex and evaluates its color
exactly. Wherever bit extraction may be unreliable, the entire group
has zero interpolation weight, and a masking step forces its
contribution to zero. Summing the nine contributions thus recovers
$G$ exactly throughout the domain.

This construction uses interpolation weights and a residue
partition adapted from
\cite[Sections~4.2 and~5.1.1]{yarotsky2018optimal}, together with a masking mechanism related to
\citet[Section~5.2.4]{yarotsky2018optimal}.
The resulting circuit has size and coefficient bit length polynomial
in the coloring circuit size and $\log(M+1)$.
We additionally ensure that every max gate has unequal inputs at each
endpoint barycenter. Hence, the same affine branches remain
active in a neighborhood of every fixed point, as required by the
subsequent index-preserving reduction.

Finally, our construction ensures that \(\|G\|_\infty\le1\) and since
\(\ind_\Phi(p)=\sign\det(-A_p)\), from a positive-index fixed point, the decoder recovers the integer vertex \(u\) directly from the barycenter and returns its \(n\)-bit representation; because \(u\) is a sink of \(H\), it
satisfies \(P(S(u))\ne u\).


\paragraph{An exact \LinFIXP to LCP reduction.}
Building on the circuit to complementarity reduction of
\citet{Mehta2018}, we construct an LCP using a single set of complementarity
variables and a unit lower triangular matrix \(A\), unrelated to the
local matrices \(A_p\) above.  Working directly with the complementarity conditions of the max gates
avoids auxiliary dual variables and gives a bijection between fixed
points and complete LCP solution vectors. This formulation also makes
the determinant calculation for index preservation particularly direct.

After putting each max gate in the normalized form
\(x_i=\max\{0,\ell_i\}\), strictness at a fixed point \(p\) gives
\(\ell_i(p)\ne0\) for every max gate, where \(\ell_i(p)\) denotes the value
of this affine expression during the circuit evaluation at \(p\).  We denote
\(I=\{i:\ell_i(p)>0\}=\{i:x_i>0\}\).
Identifying the circuit outputs with its inputs imposes the fixed point
equations and produces the LCP matrix \(M=A-UV^{\transpose}\).
Here \(U\) and \(V\) encode the equations identifying the circuit outputs
with the corresponding inputs.  Although \(A\) is lower triangular, the
final matrix \(M\) need not be triangular.  Section~\ref{sec:fixp-to-lcp}
proves the identity
\[
\det M_{II}=\det\!\bigl(I_d-D\Phi(p)\bigr),
\qquad
\sign\det M_{II}=\ind_\Phi(p).
\]
Strictness ensures that the same affine input remains selected
at every max gate in a neighborhood of \(p\). Thus the determinant
identity requires no limiting or approximate argument.

\paragraph{From the LCP to a nondegenerate bimatrix game.}
We first apply Mehta's symmetric game construction and then the imitation game
correspondence of \citet{McLennanTourky2010}.  For an \(m\)-dimensional LCP
with \(M\in\mathbb{R}^{m\times m}\) and \(b\in\mathbb{R}^m\), satisfying
\(x\ge0\), \(Mx-b\ge0\), and \(x^{\transpose}(Mx-b)=0\), we define
\[
S=
\begin{pmatrix}
-M & b+\mathbf1\\
0^{\transpose} & 1
\end{pmatrix}.
\]
The first $m$ strategies correspond to the LCP coordinates, while the last row
and column introduce one additional anchoring strategy, denoted by $\star$.
This first produces the symmetric game \((S,S^{\transpose})\).
Lemma~\ref{lem:star-positive} uses the identities that anchor the distinguished
output coordinates, together with the unit lower triangular structure of
\(A\), to show that the additional strategy \(\star\) has positive probability
in every symmetric equilibrium.
Thus, an LCP solution \(x\) corresponds to
\(y(x)=(x,1)/(1+\mathbf1^{\transpose}x)\), and \(x\) is recovered by
dividing the first \(m\) coordinates of \(y(x)\) by \(y(x)_\star\).

We denote \(I(x)=\{i:x_i>0\}\) and \(K(x)=I(x)\cup\{\star\}\).  Strict
complementarity gives
\(\operatorname{supp}y(x)=\operatorname{BR}_S(y(x))=K(x)\), where $\operatorname{BR}_{S}(y)$ denotes the set of pure row-player best
responses to $y$ in the payoff matrix $S$.  We then
consider the imitation game \((S,I_{m+1})\).
Lemma~\ref{lem:full-imitation} gives a bijection between LCP solutions and
Nash equilibria \((u(x),y(x))\) of this game.  The equilibrium index
calculation gives
\[
\ind_{S,I_{m+1}}\bigl(u(x),y(x)\bigr)
=
\sign\det M_{I(x),I(x)}
=
\ind_\Phi(p).
\]

Every equilibrium of the imitation game is regular and has strict
best response inequalities outside its support, but the game may
still be degenerate at other mixed strategies.
Section~\ref{subsec:global-perturbation} first shifts the row payoffs to be positive and, for a
small parameter $\varepsilon>0$, divides row $i$ by $1+\varepsilon^i$. Any violation of nondegeneracy would force
one of a finite family of nonzero polynomials in \(\varepsilon\)
to vanish. Uniform degree and coefficient bounds give an explicit
positive rational \(\varepsilon\) of polynomial bit length for
which the game is nondegenerate and the inequalities and determinant
signs for every original equilibrium support are preserved.
The proof also excludes additional supports, giving a bijection
between equilibria that preserves supports and indices.

The row and column supports of an equilibrium of the perturbed game are equal
to a set \(K\) containing \(\star\).  We set \(I=K\setminus\{\star\}\) and
recover the original LCP solution exactly by solving \(M_{I,I}x_I=b_I\) and
setting \(x_{[m]\setminus I}=0\).
If the equilibrium has index \(+1\), the corresponding fixed point also has
index \(+1\), and the decoder recovers a sink of the initial \(\SoL\)
instance in polynomial time.

\subsection{Further related work}
\paragraph{Total search and \PPADS.} \cite{MegiddoPapadimitriou1991} initiated the
systematic study of total NP search problems, and
\cite{Papadimitriou1994} introduced the polynomial parity-argument classes.
The class now called \PPADS originally appeared as the positive, or sink-only,
version of the directed parity principle; the terminology \PPADS was adopted
in the subsequent study of the relative complexity of NP search problems by \cite{BeameEtAl1998}. Recent
work has clarified the structure around the class. 
\cite{HollenderGoldberg2018} studied multiple-source and multiple-sink variants
of \textsc{End-of-Line}; in particular, insisting on at least one sink retains
\PPADS-completeness.  \cite{GoosEtAl2024} proved the companion identity
$\SOPL=\PLS\cap\PPADS$, alongside their collapse result for
$\EOPL=\PLS\cap\PPAD$. These results further emphasize that selecting a sink,
rather than an arbitrary endpoint, is a genuine computational requirement.

Geometric formulations of directed parity go back to Sperner's lemma and its
oriented variants. Positive Sperner is a canonical complete problem for
\PPADS (see  \cite{Papadimitriou1994,BeameEtAl1998}). \cite{FriedlEtAl2006} constructed locally two-dimensional Sperner
problems complete for \PPAD, \PPADS, and \PPA on specially constructed
surfaces; their \PPADS representative asks for a fully labeled triangle of
one specified orientation. Their sign convention is
the simplicial analogue of the source/sink distinction used here. The crucial
representational difference is that their exponentially large labeled surface
is supplied by a Turing machine, whereas our input consists directly of two
explicit payoff matrices.

\paragraph{Equilibrium computation and the index.} \cite{LemkeHowson1964} gave the classical complementary
pivoting algorithm for bimatrix games. \cite{Shapley1974} introduced
the equilibrium index and proved that the endpoints of each Lemke--Howson path
have opposite signs.  Harsanyi's oddness theorem (\cite{Harsanyi1973}) and Jansen's analysis of
regular bimatrix equilibria (\cite{Jansen1981}) place this sign in the broader stability theory of finite games.  \cite{Todd1976} developed the determinant orientation for complementary pivots.  \cite{SavaniVonStengel2006} showed
that Lemke--Howson paths can be exponentially long. 
\cite{VeghVonStengel2015} later placed the sign argument in the general
framework of oriented pivoting systems and oriented Euler complexes. Our
membership proof can be viewed as extracting the sink-only computational
content of this classical orientation theory.  V\'egh and von Stengel
explicitly observed that the sink-only class $\PPADS$ corresponds, in an
oriented pivoting system, to asking for a second completely labeled state of
opposite sign \citep{VeghVonStengel2015}.  Thus the connection between sign and
sink-only parity is classical; our contribution is completeness for an
explicit equilibrium problem and a reduction under which \emph{every}
positive index equilibrium decodes a sink.  More recently,  \cite{IckstadtTheobaldVonStengel2025} used the same
orientation to obtain combinatorial bounds on the arrangement and number of
equilibria in nondegenerate bimatrix games.

The positive sign is also an established equilibrium selection criterion,
rather than merely an algorithmic label.  For nondegenerate bimatrix games, a
positive index equilibrium admits a strategic characterization in terms of
becoming the unique equilibrium after suitable one-player strategy extensions
\citep{VonSchemdeVonStengel2008}.  More generally, an isolated equilibrium is
sustainable exactly when it has index $+1$ \citep{GovindanLarakiPahl2023}.
Classical tracing and complementary-pivoting procedures generically select
positive index equilibria \citep{Balthasar2010}.  Computing an equilibrium
selected by a specified homotopy or Lemke--Howson path is nevertheless a
different search task: it is PSPACE-complete
\citep{GoldbergPapadimitriouSavani2013}, whereas our relation accepts every
positive index equilibrium.

The computational complexity of finding an unrestricted Nash equilibrium was
settled through a celebrated sequence of results. \cite{DaskalakisGoldbergPapadimitriou2009} proved
\PPAD-completeness for sufficiently accurate approximate equilibria in
three-player games, and  \cite{ChenDengTeng2009}
established exact \PPAD-completeness for bimatrix games. Their
reductions permit one to decode a valid solution of the underlying
\textsc{End-of-Line} instance from an equilibrium,
but they do not establish an all-equilibria bijection that preserves local
index. Our result is orthogonal: the game still has only two players, but the
allowed output is restricted according to its local orientation.

\paragraph{Linear fixed points, complementarity, and imitation games.} \cite{EtessamiYannakakis2010} introduced \FIXP and
showed that its piecewise-linear fragment \LinFIXP coincides with \PPAD. \cite{Mehta2018} strengthened this connection and developed direct
reductions from low-dimensional \LinFIXP to LCPs and bimatrix games; for
background on complementarity we refer to \cite{CottlePangStone1992}.  These
reductions led, among other results, to \PPAD-hardness for constant-rank
two-player games. \cite{McLennanTourky2010} showed that
symmetric equilibrium computation can be represented by imitation games of
the form $(C,I)$. We use these constructions for a different purpose. The
strictness of our circuit permits us to identify one active principal minor,
and the determinant lemma and the bordered support formula show that its sign
is preserved through every stage. We then add a quantitative general-position
perturbation because the unperturbed correspondence alone need not produce a
globally nondegenerate game. These index-preservation and nondegeneracy steps
are not consequences of the earlier \PPAD-hardness reductions and constitute
the algebraic core of our hardness proof.
\section{Preliminaries}\label{sec:preliminaries}

\paragraph{Bimatrix games.}
A bimatrix game is specified by two rational matrices
$A,B\in\mathbb{Q}^{m\times n}$.  The row player selects a mixed
strategy $x\in\Delta_m$ and the column player selects a mixed strategy
$y\in\Delta_n$, where
\[
    \Delta_k:=\{z\in\mathbb{R}_{\geq 0}^k:\one^{\transpose} z=1\}.
\]
Their expected utilities are $x^{\transpose} Ay$ and $x^{\transpose} By$, respectively.
For $x\in\Delta_m$ and $y\in\Delta_n$, let
\[
 \BR_2(x):=\operatorname*{arg\,max}_{j\in[n]}x^{\transpose} Be_j,
 \qquad
 \BR_1(y):=\operatorname*{arg\,max}_{i\in[m]}e_i^{\transpose} Ay.
\]

\begin{definition}[Nash equilibrium]\label{def:ne}
 A pair $(x,y)\in\Delta_m\times\Delta_n$ is a Nash equilibrium if
 \[
   \supp(x)\subseteq \BR_1(y)
   \qquad\text{and}\qquad
   \supp(y)\subseteq \BR_2(x).
 \]
 The equilibrium is \emph{exact} when the rational vectors $x$ and $y$
 are returned in binary, rather than through an approximation oracle.
\end{definition}

Adding a constant to every entry of one player's payoff matrix changes
neither best responses nor Nash equilibria.  We shall consequently assume,
whenever we invoke the Lemke--Howson construction, that all entries of $A$
and $B$ are strictly positive.

\begin{definition}[Nondegeneracy]\label{def:nondegenerate}
 A bimatrix game $(A,B)$ is \emph{nondegenerate} if, for every
 $x\in\Delta_m$ and $y\in\Delta_n$,
 \[
     |\BR_2(x)|\leq |\supp(x)|
     \qquad\text{and}\qquad
     |\BR_1(y)|\leq |\supp(y)|.
 \]
\end{definition}

At every equilibrium of a nondegenerate game, the two supports have the
same cardinality.  Indeed, the best response conditions and
Definition~\ref{def:nondegenerate} give
\[
 |\supp(x)|\leq |\BR_1(y)|\leq |\supp(y)|
 \leq |\BR_2(x)|\leq |\supp(x)|.
\]
Thus equality holds throughout.  In particular, if
$I:=\supp(x)$ and $J:=\supp(y)$, then $|I|=|J|$.

\paragraph{The index of an equilibrium.}
For a matrix $C\in\mathbb{R}^{k\times k}$, define its bordered
determinant by
\begin{equation}\label{eq:bordered-determinant}
 \Delta(C):=
 \det\begin{pmatrix}
      C&-\one\\
      \one^{\transpose}&0
 \end{pmatrix}.
\end{equation}
Notice that
\begin{equation}\label{eq:shift-invariance}
       \Delta(C+c\one\one^{\transpose})=\Delta(C)
       \qquad\text{for every }c\in\mathbb{R};
\end{equation}
indeed, start from the bordered matrix defining
$\Delta(C+c\one\one^{\transpose})$ and add $c$ times its last column to
each of its first $k$ columns.  Since the top block of the last column is
$-\one$, this removes $c\one\one^{\transpose}$ and recovers the bordered
matrix defining $\Delta(C)$.

\begin{definition}[Shapley index]\label{def:index}
 Let $(x,y)$ be an equilibrium of a nondegenerate bimatrix game, and let
 $I=\supp(x)$, $J=\supp(y)$, and $|I|=|J|=k$.  Its \emph{index} is
 \begin{equation}\label{eq:index}
  \ind_{A,B}(x,y)
  :=(-1)^{k+1}\sign\!\left(
       \Delta(A_{I,J})\Delta\bigl((B_{I,J})^{\transpose}\bigr)
     \right)\in\{-1,+1\}.
 \end{equation}
\end{definition}

\noindent Here and throughout, restriction precedes transposition: thus
\((B_{I,J})^{\transpose}\) is the transpose of the \(I\) by \(J\) support
submatrix, equivalently \((B^{\transpose})_{J,I}\).  In particular, the second
factor in~\eqref{eq:index} is
\(\Delta\bigl((B_{I,J})^{\transpose}\bigr)\).

The two bordered determinants in~\eqref{eq:index} are nonzero in a
nondegenerate game.  After shifting the payoffs to be strictly positive,
nondegeneracy implies that the displayed labeled inequality descriptions of
the two best response polytopes are nondegenerate, i.e., at every vertex of the best response polytope in \(\mathbb{R}^m\), exactly
\(m\) displayed inequalities bind, and at every vertex of the best response
polytope in \(\mathbb{R}^n\), exactly \(n\) displayed inequalities bind.  In
each case, the outward normals of the binding inequalities are linearly
independent.  Equivalently, no point of the first polytope carries more than
\(m\) labels, and no point of the second carries more than \(n\) labels.  In particular, the underlying geometric
polytopes are simple and no redundant displayed inequality is binding.  At an
equilibrium, nondegeneracy also implies that the two support systems in
\eqref{eq:support-systems} are nonsingular; see
\citet[Theorem~14]{vonStengel2021}.  Their coefficient matrices are exactly
the two bordered matrices whose determinants occur in~\eqref{eq:index}.
Definition~\ref{def:index} is independent of the
ordering of $I$ and $J$, and~\eqref{eq:shift-invariance} shows that it is
unchanged when the payoffs are shifted.  When $A$ and $B$ are strictly
positive, nondegeneracy also implies full rank of $A_{I,J}$ and
$(B_{I,J})^{\transpose}$; thus the inverses below exist, and the definition agrees
with Shapley's formula
\[
 \ind_{A,B}(x,y)
   =-\sign\det\begin{pmatrix}
       0&A_{I,J}\\ (B_{I,J})^{\transpose}&0
     \end{pmatrix}
   =(-1)^{k+1}\sign\bigl(\det A_{I,J}\det B_{I,J}\bigr).
\]
For example, every pure equilibrium has index $+1$.  The equivalence follows
from the Schur-complement identity
\[
 \Delta(A_{I,J})
 =\det(A_{I,J})\one^{\transpose} A_{I,J}^{-1}\one
\]
and its analogue for $B$.  More explicitly, the factor
$\one^{\transpose} A_{I,J}^{-1}\one$ is the reciprocal of the row player's \textit{positive}
equilibrium payoff, and
$\one^{\transpose}\bigl((B_{I,J})^{\transpose}\bigr)^{-1}\one$ is the corresponding reciprocal for
the column player.  We refer to
Shapley~\citep{Shapley1974} and von Stengel~\citep{vonStengel2021} for
equivalent definitions and further properties of the index.

For later use, we call an equilibrium $(x,y)$ \emph{regular} (equivalently,
\emph{locally nondegenerate} in a bimatrix game) if $|I|=|J|$, both bordered matrices in
Definition~\ref{def:index} are nonsingular, and no unused pure strategy is a
best response.  The index~\eqref{eq:index} is then well-defined even if the
game is degenerate elsewhere.  Every equilibrium of a nondegenerate game is
regular.

\paragraph{Exact representations.}
If $I,J$ are the
actual supports of an equilibrium and have common size $k$, that equilibrium
is recovered from the two systems
\begin{equation}\label{eq:support-systems}
 \begin{pmatrix}A_{I,J}&-\one\\\one^{\transpose}&0\end{pmatrix}
 \begin{pmatrix}y_J\\u\end{pmatrix}
   =\begin{pmatrix}\zero\\1\end{pmatrix},
 \qquad
 \begin{pmatrix}(B_{I,J})^{\transpose}&-\one\\\one^{\transpose}&0\end{pmatrix}
 \begin{pmatrix}x_I\\v\end{pmatrix}
   =\begin{pmatrix}\zero\\1\end{pmatrix}.
\end{equation}
In a nondegenerate game these systems are nonsingular.  For arbitrary
equal-size sets $I,J$, the same equations define only a candidate but it is possible that a
coefficient matrix may be singular, or the solution may violate positivity
or an off-support best response inequality.  Cramer's rule and the standard
determinant bound imply that every coordinate of an actual equilibrium has
encoding length polynomial in the encoding length of $(A,B)$.  Conversely,
the Nash inequalities and the determinant signs in~\eqref{eq:index} can all
be verified using exact rational arithmetic in polynomial time.

\paragraph{Linear fixed point maps and their index.}
A rational Linear-$\FIXP$ map $\Phi:[0,1]^d\to[0,1]^d$ is computed by a
rational Linear-$\FIXP$ circuit, namely a straight-line program using
rational affine operations together with $\max$.  We say that the circuit is
\emph{strict at $p$} if the two inputs of every max gate have different
values when the circuit is evaluated at $p$, i.e., no max gate has a tie.
At such a point, the branch selected by each max gate remains unchanged in a
neighborhood of $p$.  Replacing each max gate by this selected affine input
gives the \emph{active affine piece} at $p$, an affine formula that agrees
with $\Phi$ throughout that neighborhood. An interior fixed
point $p\in(0,1)^d$ is \emph{regular} if
\[
    \det\!\bigl(I_d-D\Phi(p)\bigr)\ne0.
\]
Its fixed point index is
\begin{equation}\label{eq:fixed-point-index}
    \ind_\Phi(p):=\sign\det\!\bigl(I_d-D\Phi(p)\bigr)
    \in\{-1,+1\}.
\end{equation}
Here $I_d$ is the $d\times d$ identity matrix.
Sections~\ref{sec:fixp-to-lcp} and~\ref{sec:lcp-to-game} prove that this
fixed point index equals the index of the corresponding Nash equilibrium.
Because $I_d-D\Phi(p)$ is a real nonsingular matrix, $
    \ind_\Phi(p)=(-1)^\nu,
$ where $\nu$ is the number of its negative real eigenvalues, counted with
algebraic multiplicity.  Nonreal eigenvalues occur in conjugate pairs, whose
products are positive, and hence do not alter this parity.

\paragraph{The class \PPADS.}
The canonical problem for \PPADS\ is \textsc{Sink-of-Line}.  We use the
following equivalent normalized version.

\begin{definition}[Normalized \textsc{Sink-of-Line}]\label{def:sink}
The input consists of Boolean circuits
$S,P:\{0,1\}^n\to\{0,1\}^n$ satisfying
\[
    P(0^n)=0^n\neq S(0^n)
    \qquad\text{and}\qquad
    P(S(0^n))=0^n.
\]
The goal is to find a string $z\in\{0,1\}^n$ such that
\[
    P(S(z))\neq z.
\]
\end{definition}

There is a directed edge $z\to z'$ precisely when $S(z)=z'$ and
$P(z')=z$; self-loops are allowed.  The additional condition
$P(S(0^n))=0^n$ ensures that the distinguished vertex $0^n$ has a valid
outgoing edge.  The usual formulation does not impose this condition.  If
it fails, then $0^n$ is already a valid solution; the reduction can therefore
output a fixed normalized instance and map every solution of that instance
back to $0^n$.  Thus the normalized and usual formulations are
polynomial-time equivalent.

The class \PPADS\ consists of the search problems that admit a
polynomial-time reduction to \textsc{Sink-of-Line}
\citep{Papadimitriou1994,HollenderGoldberg2018}.  Since \PIN\ is a promise
problem, its membership reduction is required to be correct on nondegenerate
input games, whereas its hardness reduction must produce a nondegenerate
game for every \textsc{Sink-of-Line} instance.
\section{Membership in \texorpdfstring{\PPADS}{PPADS}}\label{sec:membership}

Membership is a signed version of the usual Lemke--Howson argument.  On a
nondegenerate game, Shapley's orientation makes the origin a source and makes precisely the equilibria of index $+1$
sinks. Consequently,
the sink promised by the canonical $\PPADS$ problem decodes to the required
equilibrium.

\subsection{The oriented Lemke--Howson graph}\label{sec:oriented-lh}

Shift each payoff matrix by a rational constant so that $A>0$ and $B>0$
entrywise.  This changes neither best responses nor the index, see
\eqref{eq:shift-invariance}.  Consider the best response polytopes
\begin{equation}\label{eq:br-polytopes}
 \mathcal P:=\{x\in\mathbb{R}^m:x\geq0,\ B^{\transpose} x\leq\one\},
 \qquad
 \mathcal Q:=\{y\in\mathbb{R}^n:y\geq0,\ Ay\leq\one\}.
\end{equation}
They are bounded because the payoff matrices are strictly positive.  We label the inequalities defining $\mathcal P$ and $\mathcal Q$ as follows. Give
the inequality $x_i\geq0$ of $\mathcal P$ and the 
inequality $(Ay)_i\leq1$ of $\mathcal Q$ label $i\in[m]$.  Give the
 inequality $(B^{\transpose} x)_j\leq1$ of $\mathcal P$ and the 
inequality $y_j\geq0$ of $\mathcal Q$ label $m+j$, for $j\in[n]$.  Thus,
with $d=m+n$ and
\[
 C:=\begin{pmatrix}0&A\\B^{\transpose}&0\end{pmatrix},
\]
the product $\mathcal R:=\mathcal P\times\mathcal Q$ can be written as
\[
\mathcal R=\{z\in\mathbb{R}^d:-z\leq0,\ Cz\leq\one\},
\]
with two labeled inequalities per label in $[d]$.

A vertex is \emph{completely labeled} if its binding inequalities
contain all labels in $[d]$.  The origin $(0,0)$ is completely labeled and is called the
\emph{artificial equilibrium}.  Every other completely labeled point
$z=(x,y)$ has both $x\neq0$ and $y\neq0$, and it gives the Nash equilibrium
\begin{equation}\label{eq:normalize-cl}
       \overline x=\frac{x}{\one^{\transpose} x},
       \qquad
       \overline y=\frac{y}{\one^{\transpose} y}.
\end{equation}
Conversely, if a Nash equilibrium $(\overline x,\overline y)$ has row
payoff $u>0$ and column payoff $v>0$, then
$(\overline x/v,\overline y/u)$ is a completely labeled point.  By nondegeneracy, these labeled descriptions of
$\mathcal P$ and $\mathcal Q$ are nondegenerate, namely, at every vertex exactly as
many inequalities bind as the ambient dimension, and their outward
normals are linearly independent \cite[Theorem~14]{vonStengel2021}.  Hence
the underlying polytopes are simple, no redundant inequality is
binding, and the product description of $\mathcal R$ has exactly $d$
linearly independent binding inequalities at every vertex.

We fix one label $\ell$ once and for all (for example, $\ell=1$).  The
\emph{missing-$\ell$ graph} consists of the
vertices and one dimensional faces of $\mathcal R$ whose binding labels
contain $[d]\setminus\{\ell\}$.  Every component is a path or a cycle.
Indeed, a completely labeled vertex is an endpoint.  Every other vertex in
the graph is missing $\ell$ and has one duplicated label. As a result, relaxing either
of the two binding inequalities carrying that duplicated label
gives its two incident edges.

An
edge is a line segment of the $d$ dimensional
polytope $\mathcal R$.  In the nondegenerate  presentation, the
relative interior of an edge has exactly $d-1$ linearly independent binding
 inequalities.  A missing-$\ell$ edge must contain all $d-1$ labels
in $[d]\setminus\{\ell\}$.  It therefore lies in \emph{exactly one} binding
 inequality of each such label, i.e., there is no room for a duplicated
label along the interior of the edge.  A duplicate appears only at an
internal vertex, where one additional  inequality becomes binding.
Figure~\ref{fig:lh-edge-orientation} illustrates the endpoint incidence
convention used to orient each such edge.

\begin{figure}[t]
\centering
\begin{tikzpicture}[
  >=Latex,
  every node/.style={font=\small},
  edge/.style={very thick,blue!65!black,-{Latex[length=2mm]}},
  normal/.style={thick,BrickRed,-{Latex[length=2mm]}},
  facet/.style={thick,black!65}
]
  \coordinate (u) at (1.2,0);
  \coordinate (v) at (6.3,0);
  \draw[very thick] (u)--(v);
  \draw[edge] (2.25,0)--(4.85,0)
       node[midway,above=3pt] {$t_E$};
  \fill (u) circle (2pt) node[below=5pt] {$u$};
  \fill (v) circle (2pt) node[below=5pt] {$v$};
  \draw[normal] (u)--++(-.75,.75)
       node[left] {$N_u(E)$};
  \draw[normal] (v)--++(.75,.75)
       node[right] {$N_v(E)$};
  \node[align=center] at (1.2,-1.55)
       {$N_u(E)^{\transpose} t_E<0$\\outgoing incidence};
  \node[align=center] at (6.3,-1.55)
       {$N_v(E)^{\transpose} t_E>0$\\incoming incidence};
  \node[align=center,text width=8.2cm] at (3.75,1.65)
       {The $d-1$ binding inequalities containing $E$ determine the tangent
        direction.
        The additional binding inequality at each endpoint determines whether
        that tangent points into or out of the segment.};
\end{tikzpicture}
\caption{Orientation of a missing-label edge.  The planar drawing suppresses
the $d-1$  inequalities containing \(E\); their ordered outward normals define the
generalized cross product $t_E$.}
\label{fig:lh-edge-orientation}
\end{figure}

We use the following signed form of the Lemke--Howson theorem.  The statement
is specialized to the product best response polytope, so that its endpoint
sign has the game theoretic meaning used below.

\begin{lemmabox}
\begin{lemma}[Oriented Lemke--Howson paths]\label{lem:oriented-lh}
 Let $(A,B)$ be a nondegenerate rational bimatrix game with strictly positive
 payoffs, and let $\mathcal R=\mathcal P\times\mathcal Q$ be its labeled
 product best response polytope defined above.  For every missing label
 $\ell$, the missing-$\ell$ graph has a canonical orientation, computable
 locally from a vertex basis, with the following properties.
 \begin{enumerate}
  \item Every internal vertex has one incoming and one outgoing edge.
  \item The artificial equilibrium is a source and has sign $-1$.
  \item A non-artificial endpoint is a source when its Shapley index is
        $-1$ and a sink when its Shapley index is $+1$.
 \end{enumerate}
 The predecessor and successor bases can be found in polynomial time.
\end{lemma}
\end{lemmabox}

\begin{proof}
 We first define the orientation, including the convention at an internal
 vertex.  If $E$ is an edge of the missing-$\ell$ graph, then exactly one
  inequality of each label $r\neq\ell$ binds throughout $E$.  For an inequality
 $a^{\transpose} z\leq b$, its \emph{outward normal} is $a$ and it points from the
 feasible halfspace toward the infeasible side.  Write the outward normal of
 the label-$r$ binding inequality as $N_r(E)$, and order these normals by
 increasing label.
 We define $c_E$ as the generalized cross product of these
 $d-1$ normals. Its coordinates are their signed $(d-1)\times(d-1)$ minors.
 Equivalently, it is the unique rational vector defined by the identity
 \begin{equation}\label{eq:edge-cofactor}
 q^{\transpose} c_E=
 \det\bigl[N_1(E)\ \cdots\ N_{\ell-1}(E)\ q\
               N_{\ell+1}(E)\ \cdots\ N_d(E)\bigr]
 \qquad(q\in\mathbb R^d),
 \end{equation}
 and orient $E$ in the direction
 \begin{equation}\label{eq:edge-tangent}
                     t_E:=(-1)^{d+1}c_E.
 \end{equation}
 The vector $c_E$ is nonzero because the $d-1$ normals containing $E$ are
 independent.  Moreover, $N_r(E)^{\transpose} c_E=0$ for every $r\neq\ell$, since
 the determinant in~\eqref{eq:edge-cofactor} then has two equal columns.
 Thus $t_E$ is a tangent vector to the affine line containing $E$, i.e., $t_E$ is parallel to that line.  There are two possible parallel
 directions, and~\eqref{eq:edge-tangent} chooses one canonically from the
 increasing order of the labels.

 Let $v$ be an endpoint of the line segment $E$ in the polytope, and let
 $N_v(E)$ be the unique additional outward normal that is binding at $v$ but
 not along $E$.
 Define the incidence sign
 \begin{equation}\label{eq:incidence-sign}
             \iota(v,E):=\sign\bigl(N_v(E)^{\transpose} t_E\bigr).
 \end{equation}
 Because of nondegeneracy, the $d$ normals binding at $v$ are
 linearly independent, hence this scalar is nonzero.  Since all inequalities
 use outward normals, $\iota(v,E)=-1$ precisely when $t_E$ points from $v$ into $E$ and
 this is an outgoing incidence.  The sign $+1$ is an incoming incidence.

 We now verify the pivot rule at an internal vertex.  Such a vertex is
 missing label $\ell$ and has a duplicated label $k\neq\ell$.  Order the two
  inequalities of label $k$ by putting the nonnegativity inequality
 first and the payoff inequality second, and denote their outward normals by
 $U$ and $V$.  The edge $E_U$ is obtained by relaxing the payoff inequality so it retains $U$ among its
 $d-1$ binding normals and drops $V$.  The other edge $E_V$ relaxes the
 nonnegativity inequality, retaining $V$ and dropping $U$.  Thus $V$ is the
 additional endpoint normal for $E_U$, while $U$ is the additional endpoint
 normal for $E_V$.  By~\eqref{eq:edge-cofactor},
 \begin{align}
 V^{\transpose} c_{E_U}
   &=\det\bigl[\ldots,V\text{ in position }\ell,
                    \ldots,U\text{ in position }k,\ldots\bigr],\label{eq:pivot-one}\\
 U^{\transpose} c_{E_V}
   &=\det\bigl[\ldots,U\text{ in position }\ell,
                    \ldots,V\text{ in position }k,\ldots\bigr].\label{eq:pivot-two}
 \end{align}
 All other columns in these determinants are identical and in the same label
 order.  Swapping columns $\ell$ and $k$ therefore gives the determinant
 pivot identity
 \begin{equation}\label{eq:incidence-pivot-identity}
                 V^{\transpose} t_{E_U}=-U^{\transpose} t_{E_V}.
 \end{equation}
 Both sides are nonzero because the $d$ active normals at $v$ are
 independent.  Hence $\iota(v,E_U)=-\iota(v,E_V)$: exactly one incidence is
 incoming and the other is outgoing. At the two ends of any one geometric
 edge, the same tangent points into the segment at one end and out of it at
 the other, so the local incidence rules consistently orient every path and
 cycle.

 At a completely labeled endpoint $v$, order the outward normals of its
 binding  inequalities by their labels and write them as
 $N_1,\ldots,N_d$.  Its
 unique incident edge drops $N_\ell$, so
 \eqref{eq:edge-cofactor}--\eqref{eq:incidence-sign} give
 \begin{equation}\label{eq:normal-orientation}
 \iota(v,E)=\sigma(v):=(-1)^{d+1}
                   \sign\det[N_1\ \cdots\ N_d].
 \end{equation}
At the origin, the binding inequalities are the \(d\) nonnegativity
inequalities \(-z_r\leq 0\), for \(r\in[d]\) and all payoff inequalities are
strict.  Hence \(N_r=-e_r\) and the ordered normal matrix is
\(N=-I_d\).  Therefore
\[
 \sigma(0,0)
 =(-1)^{d+1}\sign\det(-I_d)
 =(-1)^{d+1}(-1)^d
 =-1.
\]
More explicitly, after dropping the inequality with the missing label
\(\ell\), the orientation rule gives \(t_E=e_\ell\).  Since the outward
normal of the dropped inequality is \(-e_\ell\), we have
\((-e_\ell)^{\transpose}t_E=-1\).  Thus \(t_E\) points from the origin into
\(E\), so the incidence is outgoing and the artificial equilibrium is a
source. This is Shapley's determinant orientation of the Lemke--Howson
 paths \citep{Shapley1974,VeghVonStengel2015}.

 For completeness, we now derive the endpoint sign.  Consider a non-artificial endpoint with
 supports $I=\supp(x)$ and $J=\supp(y)$.  Every strategy in $I$ is a best
 response to $y$, so nondegeneracy gives $|I|\le |J|$; symmetrically,
 $|J|\le |I|$.  We denote their common size by \(k\).  In the ordered
 normal matrix $N=[N_1\ \cdots\ N_d]$, the column of label $i\in[m]$ is
 \[
  N_i=
  \begin{cases}
    -e_{x_i},&i\notin I,\\
    \bigl(0,(A_{i,*})^{\transpose}\bigr),&i\in I,
  \end{cases}
 \]
 while the column of label $m+j$ is
 \[
  N_{m+j}=
  \begin{cases}
    -e_{y_j},&j\notin J,\\
    (B_{*,j},0),&j\in J.
  \end{cases}
 \]
 Here the rows are ordered first by the $x$-coordinates and then by the
 $y$-coordinates.  Let $\Pi$ be the permutation matrix sending this row
 order to $(x_{I^c},y_{J^c},x_I,y_J)$, and apply the same permutation to
 the label columns.  The following matrix is $\Pi N\Pi^{\transpose}$, so
 its determinant equals $\det N$ because $(\det\Pi)^2=1$.  It is block
 upper triangular:
 \begin{equation}\label{eq:normal-matrix-support-reduction}
 \det N\ =\
 \det\begin{pmatrix}
  -I_{m-k}&0&0&B_{I^c,J}\\
  0&-I_{n-k}&\bigl(A_{I,J^c}\bigr)^{\transpose}&0\\
  0&0&0&B_{I,J}\\
  0&0&\bigl(A_{I,J}\bigr)^{\transpose}&0
 \end{pmatrix}.
 \end{equation}
 Thus
 \[
  \det N=(-1)^{(m-k)+(n-k)}
  \det\!\begin{pmatrix}0&B_{I,J}\\
              (A_{I,J})^{\transpose}&0\end{pmatrix}
  =(-1)^d
  \det\!\begin{pmatrix}0&B_{I,J}\\
              (A_{I,J})^{\transpose}&0\end{pmatrix},
 \]
 because $2k$ is even.  Consequently,
 \begin{equation}
 \begin{aligned}
 \sigma(v)
  &=(-1)^{d+1}\sign\det N\\
  &=-\sign\det\begin{pmatrix}
          0&B_{I,J}\\ \bigl(A_{I,J}\bigr)^{\transpose}&0
       \end{pmatrix}\\
  &=(-1)^{k+1}\sign\bigl(\det A_{I,J}\det B_{I,J}\bigr).
 \end{aligned}
 \tag{16}
 \end{equation}
 The last equality swaps the two
 blocks of $k$ columns, contributing $(-1)^{k^2}=(-1)^k$.
 Since the endpoint basis is nonsingular and the payoffs are positive, the
 two support matrices are nonsingular.
 Moreover, $A_{I,J}^{-1}\one$ and
 $\bigl((B_{I,J})^{\transpose}\bigr)^{-1}\one$ are the positive, unnormalized equilibrium
 strategies. The last term equals the bordered-determinant
 definition~\eqref{eq:index} of the Shapley index.  Thus negative endpoints
 are sources and positive endpoints are sinks.

 It remains to show that, from the basis of a vertex, we can compute its
neighboring bases and the orientations of its incident edges in polynomial
time.  To this end, list the $2d$ inequalities in a fixed order as
 $a_h^{\transpose} z\leq b_h$.  A basis is the set of its $d$ binding
 inequalities.  Given such a set, Gaussian elimination tests nonsingularity
 and computes the vertex $v$.  The neighboring basis is computed by the
 exact minimum positive ratio test.  Namely, delete the inequality indexed by
 $r$, compute a nonzero direction $q$ orthogonal to the $d-1$ retained
 normals, and choose its sign so that $a_r^{\transpose} q<0$, i.e., the dropped
 inequality becomes slack.  For every inequality that is not binding at \(v\) indexed by $h$ for which
\(a_h^{\transpose}q>0\), we compute
 \begin{equation}
       \theta_h=\frac{b_h-a_h^{\transpose} v}{a_h^{\transpose} q}.
 \tag{17}
 \end{equation}
 The entering inequality is the unique minimizer among the positive ratios
 $\theta_h$, where uniqueness follows from nondegeneracy. The dropped inequality has negative denominator and hence is
 not reconsidered. Moreover, the pivot direction~\eqref{eq:edge-tangent} and every
 incidence sign~\eqref{eq:incidence-sign} are given by the signed minors defining
 $c_E$ and by determinant signs of rational matrices.  Hadamard's determinant
 bound, together with Cramer's rule, shows that all intermediate integers and
 rationals have polynomial encoding length, see for example,
 \citet[Chapter~3]{Schrijver1986}.  This proves local polynomial computability.
\end{proof}

\begin{theorembox}
\begin{theorem}\label{thm:membership}
 \PIN belongs to $\PPADS$.
\end{theorem}
\end{theorembox}

\begin{proof}
 We reduce to the canonical \SoL instance of Definition~\ref{def:sink}.  Fix a
 missing label $\ell$ and retain the fixed ordering of the $2d$ 
 inequalities used above.  A graph vertex $v$ has a unique basis, because
 exactly $d$ inequalities bind there. We represent it by the
 characteristic vector $\chi(v)\in\{0,1\}^{2d}$ of those inequalities.  If
 $\chi_0$ is the characteristic vector of the artificial basis, use the
 canonical binary representation of $v$
 \begin{equation}
                        \operatorname{rep}(v):=\chi(v)\mathbin{\mathsf{xor}}\chi_0.
 \tag{18}
 \end{equation}
 Thus the artificial equilibrium is represented by $0^{2d}$.

 On a valid internal representation, let $P$ and $S$ return the oriented
 predecessor and successor computed by the ratio test above.  At a source
 endpoint, set $P(v)=v$ and let $S(v)$ be its unique outgoing neighbor.  At a
 sink endpoint, let $P(v)$ be its unique incoming neighbor and set $S(v)=v$.
 The self-predecessor convention at a non artificial source does not create an
 additional solution because if $u=S(v)$ is its outgoing neighbor, then $P(u)=v$, so
 $P(S(v))=v$.  Only the artificial source is used as the distinguished
 all-zero start vertex.
 Vertices on directed cycles use their oriented predecessor and successor.
 On every invalid string, set both circuits equal to the identity.  Validity
 of a string is checked by undoing the xor, testing that it selects exactly
 $d$ inequalities with independent normals, solving for their
 intersection, checking feasibility and the complete binding set, and
 checking membership in the missing-$\ell$ graph.  All these operations use
 exact rational arithmetic and polynomial-size Boolean circuits.

 Lemma~\ref{lem:oriented-lh} gives
 \begin{equation}
      P(0^{2d})=0^{2d}\neq S(0^{2d})
      \quad\text{and}\quad
      P(S(0^{2d}))=0^{2d}.
 \tag{19}
 \end{equation}
 Internal vertices, source endpoints, directed cycles, and invalid representations
 all satisfy $P(S(v))=v$.  Hence every solution of the resulting \SoL
 instance is a positive non-artificial endpoint.
 Lemma~\ref{lem:oriented-lh} and~\eqref{eq:normalize-cl} decode it to an exact
 Nash equilibrium of index $+1$.  Exactness has polynomial cost by the
 support systems~\eqref{eq:support-systems}.
\end{proof}

\section{An explicit signed fixed point field}
\label{sec:signed-fixed-point-unconditional}
\setcounter{equation}{19}

From a normalized \textsc{Sink-of-Line}, we construct a signed field in two dimensions. We use staircase paths together with the crossing switch and coloring of \citet{ChenDeng2009}.
The planar paths of the construction preserve the sources and sinks of the input graph.
We make use of affine interpolation of the colors which can be computed
exactly by a polynomial-size rational Linear-$\FIXP$ circuit. The zeros
of the field are regular and we will show that every gate is strict there.
The boundary coloring ensures that all fixed points are interior, as required
by the subsequent reductions.

\begin{definitionbox}
\begin{definition}[Certified signed circuit field]
\label{def:certified-signed-field}
For a \SoL\ instance \(\mathcal I\), a \emph{certified signed circuit field}
consists of a rational Linear-$\FIXP$ circuit computing
\(G\colon[0,1]^d\to\mathbb R^d\), a rational bound \(M_G\), and a decoder,
satisfying the following properties:
\begin{enumerate}[label=\textup{\textbf{C\arabic*}.}]
\item The zeros of \(G\) form disjoint finite sets \(Z_+,Z_-\) of rational
points in the interior of the cube. At each zero \(p\), there is a rational
nonsingular matrix \(A_p\) such that \(G(x)=A_p(x-p)\) in a neighborhood of
\(p\). The sign of \(\det(-A_p)\) is \(+1\) on \(Z_+\) and \(-1\) on
\(Z_-\).
\item For every boundary point \(x\) and every \(\eta>0\),
\([x+\eta G(x)]_0^1\ne x\).
\item At every zero, the two inputs of every max gate in the circuit are
unequal.
\item From the coordinates of any \(p\in Z_+\), the decoder returns a
solution of \(\mathcal I\) in polynomial time.
\item The dimension and circuit size are polynomially bounded, all rational
coefficients have polynomial bit length, and
\(\lvert G_i(x)\rvert\le M_G\) on the cube. The construction supplies a
positive rational lower bound on the distance of every zero from the boundary.
The zeros and their local matrices have polynomial bit length and are
computable from their endpoint descriptions in polynomial time.
\end{enumerate}
The endpoint descriptions specify the sets \(Z_+,Z_-\) implicitly. They
do not require enumeration of an exponentially large set of points.
\end{definition}
\end{definitionbox}

These are the properties used in Sections~\ref{sec:fixp-to-lcp}
and~\ref{sec:lcp-to-game}. We will establish the following result.

\begin{theorembox}
\begin{theorem}[Indexed fixed point map]
\label{thm:axis-unconditional-map}
Given a normalized \SoL\ instance
\(S,P\colon\{0,1\}^n\to\{0,1\}^n\), one can construct in polynomial time
a rational \LinFIXP\ circuit for
\(\Phi\colon[0,1]^2\to[0,1]^2\) with the following properties:
\begin{enumerate}
\item Its fixed points correspond bijectively to the sinks and noncanonical
sources of the graph \(H\) in~\eqref{eq:route-H-edges}.
\item Every fixed point is regular, lies in the interior of the square, and makes every max gate strict.
\item The fixed point index is \(+1\) at a sink and \(-1\) at a
noncanonical source.
\item Every positive index fixed point decodes in polynomial time to an
\(x\in\{0,1\}^n\) satisfying \(P(S(x))\ne x\).
\end{enumerate}
\end{theorem}
\end{theorembox}





\subsection{The input graph and its planar representation}
\label{sec: 4.1-predicates}
\label{sec:signed-route-geometry}

We work directly with the graph of the normalized \SoL\ instance on its
original vertices. We discard self-loops, route the remaining edges by
staircase paths, and apply the crossing switch of \citet{ChenDeng2009}.
At an intersection, two directed paths exchange their outgoing portions.
This changes which source is connected to which sink, but preserves the
sources and sinks. The coloring will depend only on nearby directed edges.

\paragraph{The graph used by the planar construction.}
The normalized input satisfies
\(P(0^n)=0^n\ne S(0^n)\) and \(P(S(0^n))=0^n\).
We define a graph \(H\) on the original vertices
\(\{0,1\}^n\) by keeping exactly the edges that do not create a self-loop:
\begin{equation}
 x\longrightarrow S(x)
 \quad\Longleftrightarrow\quad
 S(x)\ne x\ \text{ and }\ P(S(x))=x.
\label{eq:route-H-edges}
\tag{Graph}
\end{equation}
Thus every vertex has indegree and outdegree at most one. Throughout this
section, a \emph{source} means a vertex of indegree zero and outdegree one,
and a \emph{sink} means a vertex of indegree one and outdegree zero;
isolated vertices are neither.


Deleting self-loops preserves the sources and sinks, since sources have indegree zero and outdegree one, and sinks have indegree one and outdegree zero. To see this, suppose that $x\to x$ is a self-loop, so $S(x)=P(x)=x$. Any outgoing edge $x\to y$ must satisfy $y=S(x)=x$, and any incoming edge $y\to x$ must satisfy $y=P(x)=x$. Hence the self-loop is the only edge incident to $x$. Removing it leaves $x$ isolated, so $x$ is neither a source nor a sink before or after deletion. All other vertices are unaffected.
Moreover, the normalization guarantees that $0^n$ is a source of $H$. The conditions $S(0^n)\ne0^n$ and $P(S(0^n))=0^n$ ensure that $0^n\to S(0^n)$ is a non-loop edge and therefore belongs to $H$. On the other hand, an incoming edge $z\to0^n$ would require $z=P(0^n)=0^n$, which is impossible because $H$ contains no self-loops. Thus $0^n$ has indegree zero and outdegree one.

Every sink of \(H\) is a valid output of the normalized \SoL\ instance.
To see this, let \(u\) be a sink. If \(S(u)\ne u\), the absence of an
outgoing edge gives \(P(S(u))\ne u\). If \(S(u)=u\), then the incoming edge into \(u\) implies \(P(u)\ne u\), and hence again
\(P(S(u))=P(u)\ne u\). Moreover, following the unique outgoing edges from
the distinguished source \(0^n\) must reach a sink: a finite directed path
starting at a source cannot enter a directed cycle in a graph of indegree
at most one. Thus the graph \(H\) always contains at least one sink.

We identify each \(n\)-bit vertex with its binary integer
\(u\in\{0,\ldots,L-1\}\), where \(L=2^n\). The following functions and
flags give local access to \(H\):
\[
 \mathsf{suc}_H(u)=S(u),\qquad
 \mathsf{out}_H(u)=[S(u)\ne u]\wedge[P(S(u))=u],
\]
\[
 \mathsf{pred}_H(u)=P(u),\qquad
 \mathsf{in}_H(u)=[P(u)\ne u]\wedge[S(P(u))=u].
\]
A candidate predecessor or successor is used only when its accompanying
flag is one. These functions have Boolean circuits of polynomial size.

\begin{table}[t]
\centering
\small
\renewcommand{\arraystretch}{1.15}
\begin{tabular}{c|l}
Function / flag & Definition \\ \hline
\(\mathsf{suc}_H(u)\) & \(S(u)\) \\
\(\mathsf{out}_H(u)\) & \([S(u)\ne u]\wedge[P(S(u))=u]\) \\
\(\mathsf{pred}_H(u)\) & \(P(u)\) \\
\(\mathsf{in}_H(u)\) & \([P(u)\ne u]\wedge[S(P(u))=u]\)
\end{tabular}
\caption{Local access to graph \(H\). The notation
\([\mathcal P]\) is one when the statement \(\mathcal P\) holds and zero
otherwise.}
\label{tab:2d-graph-predicates}
\end{table}

\paragraph{An illustrative example.}
All four construction figures use the same instance with $n=3$ and
$L=8$. Its graph $H$ has exactly the edges
\[
                 0\longrightarrow3\longrightarrow1,
                 \qquad
                 4\longrightarrow2\longrightarrow5.
\]
Thus $0$ is the canonical source, $4$ is the other source, and $1$ and
$5$ are the sinks; vertices $6$ and $7$ are isolated. For instance, the
successor and predecessor functions, listed in vertex order $0,\ldots,7$,
can be chosen as
\[
 S=(3,1,5,1,2,5,6,7),\qquad
 P=(0,3,4,0,4,2,6,7).
\]
Figure~\ref{fig:section4-full-example} shows the full construction.
Figures~\ref{fig:enlarged-route-turn-switch} and~\ref{fig:two-isolated-edges}
show its local behavior at the canonical connection, a turn, a switch,
and the three remaining endpoints.

\paragraph{Staircase paths and intersections.}
We place vertex \(u\) at \(b_u=(3u,3u)\).
For each directed edge \(u\to v\), we draw the \textit{staircase} path \(P_{uv}\)
through the successive points
\begin{equation}
 (3u,3u)\longrightarrow(3v,3u)\longrightarrow(3v,3v).
\label{eq:2d-edge-path}
\tag{Path}
\end{equation}
Both segments have positive length because \(u\ne v\). We subdivide
them into directed unit grid edges. The horizontal segment leaves
\(b_u\) in direction \(\sign(v-u)e_1\), and the vertical segment enters
\(b_v\) in direction \(\sign(v-u)e_2\).

The horizontal segment of an edge leaving \(u\) lies on row \(3u\).
There is at most one such segment, since \(u\) has at most one successor.
Similarly, column \(3v\) contains at most one vertical segment, since
\(v\) has at most one predecessor. Thus distinct paths have no overlapping
segments. A horizontal segment meets the diagonal only at its initial
vertex, and a vertical segment meets it only at its terminal vertex.
Therefore, no route passes through an unrelated graph vertex.
At a bend \((3v,3u)\), the unique horizontal segment on its row and the
unique vertical segment on its column both belong to \(P_{uv}\), so no
other path meets that bend. At \(b_u\), an incoming path is vertical
and an outgoing path is horizontal. Every other intersection is thus
a transverse crossing in the interiors of one horizontal and one vertical
segment.

Let \(b\) be such a crossing. Write
\(T\in\{e_1,-e_1\}\) for the direction of the horizontal path
and \(U\in\{e_2,-e_2\}\) for the direction of the vertical path.
The four directed unit edges through \(b\) are
\(b-T\to b\to b+T\) and \(b-U\to b\to b+U\).
We delete these four edges and insert the switch mechanism used in \cite{ChenDeng2009}:
\begin{equation}
 b-T\longrightarrow b-T+U\longrightarrow b+U,
 \qquad
 b-U\longrightarrow b+T-U\longrightarrow b+T.
\label{eq:2d-switch}
\tag{Switch}
\end{equation}

It follows that the two new paths have no vertex in common. Each of the four
vertices \(b-T,b-U,b+T,b+U\) retains its incoming or outgoing
connection, and each new intermediate vertex has one incoming
and one outgoing edge. The switch therefore creates no source
or sink, although it may change which source is connected to
which sink.

Every crossing has coordinates in \(3\mathbb Z^2\), as do all bends
and original vertices. A crossing is distinct from each bend and original
vertex, so its \(\ell_\infty\)-distance from every such point, and from
every other crossing, is at least three. The switch changes edges only
within \(\ell_\infty\)-distance one of its center. These neighborhoods
are disjoint. Moreover, any other horizontal or vertical segment lies
on a different row or column, at distance at least three, so it cannot
meet an inserted detour. The switches consequently create no additional
intersections and leave every original vertex neighborhood unchanged.

We denote the resulting directed grid graph by \(\widetilde H\).
Every nonendpoint on an unchanged route has one incoming and one
outgoing edge, and the switch preserves these degrees at its four points and gives these degrees to its two new intermediate vertices. 
Hence \(\widetilde H\) has indegree and outdegree at most one. Its sources
and sinks are exactly the points \(b_u\) corresponding to sources and
sinks of \(H\), with their incident directions preserved. The coloring below
uses only vertices incident to an edge of \(\widetilde H\).

\paragraph{Local computation of the planar edges.}
We can determine whether a given directed unit grid edge \(e=(a,b)\)
belongs to \(\widetilde H\) without constructing the whole drawing.
First consider the unswitched paths. If \(e\) is horizontal, its common
height must equal \(3u\) for some integer \(0\le u<L\). We evaluate
\(\mathsf{out}_H(u)\) and, when this flag is one, we set
\(v=\mathsf{suc}_H(u)\). The edge belongs to the horizontal segment
precisely when both its horizontal coordinates lie between \(3u\) and
\(3v\), inclusive, and
\(b-a=\sign(v-u)e_1\).

If \(e\) is vertical, its common horizontal coordinate must equal
\(3v\) for some integer \(0\le v<L\). We evaluate
\(\mathsf{in}_H(v)\) and, when this flag is one, we set
\(u=\mathsf{pred}_H(v)\). The edge belongs to the vertical segment
precisely when both its heights lie between \(3u\) and \(3v\), inclusive,
and \(b-a=\sign(v-u)e_2\). Thus, edges failing these tests are absent.

To account for switches, we enumerate the constantly many points
\(c\in3\mathbb Z^2\) within \(\ell_\infty\)-distance one of an
endpoint of \(e\). Applying the preceding tests in both orientations
to each of the four incident unit edges decides whether \(c\) is a
crossing and determines its directions \(T,U\).
If so, we test whether \(e\) is deleted or inserted by
\eqref{eq:2d-switch}; otherwise its original membership is unchanged.
The disjoint switch neighborhoods make these decisions consistent.
This uses a constant number of predicate evaluations and arithmetic
on \(\mathcal O(n)\)-bit integers, and thus gives a Boolean circuit of
polynomial size for membership in \(\widetilde H\).

\paragraph{The coloring.}
We extend the coloring of \cite{ChenDeng2009} to a square grid, as we need to define our map $G$ on $[0,1]^2$.
We set $M=32L=2^{n+5}$ and define the integer grid $\{0,\ldots,M\}^2$. We enlarge the construction of
$\widetilde H$ before coloring the grid as follows.
Each vertex $v$ moves to $3v+(3,3)$, and each directed unit
edge becomes a directed path of three unit edges between the
new endpoints. Let $\mathcal E$ be the edges of this enlarged
drawing and let $\mathcal V$ be their endpoints, including the
two new vertices on each enlarged edge. In particular, the vertex
$b_u=(3u,3u)$ moves to $a_u=3b_u+(3,3)=(9u+3,9u+3)$ in the enlarged drawing.
The unswitched paths lie in $[0,3(L-1)]^2$. Each crossing is at
least three units from the ends of both segments through it, so its
detours also lie in this square. Every point of the enlarged routed
graph therefore has coordinates between $3$ and $9L-6$, inclusive.
In particular, the routes stay away from all four boundaries of the
$M\times M$ grid. Moreover, distinct nonincident unit grid edges of
$\widetilde H$ are at distance at least one.
So after enlargement, they are at distance at least three. 

We assign to points of $\mathcal V$ with color zero. To color a point
$p\notin\mathcal V$, look at the edges in $\mathcal E$ whose
$\ell_\infty$-distance from $p$ is at most one. If there are
no such edges, give $p$ color two. Otherwise, consider those
at the smallest Euclidean distance from $p$. We write
$\mathsf{left}(p)$ if one of these closest directed edges
$(a,b)$ satisfies $\det[b-a,p-a]>0$. Such a point receives
color one, and every other point off the paths receives
color two. 
Positive determinant means $p$ is to the left of the edge direction, and negative determinant means $p$ is to the right of the edge direction. 
At a turn, two edges may be equally
close. A positive determinant for either edge suffices, even
if the determinant for the other is zero. Note that only constantly many
edges are close enough to be considered, and their squared
distances from $p$ are integers. 
We make a small adjustment near the origin to join the
lower-left corner to the canonical source.

Figures~\ref{fig:section4-full-example} and~\ref{fig:enlarged-route-turn-switch}
illustrate these rules for the example discussed above. 
Their field arrows and direction shading use the affine interpolation, which will be discussed in Section~\ref{sec:discrete-grid}. Formally, we define the coloring rule as follows.

\begin{equation}
 c(i,j)=
 \begin{cases}
  0,&i=0\text{ and }j\leq3,\\
  1,&i=0\text{ and }j>3,\\
  0,&i=1\text{ and }j=3,\\
  1,&i=1\text{ and }j=4,\\
  2,&i=1\text{ and }j\notin\{3,4\},\\
  0,&i=2\text{ and }j=3,\\
  0,&(i,j)\in\mathcal V,\\
  1,&\mathsf{left}(i,j),\\
  2,&\text{otherwise}.
 \end{cases}
\label{eq:2d-complete-coloring}
\tag{Coloring}
\end{equation}

The cases are read from top to bottom, with the first applicable case determining the color.
Thus $(0,0),(0,1),(0,2),(0,3),(1,3),(2,3)$
are colored zero and lead to $a_0=(3,3)\in\mathcal V$.
Away from this adjustment, the paths have color zero, points to their left have color one, and the remaining points have color two. 

\begin{figure*}[p]
\centering
\input{figures/section4_overview.tex}
\caption{The full staircase construction for the running example
$0\to3\to1$ and $4\to2\to5$, with $L=8$ and $M=256$.
Diamonds mark the two switch centers $(30,21)$ and $(21,30)$.
The switches produce paths from $0$ to $5$ and from $4$ to $1$, and a
cycle containing $2$ and $3$, while preserving every source and sink.
The boundary connection reaches the canonical source $a_0=(3,3)$.
Black dots mark zero barycenters: sinks $1$ and $5$ have positive index,
and source $4$ has negative index. Arrows and background shading encode
the direction of the continuous field $G$; squares mark grid samples,
and open circles mark interpolated samples. The unused outer portions
of both axes are compressed; ticks give enlarged integer-grid
coordinates, which are divided by $M$ to obtain points in $[0,1]^2$.
Field directions are shown in the original coordinate axes.}
\label{fig:section4-full-example}
\end{figure*}

\begin{figure*}[t]
\centering
\input{figures/section4_turn_switch.tex}
\caption{Local views of the running example in
Figure~\ref{fig:section4-full-example}.
In (a), the canonical connection reaches $a_0=(3,3)$, and the route for
$0\to3$ turns at $(30,3)$; columns $8,\ldots,25$ are compressed.
The dashed green triangle at $a_0$ has colors $0,1,0$ and contains no zero.
In (b), the original routes for $0\to3$ and $2\to5$ meet at
$(30,21)$. The switch detours pass through
$(27,21),(27,24),(30,24)$ and
$(30,18),(33,18),(33,21)$; the center has color $2$.
Grid vertices are labeled $0$, $1$, or $2$ according to their color,
as indicated in the legend. Arrows at every displayed grid vertex and
at the interpolated sample points, marked by open circles, indicate
the local field direction; the background uses the same direction
palette. Neither diagram contains a zero.}
\label{fig:enlarged-route-turn-switch}
\end{figure*}

\begin{lemmabox}
\begin{lemma}[The coloring and the graph endpoints]
\label{lem:2d-discrete-coloring}
The coloring defined in \eqref{eq:2d-complete-coloring} can be computed by a Boolean circuit of size polynomial
in the input size. Among the triangles
\begin{equation}
 \{z,z+e_1,z+e_2\},\qquad
 \{z+e_1+e_2,z+e_1,z+e_2\},
 \qquad z\in\{0,\ldots,M-1\}^2,
\label{eq:2d-grid-triangles}
\end{equation}
the triangles with all three colors are precisely one triangle for
each source \(u\ne0\) and one triangle for each sink \(u\).
For a source, let $\sigma_u=\sign(S(u)-u)$ be its outgoing horizontal direction. For a sink, let $\tau_u=\sign(u-P(u))$ be its incoming vertical direction. Both signs belong to $\{-1,1\}$. The triangle
vertices, listed in color order \(0,1,2\), are
\begin{equation}
 \begin{aligned}
   &a_u,\quad a_u+(-\sigma_u,\sigma_u),\quad a_u+(-\sigma_u,0)
        &&\text{for a source }u\ne0,\\
   &a_u,\quad a_u+(-\tau_u,\tau_u),\quad a_u+(0,\tau_u)
        &&\text{for a sink }u.
 \end{aligned}
\label{eq:2d-endpoint-triangles}
\end{equation}
All these triangles are in the interior of the square.
\end{lemma}
\end{lemmabox}

\begin{proof}
The coloring can be computed locally as follows. To decide whether a grid point
belongs to $\mathcal V$, we test its incident unit edges. To decide
$\mathsf{left}(p)$, we test the constant number of directed segments
within distance one of $p$ and compare their squared distances from
$p$. Each edge test uses the graph predicates described above. Since
the coordinates have $\mathcal O(n)$ bits, these tests and the cases in
\eqref{eq:2d-complete-coloring} give a Boolean circuit of polynomial
size.

We next locate the possible fully colored triangles. The point
corresponding to a graph endpoint $u$ is $a_u=(9u+3,9u+3)$. Apart from
the initial connection at the boundary, every color-zero point lies
on an enlarged route. At a route vertex with both an incoming and an
outgoing edge, the nearby color-one points remain on the left of the
route and the color-two points on the right. 

The joins at the two nonendpoint graph vertices of the running example
are shown in Figure~\ref{fig:section4-vertices}. At $a_2=(21,21)$,
the incoming route for $4\to2$ travels downwards and the outgoing route
for $2\to5$ leaves to the right. At $a_3=(30,30)$, the incoming route
for $0\to3$ travels upwards and the outgoing route for $3\to1$ leaves
to the left. Both joins are ordinary turns of the color-zero corridor;
neither creates a fully colored triangle or a zero of the interpolated
field.

To check the incident unit triangles, translate the route vertex to
the origin. Their other two vertices are the consecutive pairs in the
cyclic list
\[
 e_1,\ e_2,\ -e_1+e_2,\ -e_1,\ -e_2,\ e_1-e_2.
\]
The nearest-edge rule assigns the same color to the two vertices of
each pair that contains no route vertex: both are on the directed left,
or both are on the remaining side. A pair containing a route vertex
already has color zero. This applies to a straight segment and to each
right-angle turn, so every incident triangle either has two color-zero
vertices or uses only zero and one other color. In particular, the same
check covers the turns made by \eqref{eq:2d-switch}. The enlargement
ensures that no nonincident portion of the drawing affects this check. This is the
geometric argument used in \citet{ChenDeng2009}. Hence a fully colored
triangle can occur only at a source or sink $a_u$, or along the
initial connection.

\begin{figure*}[t]
\centering
\input{figures/section4_endpoints.tex}
\caption{The three endpoint triangles containing zeros in the running example.
Subfigure (a) shows source $4$ at $(39,39)$, whose route leaves to the left;
subfigures (b) and (c) show sinks $1$ and $5$ at $(12,12)$ and $(48,48)$,
entered downwards and upwards, respectively. Green outlines mark the
fully colored triangles, and black dots mark their zero barycenters.
The source has negative index and both sinks have positive index.
Grid vertices are labeled $0$, $1$, or $2$ according to their color,
as indicated in the legend. Arrows at every displayed grid vertex
indicate the local field direction; subfigure (b) also includes interpolated
samples marked by open circles. The background uses the same direction
palette. All coordinates are enlarged integer-grid coordinates.}
\label{fig:two-isolated-edges}
\end{figure*}

\begin{figure*}[t]
\centering
\scalebox{0.8}{%
\input{figures/section4_vertices.tex}%
}\caption{The two nonendpoint graph vertices in the running example.
Diagram (a) shows the join at $a_2=(21,21)$ between the incoming route
for $4\to2$ and the outgoing route for $2\to5$; diagram (b) shows the
join at $a_3=(30,30)$ between the incoming route for $0\to3$ and the
outgoing route for $3\to1$. The plain corridor lines show the routing,
whereas arrows show the field at every displayed grid vertex. The grid labels and
continuous direction palette are the same as in
Figure~\ref{fig:enlarged-route-turn-switch}. Neither neighborhood
contains a fully colored triangle or a zero. Coordinates refer to
the enlarged integer grid.}
\label{fig:section4-vertices}
\end{figure*}

The normalization gives $S(0)>0$, so the first edge at
$a_0=(3,3)$ points to the right. Thus the initial connection has the
local configuration specified by \eqref{eq:2d-complete-coloring} and creates no fully colored triangle.
Its adjacent squares
below height $3$ use only colors zero and two. In each of the three
squares with lower-left corners $(0,3)$, $(1,3)$, and $(2,3)$, the
two bottom vertices have color zero and the two top vertices have
color one. This also explains why the distinguished source $a_0$ is
excluded: its would-be endpoint triangle contains $(2,3)$, which
has color zero under \eqref{eq:2d-complete-coloring}.

Now let $u\ne0$ be a source. When $\sigma_u=1$, its first enlarged
edge is $a_u\to a_u+e_1$. The three neighboring points directly above
$a_u$ have color one, $a_u+e_1$ has color zero, and the other four
neighbors have color two. Checking the six incident triangles listed
above gives exactly
$\{a_u,a_u+(-1,1),a_u+(-1,0)\}$, in color order.
When $\sigma_u=-1$, the local route and its left side are the half-turn
of this configuration about $a_u$. A half-turn preserves both the
triangulation in \eqref{eq:2d-grid-triangles} and the sign of the
determinant in the left-side rule. The unique fully colored triangle
therefore has offsets $(0,0),(1,-1),(1,0)$.

For a sink with $\tau_u=-1$, the route enters from above. The neighbor
$a_u+e_2$ has color zero, the three neighbors directly to the right have
color one, and the other four have color two. Its unique fully colored
triangle has offsets $(0,0),(1,-1),(0,-1)$, in color order. For
$\tau_u=1$, the half-turn gives offsets $(0,0),(-1,1),(0,1)$.
These are precisely the four cases of \eqref{eq:2d-endpoint-triangles}.
Since $u\ne0$ at every endpoint under consideration, $a_u$ has both
coordinates at least $12$, so the boundary adjustment does not affect
any of these neighborhoods.

Finally, the color-zero routes and the initial connection satisfy
$i+j\leq 18L-12<M-2$. Thus no unit triangle meeting the line
$i+j=M$, or lying beyond it in the square, has a vertex with color zero.
The displayed endpoint
triangles have both coordinates strictly between $0$ and $M$, so
they lie in the interior of the square.
\end{proof}

\subsection{Interpolation, endpoint signs, and the boundary}
\label{sec:discrete-grid}

We scale the integer grid by \(1/M\) to obtain the domain
\([0,1]^2\). The three colors must specify the values of a continuous field. The assignment of vectors to colors is chosen so that sinks have positive index and sources have negative index.

We assign the three field values
\begin{equation}
 q_0=e_2,\qquad q_1=e_1,\qquad q_2=-\one.
 \tag{Colors}
 \label{eq:discrete-palette}
\end{equation}
For every \(a\in\{0,\ldots,M-1\}^2\), the corresponding grid
square is divided into the two triangles with vertices
\begin{equation}
 \frac aM,\quad\frac{a+e_1}M,\quad\frac{a+e_2}M,
 \qquad\text{and}\qquad
 \frac{a+\one}M,\quad\frac{a+e_1}M,\quad\frac{a+e_2}M.
 \label{eq:kuhn-simplex}
 \tag{Triangles}
\end{equation}
Thus the diagonal joins the upper-left and lower-right corners.
Let \(c:\{0,\ldots,M\}^{2}\to\{0,1,2\}\) be the vertex coloring defined in \eqref{eq:2d-complete-coloring}.
For each integer grid vertex \(v=(i,j)\), set
$G(v/M)=q_{c(v)}=q_{c(i,j)}$,
and extend \(G\) affinely on each triangle.
A triangle is \emph{fully colored} if its three vertices receive
the three distinct colors \(0,1,2\).

The affine formulas on adjacent triangles give the same values at the
ends of their shared edge, so they agree along that edge. Thus \(G\)
is continuous. Every assigned value is a convex combination of the three color
vectors, giving \(\|G(x)\|_\infty\leq1\) throughout the square.

In the figures, an arrow based at an enlarged grid coordinate $y$ depicts
the direction of $G(y/M)$. At grid vertices, these directions are up,
right, and southwest for colors $0$, $1$, and $2$, respectively.
Open circles mark samples in triangle interiors, evaluated using the
same affine interpolation. Arrow lengths are normalized for readability;
their hues and the continuous background shading encode the angle of
$G$. The underlying routes are drawn with plain lines.

\begin{lemmabox}
\begin{lemma}[Zeros and their local matrices]
\label{lem:kuhn-barycenter}
The zeros of \(G\) are exactly the barycenters of its fully colored
triangles. Each zero is regular. Let \(p_j\) be the vertex of color \(j\) and also let $A_p$ be
the local matrix at their barycenter \(p\). Then, the following hold:
\begin{equation}
 E=[p_1-p_0\mid p_2-p_0],\qquad
 Q=[q_1-q_0\mid q_2-q_0]
   =\begin{pmatrix}1&-1\\-1&-2\end{pmatrix},\qquad
 A_p=QE^{-1}.
 \label{eq:kuhn-local-matrix}
 \nonumber
\end{equation}
In particular, \(G(x)=A_p(x-p)\) on this triangle and
\(\det(-A_p)=-3/\det E\).
\end{lemma}
\end{lemmabox}
\begin{proof}
A convex combination of the three field vectors is
\(\alpha_0q_0+\alpha_1q_1+\alpha_2q_2
=(\alpha_1-\alpha_2,\alpha_0-\alpha_2)\).
It vanishes exactly when all three coefficients equal \(1/3\).
A triangle has only three vertices, so it contains a zero exactly
when each color occurs once. The zero is then its barycenter and lies
strictly inside the triangle.

Column \(j\) of \(E\) is the vector from the vertex of color zero
to the vertex of color \(j\).
Since \(G\) is affine on the triangle and \(G(p)=0\), there is a
matrix \(A_p\) such that \(G(x)=A_p(x-p)\) throughout the triangle.
For \(j=1,2\), evaluating this identity at \(p_j\) and \(p_0\)
and subtracting gives
$A_p(p_j-p_0)
  =G(p_j)-G(p_0)
  =q_j-q_0$.
These are precisely the two column identities of \(A_pE=Q\).
The triangle has nonzero area,
which makes \(E\) nonsingular. Also \(\det Q=-3\), so \(A_p\)
is nonsingular. Since the dimension is two,
\(\det(-A_p)=\det A_p=\det Q/\det E\).
\end{proof}

\begin{lemmabox}
\begin{lemma}[The two endpoint signs]
\label{lem:discrete-endpoints}
Each sink of \(H\) gives exactly one zero of \(G\), with
\(\det(-A_p)=3M^2\). Each noncanonical source gives exactly one
zero, with \(\det(-A_p)=-3M^2\). There are no other zeros.
\end{lemma}
\end{lemmabox}
\begin{proof}
Lemma~\ref{lem:2d-discrete-coloring} identifies exactly one fully
colored triangle at each of these endpoints and excludes every other
such triangle. In the order of colors \(0,1,2\), the integer offsets from
$a_u=(9u+3,9u+3)$ are
\begin{equation}
 \begin{aligned}
 \text{source}\quad &(0,0),\ (-\sigma_u,\sigma_u),\ (-\sigma_u,0),\\
 \text{sink}\quad &(0,0),\ (-\tau_u,\tau_u),\ (0,\tau_u).
 \end{aligned}
 \label{eq:2d-endpoint-vertices}
\end{equation}
The color-ordered edge matrices after normalization by $1/M$ satisfy
\begin{equation}
 ME_{\rm source}=\sigma_u\begin{pmatrix}-1&-1\\1&0\end{pmatrix},
 \qquad
 ME_{\rm sink}=\tau_u\begin{pmatrix}-1&0\\1&1\end{pmatrix}.
 \label{eq:2d-endpoint-matrices}
\end{equation}
Since $\sigma_u^2=\tau_u^2=1$, their determinants are
$\det E_{\rm source}=1/M^2$ and $\det E_{\rm sink}=-1/M^2$,
independently of the incident directions. Lemma~\ref{lem:kuhn-barycenter}
then gives $\det(-A_p)=-3M^2$ at a source and $3M^2$ at a sink.
\end{proof}

For the illustrative example, $M=256$, $\sigma_4=-1$, $\tau_1=-1$,
and $\tau_5=1$. The three zero locations are
\[
 p_4=\frac{(119,116)}{3M},\qquad
 p_1=\frac{(37,34)}{3M},\qquad
 p_5=\frac{(143,146)}{3M}.
\]
They are the barycenters of the triangles shown in
Figure~\ref{fig:two-isolated-edges}, with indices $-1,+1,+1$,
respectively. At the canonical source, the triangle with vertices
$(3,3),(2,4),(2,3)$ has colors $0,1,0$, as shown in
Figure~\ref{fig:enlarged-route-turn-switch}(a), so it contains no zero.

\begin{lemmabox}
\begin{lemma}[Clipping creates no boundary fixed point]
\label{lem:2d-boundary-clipping}
For every \(x\in\partial[0,1]^2\) and every \(\eta>0\),
\([x+\eta G(x)]_0^1\ne x\).
\end{lemma}
\end{lemmabox}
\begin{proof}
Clipping is applied coordinatewise. Thus, if $x_i<1$ and
$G_i(x)>0$, the $i$th coordinate of $[x+\eta G(x)]_0^1$
is greater than $x_i$ for every $\eta>0$. Likewise, if
$x_i>0$ and $G_i(x)<0$, that coordinate becomes smaller.
We find such a coordinate on each side of the square.

On the left edge, only colors $0$ and $1$ occur. We may
therefore write $G(x)=\lambda_0q_0+\lambda_1q_1$, where
$\lambda_0+\lambda_1=1$ and both coefficients are
nonnegative. If $\lambda_1>0$, then $G_1(x)=\lambda_1>0$,
so the first coordinate increases from zero. If
$\lambda_1=0$, then $G(x)=q_0=e_2$. This case occurs only
at heights $x_2\leq 3/M<1$, so the second coordinate
increases.

On the bottom edge, at every point other than the origin,
color $2$ has positive interpolation weight and color $1$
is absent. Since $q_0=(0,1)$ and $q_2=(-1,-1)$, it follows
that $G_1(x)<0$. Here $x_1>0$, so the first coordinate
decreases after clipping. The origin was covered by the
left-edge argument.

On the top edge, at every point other than the upper-left
corner, color $2$ has positive interpolation weight and the
only other possible color is $1$. Both $q_1=(1,0)$ and
$q_2=(-1,-1)$ have nonpositive second coordinate, and
$q_2$ has negative second coordinate. Hence $G_2(x)<0$,
so the second coordinate decreases from one. The
upper-left corner was covered by the left-edge argument.

Finally, every vertex on the right edge has color $2$.
Consequently $G(x)=q_2=(-1,-1)$ throughout that edge,
and its first coordinate decreases from one. In each
case, at least one coordinate changes for every $\eta>0$.
\end{proof}

\subsection{An exact circuit for interpolation}
\label{sec:discrete-interpolation}

At any particular point, only the vertices of its containing triangle contribute to \(G\), so there are at most three colors to look up. The challenge is to find those vertices with one polynomial-size, rational Linear-$\FIXP$ circuit.
The obvious instruction is to compute \(\lfloor Mx_1\rfloor\) and \(\lfloor Mx_2\rfloor\), then query the nearby colors using the floor operation. However, this would jump discontinuously at grid lines.

The coloring circuit evaluates $c(v)$ at any specified grid vertex
$v$, but we cannot evaluate it at every vertex when $M$ is given in binary. 
For $x\in[0,1]^2$, let $y=Mx$ and define $h_v(Mx)$ as the interpolation weight of the grid vertex $v$ at $x$. Thus $G(x)=\sum_v h_v(Mx)q_{c(v)}$.
The weight $h_v(y)$ can be positive only if $|Mx_i-v_i|<1$ for $i=1,2$.
We partition the vertices according to their two coordinate remainders modulo three. 
There are thus nine groups, and at most one vertex in each group can have positive weight at a given point. The circuit therefore needs to find and color only one candidate vertex per group.

Finding a candidate requires decoding integers from the coordinates
of $y$. Near a decoding threshold, those integers may not be
recovered correctly. We arrange the decoding so that, whenever this
happens, every vertex in the group has weight zero. We can then set
the group's contribution to zero without changing the interpolated
field. 
For example, consider vertices whose first coordinate is divisible by \(3\). The candidates with first coordinates \(0\) and \(3\) can have positive weight only when the first coordinate of \(y\) is within one unit of \(0\) or \(3\), respectively. The decoder switches between them at \(y_1=1.6\). At that point, and in a small neighborhood around it, neither candidate has positive weight. So it is safe for this group to output zero while decoding is uncertain. If decoding is reliable, the circuit uses the decoded vertex’s color and interpolation weight.
The lemma below implements these nine contributions with a
circuit of size and coefficient bit length polynomial in the size of the color circuit and
$\log(M+1)$. It also checks that no max gate has equal inputs at
a triangle barycenter.

The interpolation weights and the partition by coordinate remainders are adapted from \citet[Sections 4.2 and 5.1.1]{yarotsky2018optimal}.
Suppressing a contribution when its local data cannot be selected reliably is related to \citet[Section 5.2.4]{yarotsky2018optimal}.
For related exact interpolation constructions, see \citet{fearnley2025complexity}.



\begin{lemmabox}
\begin{lemma}[Exact interpolation in two dimensions]
\label{lem:2d-exact-interpolation}
Given a positive integer \(M\) in binary and a Boolean circuit of size
\(s\) computing the coloring
\(c\colon\{0,\ldots,M\}^2\to\{0,1,2\}\), one can construct a rational
Linear-$\FIXP$ circuit for the affine interpolation of
\(v\mapsto q_{c(v)}\) on the triangulation of \([0,M]^2\) whose two
triangles in each unit square have vertex sets
\(\{a,a+e_1,a+e_2\}\) and
\(\{a+\one,a+e_1,a+e_2\}\), where \(a\in\{0,\ldots,M-1\}^2\).
Here \(q_0=e_2\), \(q_1=e_1\), and \(q_2=-\one\).
The circuit size, construction time, and coefficient bit lengths are
polynomial in \(s+\log(M+1)\). At every triangle barycenter, the two
inputs of every max gate are unequal.
\end{lemma}
\end{lemmabox}

\begin{proof}
We work in integer grid coordinates throughout the construction.
We first treat the triangulation in which each unit square has the
diagonal from \(a\) to \(a+\one\). At the end of the proof, a reflection
of the second input coordinate gives the diagonals in the statement.

\paragraph{The interpolation weight of a vertex.}
For an integer vertex \(v\), we define
\begin{equation}
\begin{aligned}
 g_v(x)
   &=1-\max\{0,x_1-v_1,x_2-v_2\}
        +\min\{0,x_1-v_1,x_2-v_2\},\\
 h_v(x)&=\max\{0,g_v(x)\}.
\end{aligned}
\label{eq:2d-vertex-weight}
\end{equation}
The function $h_v$ is the barycentric weight of $v$ in this
triangulation. Let $x=a+t$, where $a\in\mathbb Z^2$ and
$t=(t_1,t_2)\in[0,1]^2$. Let $\pi$ be a permutation of
$\{1,2\}$ such that $t_{\pi(1)}\geq t_{\pi(2)}$, breaking ties
arbitrarily. One of the triangles containing $x$ has vertices
$a$, $a+e_{\pi(1)}$, and $a+\one$. Substitution into
\eqref{eq:2d-vertex-weight} gives their respective weights
\begin{equation}
 1-t_{\pi(1)},\qquad
 t_{\pi(1)}-t_{\pi(2)},\qquad
 t_{\pi(2)}.
\label{eq:2d-three-weights}
\tag{Coef}
\end{equation}
When $t_1=t_2$, the middle weight is zero, so either choice of
$\pi$ gives the same result.
These weights are nonnegative, sum to one, and their weighted sum of
vertices is \(x\). Every other integer vertex has weight zero.
For a vertex outside this square, some coordinate difference has absolute
value at least one, which makes \(g_v(x)\leq0\). For a square vertex outside
the indicated chain, there are indices \(i,j\) with \(v_i=a_i+1\),
\(v_j=a_j\), and \(t_i\leq t_j\), which also gives \(g_v(x)\leq0\).
The formulas include shared edges, where some weights vanish.

In particular, \(h_v(x)>0\) implies \(|x_i-v_i|<1\) for all \(i\in\{1,2\}\).
Vertices outside \(\{0,\ldots,M\}^2\) therefore have weight zero on
\([0,M]^2\). The field to be computed is
\(G(x)=\sum_v h_v(x)q_{c(v)}\), with the sum over the grid vertices.
The circuit below will compute this sum without enumerating the vertices.

\paragraph{The nine groups.}
For each \(r\in\{0,1,2\}^2\), we consider the vertices whose coordinate
remainders modulo three are \(r_1,r_2\). Thus, two distinct vertices in this group differ by at least three in some coordinate. 
Note that their barycentric weights cannot both be positive at the same point, since that point would have to be less than one unit from both vertices in this coordinate.

We describe the circuit for one fixed group \(r\). 
In coordinate $i$, the vertices of this
group have positions $r_i+3k$, where $k$ is an integer. We specify
the vertex coordinate that the circuit should recover by setting
\begin{equation}
 z_i=1+\frac{x_i-r_i+7/5}{3},
 \qquad
 v_i^*=r_i+3\bigl(\lfloor z_i\rfloor-1\bigr).
\label{eq:2d-group-candidate}
\end{equation}
The floor is not a gate
of the circuit.

To see that this choice works, suppose a vertex $w$ in group $r$
has positive interpolation weight at $x$. Write $w_i=r_i+3k_i$.
By \eqref{eq:2d-vertex-weight}, $|x_i-w_i|<1$, and hence
\begin{equation}
 z_i=k_i+1+\frac{x_i-w_i+7/5}{3},
 \qquad
 \frac{2}{15}
 <\frac{x_i-w_i+7/5}{3}
 <\frac45.
\label{eq:2d-positive-weight-decoding}
\end{equation}
The second term of \eqref{eq:2d-positive-weight-decoding} lies strictly between zero and one. Therefore
$\lfloor z_i\rfloor=k_i+1$ and $v_i^*=w_i$ for both coordinates.
The candidate changes between $w_i$ and $w_i+3$ at
$x_i'=w_i+8/5$. 
At the switching point,
\[
|x_i'-w_i|=\frac85>1,
\qquad
|x_i'-(w_i+3)|=\frac75>1.
\]
So, at that coordinate, neither vertex can have
positive weight. The choice of $7/5$ will also make the decoding
comparisons strict at triangle barycenters, as checked below.

\paragraph{Decoding the candidate.}
This part builds a continuous substitute for the integer \(\lfloor z_i\rfloor\). The circuit needs that integer to identify the candidate vertex, but floor itself is not an allowed continuous gate.
We use \(\varepsilon=1/1000\) and
\(B=\lceil\log_2(M+1)\rceil+2\), so that \(0<z_i<2^B\).
For a real number \(u\), the notation \([u]_0^1\) means the two gate
expression \(\max\{0,\min\{1,u\}\}\).
Starting with \(u_{i,B}=z_i\), we compute, in the order
\(j=B-1,B-2,\ldots,0\),
\begin{equation}
 b_{i,j} =\left[\frac{u_{i,j+1}-2^j+\varepsilon}{2\varepsilon}\right]_0^1,
 \qquad
 u_{i,j}=u_{i,j+1}-2^j b_{i,j}.
 \tag{Decoding}
\label{eq:2d-binary-decoding}
\end{equation}
The above essentially asks: Is the remaining number at least \(2^j\)? When the answer is \(1\), we subtract \(2^j\) and move to the next, smaller power of two. The use of $\varepsilon$ makes this test continuous:
The returned bit $b_{i,j}$ is \(0\) below the threshold \(2^j-\varepsilon\) and \(1\) above \(2^j+\varepsilon\). Only inside the narrow interval between those values can it return a fractional bit.

Therefore, if the distance from \(z_i\) to the nearest integer exceeds
\(\varepsilon\), these values are exactly the binary digits of
\(\lfloor z_i\rfloor\). 
This follows by induction over the digits.
After the higher digits have been subtracted, the residual has the same fractional part as \(z_i\). It is therefore outside the interval
\([2^j-\varepsilon,2^j+\varepsilon]\), and the clipping operation
returns the correct digit. 
At other inputs, each \(b_{i,j}\) still belongs to \([0,1]\), although it need not be Boolean. The circuit's candidate coordinates are the affine expressions
\(\widetilde v_i=r_i+3(\sum_{j=0}^{B-1}2^j b_{i,j}-1)\).
For example, \(z_i=5.3\) gives \(b_{i,2}=1\), leaving \(u_{i,2}=1.3\); then \(b_{i,1}=0\) and \(b_{i,0}=1\). The decoded integer is then \(4+1=5=\lfloor5.3\rfloor\).

The digit circuit above computes \(\lfloor z_i\rfloor\) exactly away from integers, but its outputs may be fractional near an integer. We next compute the distance to the nearest integer. This lets us suppress a group's contribution whenever its candidate may have been decoded incorrectly.
For this purpose, we compute
$D_i=\operatorname{dist}(z_i,\mathbb Z)$ using minimum gates:
\begin{equation}
 t_{i,B+1}=z_i,\qquad
 t_{i,j}=\min\{t_{i,j+1},\,2^j-t_{i,j+1}\}
 \quad(j=B,B-1,\ldots,0),\qquad
 D_i=t_{i,0}.
\label{eq:2d-integer-distance}
\end{equation}
To see why this works, suppose the input $t_{i,j+1}$ at step $j$
lies in $[0,2^j]$. The minimum lies in $[0,2^{j-1}]$.
If the minimum selects $2^j-t_{i,j+1}$, the distance to the
nearest integer is unchanged because $2^j$ is an integer.
It is also unchanged if the minimum selects $t_{i,j+1}$.
Starting from $0<z_i<2^B$, the recurrence therefore preserves
the distance to $\mathbb Z$ at every step. Its final value lies
in $[0,1/2]$, so that value is the distance itself. For example,
when $B=3$, an input of $z_i=5.3$ passes through
$2.7$, $1.3$, $0.7$, and $0.3$.

Now we set $\sigma_i=16(D_i-\varepsilon)$. If $D_i\leq\varepsilon$,
the decoder may be unreliable, but $\sigma_i\leq0$ will make the
group's contribution zero in the formula below. 
If $D_i>\varepsilon$, it follows that the decoded bits are exact. 
Moreover, if a vertex of the group has positive interpolation weight, \eqref{eq:2d-positive-weight-decoding}
places the fractional part of $z_i$ in $(2/15,4/5)$. Hence
$D_i>2/15$ and $\sigma_i>16(2/15-\varepsilon)>1$.

The reason we can assign zero weight is geometric. Suppose \(D_i\leq\varepsilon\). For some integer \(m\), write \(z_i=m+\delta\) with \(|\delta|\leq\varepsilon\). Rearranging the definition of \(z_i\), and setting \(w=r_i+3(m-2)\), gives
$x_i=w+\frac85+3\delta$.
Since \(\varepsilon=1/1000\), this places \(x_i\) between \(w+1.597\) and \(w+1.603\), well inside the gap \((w+1,w+2)\).
The neighboring coordinates of vertices in this group are \(w\) and \(w+3\). A vertex can have positive weight only if its coordinate is less than \(1\) from \(x_i\). But throughout this gap, \(x_i\) is more than \(1\) from both \(w\) and \(w+3\). Every other vertex in the group is farther away. Hence the true contribution of this entire group is already zero.



\paragraph{Computing and selecting the color.}
We extend the given Boolean color circuit to every possible tuple of
the \(2B\) decoded bits. On Boolean inputs, it first computes the
integers \(k_i=\sum_j2^j b_{i,j}\) and the signed integers
\(r_i+3(k_i-1)\). These signed integers fit in \(B+4\) bits. 
If both belong to \(\{0,\ldots,M\}\), it evaluates the given color circuit on their binary representations according to \eqref{eq:2d-complete-coloring}. We overload the coloring function $c$ giving color zero for points that are not integer vertices. 
All integer arithmetic, range comparisons, and the conversion to
the original input representation require a Boolean circuit of size
polynomial in \(B\).

On Boolean decoded inputs, the extended circuit has three output color indicators \(c_0,c_1,c_2\), where $c_j:=\mathds{1}\{c(\widetilde v)=j\}$. We replace the Boolean
operations by
\begin{equation}
 \operatorname{NOT}(a)=1-a,
 \qquad
 \operatorname{AND}_*(a,b)=\max\{0,2a+2b-3\}.
\label{eq:2d-strict-boolean}
\end{equation}
OR is obtained by De Morgan's identity. These operations are correct on
Boolean inputs and preserve the interval \([0,1]\) on arbitrary inputs
from that interval. Every max gate in this evaluation is strict on
Boolean inputs, because its second argument belongs to \(\{-3,-1,1\}\).
The range comparisons and the choice of the output color are part of
this same Boolean circuit, so the statement also applies to them.

For \(j\in\{0,1,2\}\), the contribution of color \(j\) in group
\(r\) is
\begin{equation}
 H_{r,j}(x)=
 \max\bigl\{0,
   \min\{g_{\widetilde v}(x),\sigma_1,\sigma_2,2c_j-1\}
 \bigr\}.
\label{eq:2d-group-contribution}
\tag{Contribution}
\end{equation}
The minimum is computed successively, starting with
\(g_{\widetilde v}(x)\), then comparing with \(\sigma_1,\sigma_2\),
in this order, and finally with \(2c_j-1\). The use of
\(g_{\widetilde v}\), before taking its maximum with zero, and this
order of comparisons will be important for strictness.
If some \(D_i\leq\varepsilon\), then \(\sigma_i\leq0\), so all three contributions are zero regardless of the decoded values. Indeed, any vertex of positive weight would imply \(D_i>2/15\) in every coordinate by
\eqref{eq:2d-positive-weight-decoding}.

Otherwise, every decoded bit is correct and \(\widetilde v=v\) is an
integer vertex. If \(g_v(x)\leq0\), all three contributions are again
zero, and no other vertex of the group can have positive weight.
If \(g_v(x)>0\), the vertex belongs to the grid, all \(\sigma_i>1\),
and exactly one indicator, \(c_{c(v)}\), is one. Since
\(0<g_v(x)\leq1\), \eqref{eq:2d-group-contribution} gives
\(H_{r,c(v)}(x)=g_v(x)=h_v(x)\), with all other contributions zero.
Thus the circuit computes the required interpolation exactly through
\begin{equation}
 G(x)=\sum_{r\in\{0,1,2\}^2}\sum_{j=0}^2 H_{r,j}(x)q_j.
\label{eq:2d-exact-field-circuit}
\nonumber
\end{equation}

\paragraph{Strictness at barycenters.}
At a triangle barycenter, the fractional parts of the two coordinates
of \(x\) are \(1/3\) and \(2/3\), in some order. 
If the fractional part of \(x_i\) is \(a/3\) with
\(a\in\{1,2\}\), the fractional part of \(z_i\) is that of
\((15b+5a+21)/45\), where \(b\in\{0,1,2\}\) is the remainder of
\(\lfloor x_i\rfloor-r_i\) modulo three. These six choices give
\begin{equation}
 D_i\in
 \left\{\frac1{45},\frac4{45},\frac{11}{45},
 \frac{14}{45},\frac{16}{45},\frac{19}{45}\right\}.
\label{eq:2d-barycenter-distances}
\nonumber
\end{equation}
In particular, \(D_i\geq1/45>\varepsilon\). The bits are Boolean, and
every argument inside the clipping operation in
\eqref{eq:2d-binary-decoding} is strictly less than zero or strictly
greater than one. Both gates in each clipping operation are therefore
strict. Every max gate in the Boolean color computation is strict by
\eqref{eq:2d-strict-boolean}.

There are no ties in~\eqref{eq:2d-integer-distance} either. A tie at
step \(j\) would put the current input at \(2^{j-1}\). Every earlier
operation was an integer translation with a possible change of sign,
so a tie would require \(z_i\in\tfrac12\mathbb Z\).
The fractional parts just computed exclude this possibility.

For any decoded candidate \(v\), the numbers
\(0,x_1-v_1,x_2-v_2\) have distinct fractional parts. All
comparisons in the maximum and minimum defining \(g_v\) are therefore
strict. Moreover, their largest and smallest elements have distinct
fractional parts, so \(g_v(x)\) is not an integer. If it is positive,
it is one of the three weights in~\eqref{eq:2d-three-weights} and equals
\(1/3\). Otherwise, it is strictly negative.

At these points, every \(\sigma_i\) is at least
\(16(1/45-1/1000)=382/1125>1/3\).
Thus each of the first two comparisons in
\eqref{eq:2d-group-contribution} strictly retains \(g_v(x)\).
The next comparison is with \(-1\) or \(1\), and is strict because
\(g_v(x)\) is not an integer. Its result is nonzero, so the final
maximum with zero is also strict. This includes every group whose
contribution is zero.

\paragraph{Size and rational constants.}
Each group uses \(\mathcal O(B)\) operations for decoding and distance
computation, a Boolean circuit of size polynomial in \(s+B\), and a
constant number of further affine, minimum, and maximum operations.
There are only nine groups. Every minimum is represented as
\(\min\{a,b\}=-\max\{-a,-b\}\), which preserves strictness.
The rational constants have bit length \(\mathcal O(B)\), including the powers
of two used in decoding. All circuits can therefore be constructed in
the claimed polynomial time.

\paragraph{The diagonal direction.}
We define \(R(u_1,u_2)=(u_1,M-u_2)\). This reflection takes the two
triangles in the statement to triangles of the triangulation used above.
The Boolean circuit for the reflected colors
\(c'(v_1,v_2)=c(v_1,M-v_2)\) is obtained by adding binary subtraction
to the given color circuit. It still has size polynomial in
\(s+\log(M+1)\). The preceding construction applied to \(c'\), with
input \(x=R(u)\), therefore computes exactly the desired affine
interpolation at \(u\). Only the input coordinates are reflected.
The output vectors \(q_0,q_1,q_2\) keep their stated values.
A reflection takes a triangle barycenter to a triangle barycenter, so
all strictness conclusions continue to hold. Rescaling the input square
adds only affine operations and also preserves strictness.
\end{proof}

\subsection{The fixed point map and decoder}
\label{sec:discrete-completion}

We now combine the coloring with the exact interpolation circuit. The clipped fixed point map has zeros of \(G\) as its fixed points. 

\begin{proof}[Proof of Theorem~\ref{thm:axis-unconditional-map}]
Lemma~\ref{lem:2d-discrete-coloring} gives a Boolean circuit of polynomial
size for the coloring. Lemma~\ref{lem:2d-exact-interpolation}, with the
affine input substitution \(x\mapsto Mx\), computes its interpolation
\(G\) exactly on \([0,1]^2\). We define the map
\begin{equation}
 \Phi(y)=[y+\eta G(y)]_0^1,\qquad \eta=\frac1{32}.
 \label{eq:2d-raw-map}
 \nonumber
\end{equation}
Lemma~\ref{lem:2d-boundary-clipping} excludes boundary fixed points.
At an interior point, \(\Phi(y)=y\) is equivalent to \(G(y)=0\).
Its fixed points are therefore exactly the endpoint barycenters in
Lemma~\ref{lem:discrete-endpoints}.

At such a point \(p\), every gate computing \(G\) is strict. Since
\(p_i\in(0,1)\), the two gates clipping the value
\(p_i+\eta G_i(p)=p_i\) are also strict. In a neighborhood of \(p\),
\(D\Phi(p)=I_2+\eta A_p\), and hence
\begin{equation}
 \det(I_2-D\Phi(p))=\eta^2\det(-A_p),\qquad
 \ind_\Phi(p)=\sign\det(-A_p).
 \label{eq:discrete-index-map}
\end{equation}
The positive factor \(\eta^2\) preserves the determinant sign.
Thus these fixed points are regular, with positive index exactly at
the sinks.


\paragraph{Decoding.}
Let $p$ be a fixed point of $\Phi$ with positive index.
It is the barycenter of a sink triangle in
\eqref{eq:2d-endpoint-vertices}. Thus
$Mp=a_u+(-\tau_u/3,2\tau_u/3)$, where
$a_u=(9u+3,9u+3)$ and $\tau_u\in\{-1,1\}$. In particular,
\[
 9u< Mp_1=9u+3-\tau_u/3<9(u+1).
\]
Thus the decoder, without knowing the incident direction, recovers
\begin{equation}
 u=\left\lfloor\frac{Mp_1}{9}\right\rfloor.
 \label{eq:2d-sink-decoder}
\end{equation}
This floor is evaluated by the discrete decoder using exact rational
arithmetic; it is not a gate in the continuous circuit. Since $u$ is a
sink of $H$, Section~\ref{sec:signed-route-geometry} gives
$P(S(u))\ne u$. Its $n$-bit representation is therefore a valid
\SoL\ output.

\paragraph{Encoding bounds.}
The integer $M=32L=2^{n+5}$ has linear bit length.
The field bound is $M_G=1$, and every fixed point coordinate has
denominator dividing $3M$. Since the fixed points are interior,
their distance from the boundary is at least $1/(3M)$.
This gives an explicit positive rational bound of polynomial
bit length. The endpoint and its three triangle vertices also
determine $E$ and $A_p$ in polynomial time.

An open ball of radius $1/(12M)$ in the $\ell_\infty$ norm
around a fixed point $p$ remains in its triangle.
Indeed, in integer grid coordinates, the barycenter $Mp$ is at
distance at least $1/3$ from each side of its unit square.
The sum of its coordinates differs by $1/3$ from the value
defining the diagonal. A perturbation of $p$ of the stated radius
changes each integer grid coordinate by less than $1/12$
and their sum by less than $1/6$. It therefore crosses neither
a side of the square nor its diagonal. Hence $G$ agrees with
its affine formula throughout this neighborhood.

Together with the preceding lemmas, these bounds establish
all five properties of Definition~\ref{def:certified-signed-field}
for $G$. The correspondence between zeros of $G$ and fixed points
of $\Phi$, the index identity \eqref{eq:discrete-index-map},
and the decoder \eqref{eq:2d-sink-decoder} establish the assertions
about $\Phi$. The construction uses polynomial time and polynomial
bit length throughout.
\end{proof}

\section{From fixed points to complementarity}
\label{sec:fixp-to-lcp}

We next replace the \LinFIXP circuit by a single linear
complementarity problem, see \citet{CottlePangStone1992} for background.  The construction is the one block
construction underlying Section~4.3 of \citet{Mehta2018}.  We record it
in some detail because, beyond preserving the set of solutions, we need
it to preserve their signs.  This is the reason for working directly
with the complementarity conditions of the max gates.

\subsection{A normal form for circuits strict at their fixed points}
\label{subsec:gate-normal-form}

Let
\[
        \Phi\colon [0,1]^d\longrightarrow [0,1]^d
\]
and its rational max-affine circuit \(\mathcal C\) guaranteed by
Theorem~\ref{thm:axis-unconditional-map}.  We use the following two
certified properties of \(\mathcal C\).  First, every fixed point is
\emph{circuit strict}, i.e., when the circuit is evaluated at a fixed point,
the two inputs of every max gate are different.  Second, every fixed
point is regular.  Because of continuity, the same input of every max gate remains selected throughout a
neighborhood of the fixed point.  Replacing each max gate by that selected
affine input yields the local affine formula for \(\Phi\) and its Jacobian,
and
\[
        \det\bigl(I_d-D\Phi(\boldsymbol{\lambda})\bigr)\ne 0
        \qquad\text{whenever }
        \Phi(\boldsymbol{\lambda})=\boldsymbol{\lambda}.
\]

Following \citet{Mehta2018}, additions and multiplications by rational
constants can be absorbed into the affine inputs of the max gates.
Moreover, a gate
\[
        y=\max\{a,b\}
\]
can be replaced by
\[
        x=\max\{0,b-a\},\qquad y=a+x.
\]
This replacement preserves strictness. Observe the new max gate ties exactly when
\(b-a=0\), which is exactly when the inputs \(a\) and \(b\) of the original
max gate tie.
After processing the gates in topological order, the circuit 
has \(m_0\) distinguished max-gate variables
\(\mathbf{x}=(x_1,\ldots,x_{m_0})^{\transpose}\), and gate \(i\) has the
form
\begin{equation}
    x_i=\max\{0,\ell_i(\mathbf{x}_{<i},
                        \boldsymbol{\lambda})\},
    \qquad
    \ell_i
       =b_i+\mathbf{u}_i^{\transpose}\boldsymbol{\lambda}
          +\sum_{j<i}c_{ij}x_j .
    \label{eq:relu-gate}
\end{equation}
Here \(\boldsymbol{\lambda}\in\mathbb R^d\) is the input to the
circuit, \(b_i,c_{ij}\in\mathbb Q\), and
\(\mathbf{u}_i\in\mathbb Q^d\).

We next modify the representation of each output coordinate without changing
the map \(\Phi\).  Let \(h_\ell\) denote the value currently designated as
the \(\ell\)th output coordinate.  The circuit clips this value to
\([0,1]\).  For every \(\ell\in[d]\), append the two max gates
\begin{equation}
  p_\ell=\max\{0,1-h_\ell\},
  \qquad
  o_\ell=\max\{0,1-p_\ell\}.
  \label{eq:anchored-output-gates}
\end{equation}
Since \(0\leq h_\ell\leq1\), we have
\(p_\ell=1-h_\ell\) and \(o_\ell=h_\ell\) for every input.  We designate
\(o_\ell\) as the \(\ell\)th output, so the computed map remains unchanged
and every output coordinate is the value of a nonnegative max gate.  The
use of both gates also creates the structural row identity used in
Section~\ref{sec:lcp-to-game} to force the anchoring strategy \(\star\) into
every equilibrium. By
Theorem~\ref{thm:axis-unconditional-map},
at a fixed point \(\boldsymbol{\lambda}\), we have
\(h_\ell=\Phi_\ell(\boldsymbol{\lambda})=\lambda_\ell\in(0,1)\).
Hence $p_\ell=1-\lambda_\ell>0$ and $o_\ell=\lambda_\ell>0$,
so both appended max gates are strict.
We place \(p_\ell\) immediately before \(o_\ell\) in the topological ordering
of the max gates. 

Set \(m:=m_0+2d\) and extend
\(\mathbf x\) by these appended variables to obtain
\(\mathbf{x}=(x_1,\ldots,x_m)^{\transpose}\).  After reindexing, the
normal form~\eqref{eq:relu-gate} applies to every \(i\in[m]\), including
the appended gates. There is a selector matrix
$
        V\in\{0,1\}^{m\times d}
$
such that, if \(\mathbf{x}(\boldsymbol{\lambda})\) denotes the complete
vector of max-gate values, then
\begin{equation}
        \Phi(\boldsymbol{\lambda})
             =V^{\transpose}\mathbf{x}(\boldsymbol{\lambda}).
    \label{eq:output-selector}
\end{equation}
Each column of \(V\) has a single nonzero entry.
We define \(A\in\mathbb Q^{m\times m}\),
\(U\in\mathbb Q^{m\times d}\), and
\(\mathbf b\in\mathbb Q^m\) by
\[
 A_{ii}=1,\qquad
 A_{ij}=-c_{ij}\ \text{ for }j<i,\qquad
 A_{ij}=0\ \text{ for }j>i,
\]
by taking row \(i\) of \(U\) to be \(\mathbf u_i^{\transpose}\), and
by taking coordinate \(i\) of \(\mathbf b\) to be \(b_i\).  Thus \(A\)
is lower triangular with diagonal equal to one.  If
\[
        \mathbf w(\boldsymbol{\lambda},\mathbf x)
               :=A\mathbf x-U\boldsymbol{\lambda}-\mathbf b,
\]
then \(w_i=x_i-\ell_i\).  Hence the complete execution of the circuit
on input \(\boldsymbol{\lambda}\) is characterized by
\begin{equation}
 \mathbf x\geq 0,\qquad
 \mathbf w\geq 0,\qquad
 x_iw_i=0
 \quad(i\in[m]).
 \label{eq:circuit-complementarity}
\end{equation}

The anchored output gates imply a structural identity that will exclude
spurious equilibria in the game construction.  In the row of $o_\ell$,
the affine input is $1-p_\ell$ and has no direct dependence on the circuit
input \(\boldsymbol{\lambda}\).  Thus the entire \(o_\ell\)-row of \(U\)
is zero, and
\begin{equation}
 \begin{aligned}
  A_{o_\ell,o_\ell}&=A_{o_\ell,p_\ell}=1,
  & A_{o_\ell,j}&=0
     &&\text{for }j\notin\{p_\ell,o_\ell\},\\
  U_{o_\ell,:}&=0^{\transpose},
  & b_{o_\ell}&=1.
 \end{aligned}
  \label{eq:anchored-output-row}
\end{equation}

\begin{lemmabox}
\begin{lemma}
\label{lem:gate-normal-form}
For every \(\boldsymbol{\lambda}\in\mathbb R^d\),
system~\eqref{eq:circuit-complementarity} has exactly one solution
\(\mathbf x\), namely the vector
\(\mathbf{x}(\boldsymbol{\lambda})\) obtained by evaluating
\(\mathcal C\).  The matrices \(A,U,V\) and vector \(\mathbf b\) can
be constructed in polynomial time and have encoding length polynomial
in the encoding length of \(\mathcal C\).
\end{lemma}
\end{lemmabox}

\begin{proof}
For a circuit evaluation, equation~\eqref{eq:relu-gate} gives
\(x_i\geq0\), \(w_i = x_i-\ell_i\geq0\), and
\(x_i(x_i-\ell_i)=0\), so the resulting vector satisfies
\eqref{eq:circuit-complementarity}.

Conversely, suppose that \(\mathbf x\) satisfies
\eqref{eq:circuit-complementarity}.  For \(i=1\), complementarity
forces \(x_1=\max\{0,\ell_1(\boldsymbol{\lambda})\}\).  Inductively,
once \(x_1,\ldots,x_{i-1}\) have been fixed, the same two inequalities
and the complementary equality force
\[
        x_i=\max\{0,\ell_i(\mathbf{x}_{<i},
                            \boldsymbol{\lambda})\}.
\]
Thus all coordinates agree with the unique circuit evaluation.

The circuit can be transformed into this normal form by processing its gates
in topological order.  At each
circuit node we maintain its affine expression in
\(\boldsymbol{\lambda}\) and in the preceding max-gate variables.
There are at most \(m+d+1\) coefficients per expression.  Addition
and multiplication by a rational constant increase their bit
length by at most a polynomial amount over the entire circuit.
Thus both the number and the encoding length of the resulting
coefficients are polynomial in size of \(\mathcal C\).
\end{proof}

\subsection{The complementarity problem}
\label{subsec:one-block-lcp}

The notation
\[
             \boldsymbol{\lambda}\longmapsto
             \mathbf x(\boldsymbol{\lambda})
\]
denotes the deterministic map that evaluates every max gate of the circuit
on input \(\boldsymbol{\lambda}\).  At a fixed point,
equations~\eqref{eq:output-selector} and
\eqref{eq:circuit-complementarity} allow us to substitute
\(\boldsymbol{\lambda}=V^{\transpose}\mathbf x\).  Define
\begin{equation}
        M:=A-UV^{\transpose}\in\mathbb Q^{m\times m}.
    \label{eq:def-M}
\end{equation}

\begin{lemmabox}
\begin{lemma}[Output anchoring]
\label{lem:output-anchoring}
For every $z\in\mathbb R_{\geq0}^m$ and every output coordinate
$o_\ell$,
\begin{equation}
       (Mz)_{o_\ell}=z_{p_\ell}+z_{o_\ell}
       \geq z_{o_\ell}.
  \label{eq:output-anchor-identity}
\end{equation}
Moreover,
\begin{equation}
       M=A-\sum_{\ell=1}^{d}u^\ell e_{o_\ell}^{\transpose},
  \label{eq:anchored-rank-form}
\end{equation}
where $u^\ell$ is the $\ell$th column of $U$.
\end{lemma}
\end{lemmabox}

\begin{proof}
Since row $o_\ell$ of $U$ is zero, the same row of $M$ equals the
corresponding row of $A$.  Equation~\eqref{eq:anchored-output-row}
then gives \eqref{eq:output-anchor-identity}.  Since the $\ell$th
column of $V$ is $e_{o_\ell}$, expanding $UV^{\transpose}$ gives
\eqref{eq:anchored-rank-form}.
\end{proof}
We obtain the following LCP:
\begin{equation}
 \mathbf x\geq0,\qquad
 \mathbf w:=M\mathbf x-\mathbf b\geq0,\qquad
 x_iw_i=0\quad(i\in[m]).
 \label{eq:one-block-lcp}
\end{equation}
We use the convention that \(\operatorname{LCP}(M,\mathbf q)\) denotes
\[
 \mathbf x\geq0,\qquad
 M\mathbf x+\mathbf q\geq0,\qquad
 x_i(M\mathbf x+\mathbf q)_i=0\quad(i\in[m]).
\]
As a result, \eqref{eq:one-block-lcp} is precisely
\(\operatorname{LCP}(M,-\mathbf b)\), the vector \(\mathbf b\)
is the right-hand side, while the LCP offset is \(-\mathbf b\).

\begin{lemmabox}
\begin{lemma}[Exact correspondence]
\label{lem:fixp-lcp-bijection}
There is a bijection
\[
 \operatorname{Fix}(\Phi)
   \longleftrightarrow
 \operatorname{Sol}\bigl(\operatorname{LCP}(M,-\mathbf b)\bigr).
\]
The forward map sends a fixed point
\(\boldsymbol{\lambda}\) to the circuit execution
\(\mathbf{x}(\boldsymbol{\lambda})\).  The inverse map is
\[
        \mathbf x\longmapsto V^{\transpose}\mathbf x .
\]
\end{lemma}
\end{lemmabox}

\begin{proof}
Let \(\Phi(\boldsymbol{\lambda})=\boldsymbol{\lambda}\), and set
\(\mathbf x=\mathbf{x}(\boldsymbol{\lambda})\).  By
\eqref{eq:output-selector}, we have that
\(\boldsymbol{\lambda}=V^{\transpose}\mathbf x\).  Substituting this
identity into~\eqref{eq:circuit-complementarity} gives
\[
 A\mathbf x-UV^{\transpose}\mathbf x-\mathbf b
        =M\mathbf x-\mathbf b\geq0
\]
and the required complementary equalities.  Hence \(\mathbf x\)
solves~\eqref{eq:one-block-lcp}.

Conversely, let \(\mathbf x\) solve~\eqref{eq:one-block-lcp} and put
\(\boldsymbol{\lambda}=V^{\transpose}\mathbf x\).  Then
\(\lambda_\ell=x_{o_\ell}\).  At the \(o_\ell\)-coordinate,
\eqref{eq:output-anchor-identity} and the definition of \(\mathbf w\) give
\[
       w_{o_\ell}=x_{p_\ell}+x_{o_\ell}-1.
\]
If \(x_{o_\ell}=0\), then \(\lambda_\ell=0\).  If
\(x_{o_\ell}>0\), complementarity gives \(w_{o_\ell}=0\), and hence
\(x_{o_\ell}=1-x_{p_\ell}\leq1\).  Since \(\mathbf x\geq0\), this proves
\(\boldsymbol{\lambda}\in[0,1]^d\) directly from the anchor gates, before
we invoke the map on its declared domain.

Now
\[
 A\mathbf x-U\boldsymbol{\lambda}-\mathbf b
       =M\mathbf x-\mathbf b,
\]
so \((\boldsymbol{\lambda},\mathbf x)\) satisfies
\eqref{eq:circuit-complementarity}.  Lemma~\ref{lem:gate-normal-form}
implies that \(\mathbf x=\mathbf{x}(\boldsymbol{\lambda})\).
Equation~\eqref{eq:output-selector} now gives
\[
        \Phi(\boldsymbol{\lambda})
          =V^{\transpose}\mathbf{x}(\boldsymbol{\lambda})
          =V^{\transpose}\mathbf x
          =\boldsymbol{\lambda}.
\]
The same argument shows that the two maps are inverses.
\end{proof}

Strictness of the circuit becomes strict complementarity of the LCP.

\begin{lemmabox}
\begin{lemma}[Strict complementarity]
\label{lem:lcp-strict-complementarity}
Every solution of~\eqref{eq:one-block-lcp} is strictly complementary, i.e.,
\[
        x_i+w_i>0\qquad\text{for every }i\in[m].
\]
In particular, if
$I:=\{i\in[m]:x_i>0\}$,
then \(w_i=0\) precisely for \(i\in I\), while \(w_i>0\) for
\(i\notin I\).
\end{lemma}
\end{lemmabox}

\begin{proof}
By Lemma~\ref{lem:fixp-lcp-bijection}, a solution \(\mathbf x\) is the
vector of max-gate values obtained by evaluating the circuit at the fixed
point \(\boldsymbol{\lambda}=V^{\transpose}\mathbf x\).  At gate \(i\) we have $
        w_i=x_i-\ell_i.
$ 
Circuit-strictness gives \(\ell_i\ne0\).  If \(\ell_i>0\), then
\(x_i=\ell_i>0\) and \(w_i=0\).  If \(\ell_i<0\), then \(x_i=0\)
and \(w_i=-\ell_i>0\).
\end{proof}

\subsection{Preservation of the index}
\label{subsec:lcp-index}

Let \(\mathbf x\) be an LCP solution,
\(\boldsymbol{\lambda}=V^{\transpose}\mathbf x\) its corresponding
fixed point, and \(I=\operatorname{supp}(\mathbf x)\).  We set
\(k=|I|\).  For any square matrix \(Q\), \(Q_{II}\) means that we retain
exactly the rows and columns indexed by \(I\).  Thus
\(A_{II},M_{II}\in\mathbb Q^{k\times k}\) are principal submatrices.
The notation \(U_I,V_I\in\mathbb Q^{k\times d}\) means that only rows in
\(I\) are retained.  Finally, \(D\Phi(\boldsymbol{\lambda})\) denotes the
\(d\times d\) Jacobian of the active affine branch of \(\Phi\).

Because the circuit is strict, the same gates remain active in a
neighborhood of \(\boldsymbol{\lambda}\).  On this neighborhood,
\(\mathbf x_{[m]\setminus I}=0\), and the active gate equations are
\[
        A_{II}\mathbf x_I
             =U_I\boldsymbol{\lambda}+\mathbf b_I.
\]
Since \(A_{II}\) is lower triangular with diagonal equal to one, it is
invertible and
\[
        \mathbf x_I
           =A_{II}^{-1}
              (U_I\boldsymbol{\lambda}+\mathbf b_I).
\]
It follows from \(\Phi(\boldsymbol{\lambda})=V_I^{\transpose}\mathbf x_I\)
that
\begin{equation}
        D\Phi(\boldsymbol{\lambda})
             =V_I^{\transpose}A_{II}^{-1}U_I
             \in\mathbb Q^{d\times d}.
    \label{eq:jacobian-active-set}
\end{equation}

\begin{lemmabox}
\begin{lemma}[Determinant identity]
\label{lem:lcp-fixed-point-determinant}
For every solution \(\mathbf x\) of
\eqref{eq:one-block-lcp}, with corresponding fixed point
\(\boldsymbol{\lambda}=V^{\transpose}\mathbf x\) and active set
\(I=\operatorname{supp}(\mathbf x)\), it holds
\begin{equation*}
        \det M_{II}
          =\det\bigl(I_d-D\Phi(\boldsymbol{\lambda})\bigr).
    \label{eq:determinant-identity}
\end{equation*}

\noindent Moreover it holds
\[
 \sign\det M_{II}
   =\sign
       \det\bigl(I_d-D\Phi(\boldsymbol{\lambda})\bigr)
   =\operatorname{ind}_\Phi(\boldsymbol{\lambda}).
\]
\end{lemma}
\end{lemmabox}

\begin{proof}
By~\eqref{eq:def-M},
\[
        M_{II}=A_{II}-U_IV_I^{\transpose}.
\]
For \(B\in\mathbb R^{k\times d}\) and \(C\in\mathbb R^{d\times k}\), Sylvester's determinant identity gives 
$$\det (I_k - BC) = \det (I_d-CB).$$ Therefore it holds that
\begin{align*}
 \det M_{II}
   &=
   \det A_{II}\,
   \det\!\left(
       I_k-A_{II}^{-1}U_IV_I^{\transpose}
   \right)                                                    \\
   &=
   \det A_{II}\,
   \det\!\left(
       I_d-V_I^{\transpose}A_{II}^{-1}U_I
   \right).
\end{align*}
Every diagonal entry of \(A_{II}\) is one, so
\(\det A_{II}=1\). Substituting
\eqref{eq:jacobian-active-set} and the definition of the fixed point
index in~\eqref{eq:fixed-point-index},
$
\sign\det\bigl(I_d-D\Phi(\boldsymbol{\lambda})\bigr)
=\ind_\Phi(\boldsymbol{\lambda}),
$ proves the claim.
\end{proof}

We call an LCP solution
\emph{regular} when its active principal block $M_{II}$ is nonsingular.
The identity shows that this holds if and only if the corresponding fixed
point is regular.  To see explicitly that the LCP solution is isolated, use
strict complementarity.  If \(\mathbf x'\) is a sufficiently nearby LCP
solution, then \(x'_i>0\) for every \(i\in I\) and
\(w'_j>0\) for every \(j\notin I\).  Complementarity therefore forces
\(w'_I=0\) and \(\mathbf x'_{[m]\setminus I}=0\), so
\(M_{II}\mathbf x'_I=\mathbf b_I\).  Nonsingularity of \(M_{II}\) gives
\(\mathbf x'_I=\mathbf x_I\).  Thus every LCP solution produced here is
isolated.
\begin{lemmabox}
\begin{lemma}[Bit complexity and exact decoding]
\label{lem:lcp-bit-decoder}
Every solution of~\eqref{eq:one-block-lcp} is rational and has
encoding length polynomial in the size of \(\mathcal C\), and its fixed point
can be recovered in polynomial time using exact rational arithmetic.  If
\(\sign\det M_{II}=+1\), the certified decoder returns a valid solution of the
original \SoL\ instance.
\end{lemma}
\end{lemmabox}
\begin{proof}
Let \(I=\operatorname{supp}(\mathbf x)\).  Since \(w_I=0\) and
\(\mathbf x_{[m]\setminus I}=0\),
\begin{equation}
        M_{II}\mathbf x_I=\mathbf b_I.
    \label{eq:active-lcp-system}
\end{equation}
Lemma~\ref{lem:lcp-fixed-point-determinant} and regularity imply that
\(M_{II}\) is nonsingular.  Hence
\[
        \mathbf x_I=M_{II}^{-1}\mathbf b_I,\qquad
        \mathbf x_{[m]\setminus I}=0.
\]
The entries of \(M\) and \(\mathbf b\) have polynomial encoding length
by Lemma~\ref{lem:gate-normal-form}. Cramer's rule and Hadamard's inequality
give numerator and denominator bit lengths polynomial in \(m\) and in
the maximum input coefficient length.  Since \(m\) itself is
polynomial in the size of \(\mathcal C\), the claimed bound follows. Finally, given the exact solution vector, compute
\[
        \boldsymbol{\lambda}=V^{\transpose}\mathbf x
\]
and apply the certified decoder supplied by
Theorem~\ref{thm:axis-unconditional-map}.  By
Lemma~\ref{lem:lcp-fixed-point-determinant}, a positive sign of
\(\det M_{II}\) is exactly positive fixed point index, so that decoder returns
a valid \SoL\ output.

In fact, the values of \(\mathbf x\) need not be supplied if its support
\(I\) is known. We can solve
\eqref{eq:active-lcp-system}, check the complementary inequalities,
and then compute \(V^{\transpose}\mathbf x\).  All comparisons and
linear solves are exact and take polynomial time.
\end{proof}

\section{From complementarity to bimatrix games}
\label{sec:lcp-to-game}

In this section, we turn the complementarity problem constructed in
Section~\ref{sec:fixp-to-lcp} into a bimatrix game.  The reduction follows
the route of \citet{Mehta2018}. We first represent the complementarity
solutions as symmetric equilibria and then pass to an imitation game.  Two
refinements are needed for our purposes.  First, we keep track of the sign of
every solution.  Second, we perturb the final game into a globally
nondegenerate game without changing its equilibria at the level of supports.
The section has three steps: (a) establish an exact LCP equilibrium bijection,
(b) compute its index without a parity loss, and then (c) make the whole game
nondegenerate by an effective rational perturbation.

Recall the LCP
\begin{equation}
  \mathbf x\geq \zero,\qquad
  \mathbf w:=M\mathbf x-\mathbf b\geq \zero,\qquad
  x_iw_i=0\quad\text{for every }i\in[m].
  \label{eq:lcp-game-input}
\end{equation}
For a solution \(\mathbf x\), write
\[
  I(\mathbf x):=\{i\in[m]:x_i>0\}.
\]
By Lemmas~\ref{lem:lcp-strict-complementarity} and
\ref{lem:lcp-fixed-point-determinant}, every solution is strictly
complementary and satisfies
\[
  \det M_{I(\mathbf x),I(\mathbf x)}\neq 0.
\]
We will also use the output anchoring established in
Lemma~\ref{lem:output-anchoring}.  In the notation of
Section~\ref{sec:fixp-to-lcp},
\begin{equation}
  M=A-\sum_{\ell=1}^{d}\mathbf u^\ell e_{o_\ell}^{\transpose},
  \label{eq:anchored-lcp-form}
\end{equation}
where \(A\) is lower triangular with diagonal equal to one, the coordinates
\(o_1,\ldots,o_d\) are the distinguished circuit output coordinates, and
there is a predecessor coordinate \(p_\ell\) such that
\begin{equation}
  (M\mathbf z)_{o_\ell}=z_{p_\ell}+z_{o_\ell}
  \qquad\text{for every }\mathbf z\geq\zero
  \text{ and every }\ell\in[d].
  \label{eq:output-anchor-recalled}
\end{equation}
The only role of this last identity is to exclude equilibria supported
entirely on the original \(m\) strategies.  This is the same anchoring
argument as in Lemma~36 of \citet{Mehta2018}.
\subsection{The symmetric and imitation games}
\label{subsec:imitation-game}

We introduce one additional strategy, denoted by \(\star\), and set
\(N=m+1\).  We also define
\begin{equation}
  S:=
  \begin{pmatrix}
       -M & \mathbf b+\one\\
       \zero^{\transpose} & 1
  \end{pmatrix}
  \in\mathbb Q^{N\times N}.
  \label{eq:mehta-matrix}
\end{equation}
Consider the symmetric game \((S,S^{\transpose})\).  The strategy
\(\star\) is the homogenizing or anchoring coordinate. Once it receives
positive probability, dividing the first \(m\) probabilities by its
probability recovers the unnormalized LCP vector.  Lemma~\ref{lem:star-positive}
shows that the denominator is positive. 
For any mixed strategy \(y\), write
\[
 \BR_S(y):=\left\{i\in[N]:e_i^{\transpose}Sy
       =\max_{j\in[N]}e_j^{\transpose}Sy\right\}
\]
for its set of pure best responses in the row payoff matrix \(S\).

\begin{lemmabox}
\begin{lemma}[The anchoring strategy is used]
\label{lem:star-positive}
Every symmetric Nash equilibrium $(\mathbf y,\mathbf y)$ of \((S,S^{\transpose})\) with \(\mathbf y=(\widehat{\mathbf x},t)\)
 satisfies \(t>0\).
\end{lemma}
\end{lemmabox}

\begin{proof}
Suppose, towards contradiction, that \(t=0\).  Let \(\pi\) be the
equilibrium payoff.  Since the payoff of the pure strategy \(\star\) against
\(\mathbf y\) is zero, we have \(\pi\geq 0\).  Moreover, every coordinate
\(i\) with \(\widehat x_i>0\) satisfies
\[
    -(M\widehat{\mathbf x})_i=\pi.
\]
If \(\widehat x_{o_\ell}>0\), then
\eqref{eq:output-anchor-recalled} gives
\[
    -(M\widehat{\mathbf x})_{o_\ell}
       =-\widehat x_{p_\ell}-\widehat x_{o_\ell}<0,
\]
and hence $\pi<0$, contradicting \(\pi\geq0\).  Thus
\(\widehat x_{o_\ell}=0\) for every \(\ell\).  It follows from
\eqref{eq:anchored-lcp-form} that
\(M\widehat{\mathbf x}=A\widehat{\mathbf x}\).

Let \(i\) be the smallest index for which \(\widehat x_i>0\).  Since
\(A\) is lower triangular, its diagonal is one, and all earlier coordinates
of \(\widehat{\mathbf x}\) vanish,
\[
      (M\widehat{\mathbf x})_i
       =(A\widehat{\mathbf x})_i=\widehat x_i>0.
\]
Consequently, the payoff of strategy \(i\) is
\(-\widehat x_i<0\), again contradicting 
\(\pi\geq0\).  The proof is complete.
\end{proof}

\begin{lemmabox}
\begin{lemma}[LCP--symmetric equilibrium correspondence]
\label{lem:lcp-symmetric}
The solutions of \eqref{eq:lcp-game-input} are in one to one correspondence
with the symmetric Nash equilibria of \((S,S^{\transpose})\).  More
precisely, if \(\mathbf x\) solves \eqref{eq:lcp-game-input}, then setting
\begin{equation}
      \mathbf y
        :=\frac{(\mathbf x,1)}{1+\one^{\transpose}\mathbf x},
      \label{eq:lcp-to-symmetric}
\end{equation}
it follows $(\mathbf y,\mathbf y)$ is a symmetric Nash equilibrium.  Conversely, if
\(\mathbf y=(\widehat{\mathbf x},t)\) and $(\mathbf y,\mathbf y)$ is a symmetric Nash equilibrium,
then \(\widehat{\mathbf x}/t\) solves \eqref{eq:lcp-game-input}.
\end{lemma}
\end{lemmabox}

\begin{proof}
Let \(\mathbf x\) solve \eqref{eq:lcp-game-input}, let
    \(\mathbf w=M\mathbf x-\mathbf b\), and set
\(Z=1+\one^{\transpose}\mathbf x\).  Directly from
\eqref{eq:mehta-matrix},
\begin{equation}
    (S\mathbf y)_i=\frac{1-w_i}{Z}\quad(i\in[m]),
    \qquad
    (S\mathbf y)_\star=\frac{1}{Z}.
    \label{eq:symmetric-payoffs}
\end{equation}
If \(x_i>0\), complementarity gives \(w_i=0\), so strategy \(i\)
has the same payoff as \(\star\).  If \(x_i=0\), then \(w_i\geq0\),
so its payoff is no larger.  Hence every strategy in the support of
\(\mathbf y\) is a best response to \(\mathbf y\),
proving that \((\mathbf y,\mathbf y)\) is a Nash
equilibrium.

Conversely, let \(\mathbf y=(\widehat{\mathbf x},t)\) be a symmetric
equilibrium.  Lemma~\ref{lem:star-positive} gives \(t>0\).  Because
\(\star\) is in the support, the equilibrium payoff is \(t\).  The
best response conditions for the first \(m\) strategies are
\[
    -M\widehat{\mathbf x}+(\mathbf b+\one)t\leq t\one,
\]
or equivalently,
\[
    M\frac{\widehat{\mathbf x}}{t}-\mathbf b\geq\zero.
\]
Since every strategy used with positive probability must attain the
equilibrium payoff, the Nash equilibrium conditions also give for every $i$
\[
   \widehat x_i
   \bigl(-(M\widehat{\mathbf x})_i+b_it\bigr)=0.
\]
Dividing by \(t^2\) shows that
\(\mathbf x=\widehat{\mathbf x}/t\) satisfies
\eqref{eq:lcp-game-input}.  Finally, the normalization
\(\one^{\transpose}\widehat{\mathbf x}+t=1\) uniquely recovers
\((\widehat{\mathbf x},t)\) from \(\mathbf x\), and hence the
correspondence is one to one.
\end{proof}

Strict complementarity strengthens this correspondence in a way that will be
useful below.  If \(I=I(\mathbf x)\) and
\(K=I\cup\{\star\}\), then \eqref{eq:symmetric-payoffs} and
Lemma~\ref{lem:lcp-strict-complementarity} give
\begin{equation}
      \supp\mathbf y
      =\BR_{S}(\mathbf y)
      =K.
      \label{eq:quasi-strict-symmetric}
\end{equation}
An equilibrium is \emph{quasi-strict} when every pure best response is used
with positive probability, equivalently when its support equals its full
best response set.  Thus every symmetric equilibrium obtained above is
quasi-strict.

We now consider the imitation game
\begin{equation}
       \mathcal G=(S,I_N),
       \label{eq:imitation-game}
\end{equation}
where \(I_N\) is the \(N\times N\) identity matrix.  The correspondence
between symmetric equilibria and equilibrium strategies of an imitation game
goes back to \citet{McLennanTourky2010}.  In our strict setting, it gives a
correspondence between complete equilibrium profiles.

\begin{lemmabox}
\begin{lemma}[Full imitation game correspondence]
\label{lem:full-imitation}
Let \(\mathbf x\) be a solution of \eqref{eq:lcp-game-input}, let
\(I=I(\mathbf x)\), \(K=I\cup\{\star\}\).  Define the row
strategy coordinatewise by
\begin{equation}
       u_i:=
       \begin{cases}1/|K|,&i\in K,\\0,&i\notin K,\end{cases}
       \qquad(i\in[N]),
       \qquad
       \mathbf y
          :=\frac{(\mathbf x,1)}{1+\one^{\transpose}\mathbf x}.
       \label{eq:imitation-equilibrium}
\end{equation}
Then \((\mathbf u,\mathbf y)\) is a Nash
equilibrium of \(\mathcal G\).  Moreover, every Nash equilibrium of
\(\mathcal G\) is obtained uniquely in this way.
\end{lemma}
\end{lemmabox}

\begin{proof}
By \eqref{eq:quasi-strict-symmetric}, the row player's best responses to
\(\mathbf y\) are exactly the strategies in \(K\), and
\(\mathbf u\) is supported on \(K\).  Against
\(\mathbf u\), each column in \(K\) receives payoff \(1/|K|\),
while every column outside \(K\) receives payoff zero.  Since
\(\mathbf y\) is supported on \(K\), the profile in
\eqref{eq:imitation-equilibrium} is a Nash equilibrium.

For the converse, let \((\mathbf u,\mathbf y)\) be an arbitrary equilibrium
of \(\mathcal G\), and we define the set
\[
       T:=\arg\max_{j\in[N]}u_j.
\]
The column payoff of pure strategy \(j\) is \(u_j\).  Therefore 
$
  \supp\mathbf y\subseteq T.
$
Every coordinate in \(T\) is strictly positive, and hence
\(T\subseteq\supp\mathbf u\).  The row player's best response condition
gives the inclusion
\[
  \supp\mathbf y
       \subseteq T
       \subseteq\supp\mathbf u
       \subseteq\BR_{S}(\mathbf y).
\]
In particular, \(\supp\mathbf y\subseteq\BR_{S}(\mathbf y)\), so
\((\mathbf y,\mathbf y)\) is a symmetric equilibrium of
\((S,S^{\transpose})\).  By Lemma~\ref{lem:lcp-symmetric}, there is a unique LCP solution
\(\mathbf x\) corresponding to
\(\mathbf y=(\widehat{\mathbf x},t)\), namely
\(\mathbf x=\widehat{\mathbf x}/t\).  Equation~\eqref{eq:quasi-strict-symmetric} (due to strict complementarity) now gives
\[
      \BR_{S}(\mathbf y)=\supp\mathbf y=K.
\]
Thus every inclusion in the preceding chain is an equality.  The entries of
\(\mathbf u\) indexed by \(K\) are all equal to \(\max_j u_j\), while
every entry outside \(K\) is zero.  Since \(\mathbf u\) is a probability
vector,
\[
 u_i=\frac1{|K|}\quad(i\in K),\qquad
 u_i=0\quad(i\notin K).
\]
This proves both surjectivity and uniqueness.
\end{proof}

\begin{remark}[Why anchoring is necessary]
\label{rem:anchoring-necessary}
Strict complementarity and regularity of the LCP alone do not imply
Lemma~\ref{lem:star-positive}.  For example, take \(M=(-1)\) and
\(\mathbf b=(-1)\).  The LCP has the two solutions \(x=0\) and \(x=1\),
both strictly complementary and regular.  Nevertheless,
\[
    S=\begin{pmatrix}1&0\\0&1\end{pmatrix},
\]
and the first pure strategy is a symmetric equilibrium assigning zero
probability to \(\star\).  Thus the anchored structure in
Lemma~\ref{lem:output-anchoring} is a genuine part of the reduction.
\end{remark}

\subsection{Preservation of the index}
\label{subsec:lcp-game-index}

We next calculate the index of the equilibrium
\((\mathbf u(\mathbf x),\mathbf y(\mathbf x))\), using the bordered
determinant \(\Delta\) from \eqref{eq:bordered-determinant}.
For a regular equilibrium of a bimatrix game \((R,C)\) with row and
column supports \(I\) and \(J\), respectively, where \(|I|=|J|=k\),
Shapley's determinant formula gives
\begin{equation}
   \ind(\mathbf p,\mathbf q)
      =(-1)^{k+1}
       \sign\!\left(
          \Delta(R_{I,J})\,
          \Delta\bigl((C_{I,J})^{\transpose}\bigr)
       \right).
   \label{eq:shapley-index}
\end{equation}
The transpose in the second factor reflects that the column player's
support equations use \((C_{I,J})^{\transpose}\).  Our convention in
\eqref{eq:bordered-determinant} is chosen so that a pure equilibrium has
index \(+1\).

\begin{lemmabox}
\begin{lemma}[Exact index identity]
\label{lem:lcp-game-index}
Let \(\mathbf x\) solve \eqref{eq:lcp-game-input}, set
\(I=I(\mathbf x)\), and let
\((\mathbf u(\mathbf x),\mathbf y(\mathbf x))\) be its equilibrium in
the imitation game.  Then this equilibrium is regular and quasi-strict, and
\begin{equation}
       \ind(\mathbf u(\mathbf x),\mathbf y(\mathbf x))
          =\sign\det M_{I,I}.
       \label{eq:index-identity}
\end{equation}
\end{lemma}
\end{lemmabox}

\begin{proof}
Quasi-strictness was proved in \eqref{eq:quasi-strict-symmetric} and
Lemma~\ref{lem:full-imitation}.  Put \(k=|I|\), so that
\(K=I\cup\{\star\}\) has size \(k+1\).  We order \(K\) with
\(\star\) last.  The two support matrices are
\[
     S_{K,K}
       =
       \begin{pmatrix}
           -M_{I,I}&\mathbf b_I+\one\\
           \zero^{\transpose}&1
       \end{pmatrix},
     \qquad
     (I_N)_{K,K}^{\transpose}=I_{k+1}.
\]
Since \(M_{I,I}\mathbf x_I=\mathbf b_I\),
\begin{equation}
      S_{K,K}^{-1}\one
        =
        \begin{pmatrix}\mathbf x_I\\1\end{pmatrix}.
      \label{eq:inverse-times-one}
\end{equation}
Moreover,
\begin{equation}
      \det S_{K,K}
         =\det(-M_{I,I})
         =(-1)^k\det M_{I,I}.
      \label{eq:det-support-s}
\end{equation}
Applying the Schur complement formula to
\eqref{eq:bordered-determinant} and using
\eqref{eq:inverse-times-one}--\eqref{eq:det-support-s}, we obtain
\begin{align}
   \Delta(S_{K,K})
      &=
      \det S_{K,K}\,
      \one^{\transpose}S_{K,K}^{-1}\one \notag\\
      &=
      (-1)^k\det M_{I,I}
      \bigl(1+\one^{\transpose}\mathbf x_I\bigr),
      \label{eq:bordered-s}
\end{align}
where the last factor is strictly positive.  Similarly,
\begin{equation}
      \Delta(I_{k+1})
       =
       \det
       \begin{pmatrix}
          I_{k+1}&-\one\\
          \one^{\transpose}&0
       \end{pmatrix}
       =k+1>0.
      \label{eq:bordered-identity}
\end{equation}
Lemma~\ref{lem:lcp-fixed-point-determinant} and~\eqref{eq:bordered-s}
make the first bordered determinant nonzero, while
\eqref{eq:bordered-identity} makes the second positive.  Hence the
equilibrium is regular.
Finally, substituting \eqref{eq:bordered-s} and
\eqref{eq:bordered-identity} into \eqref{eq:shapley-index} gives
\[
 \begin{split}
   \ind(\mathbf u(\mathbf x),\mathbf y(\mathbf x))
      &=
       (-1)^{k+2}
       \sign\bigl((-1)^k\det M_{I,I}\bigr)\\
      &=
       (-1)^{k+2}(-1)^k\sign\det M_{I,I}\\
      &=\sign\det M_{I,I}.
 \end{split}
\]
The proof is complete.
\end{proof}

Combining Lemma~\ref{lem:lcp-game-index} with the determinant identity
proved in Lemma~\ref{lem:lcp-fixed-point-determinant}, we have
\begin{definitionbox}
\begin{equation*}
  \ind(\mathbf u(\mathbf x),\mathbf y(\mathbf x))
   =\sign\det M_{I,I}
   =\sign\det(I_d-D\Phi(\boldsymbol{\lambda})),
  \label{eq:three-indices}
\end{equation*}
\end{definitionbox}
where \(\boldsymbol{\lambda}\) is the fixed point represented by
\(\mathbf x\).  Thus the passage from the signed \LinFIXP instance to
the imitation game preserves the sign.

\subsection{A globally nondegenerate rational perturbation}
\label{subsec:global-perturbation}

Although every equilibrium of \(\mathcal G\) is regular, the game itself
may be degenerate away from its equilibrium set.  We next give a
deterministic perturbation that makes the whole game nondegenerate.  The
perturbation has polynomial encoding length and preserves every equilibrium
support and index. We choose an integer \(c>1+\max_{i,j}|S_{ij}|\) and define the strategically
equivalent game
\begin{equation}
     R:=S+c\one\one^{\transpose},
     \qquad
     C:=I_N
     \label{eq:positive-shift}
\end{equation}
Adding a constant to all the payoffs of one player does not change that
player's best responses.  The bordered determinants are unchanged as well.
Indeed, in the matrix defining
\(\Delta(D+c\one\one^{\transpose})\), adding \(c\) times the last column
to each of the first \(k\) columns removes
\(c\one\one^{\transpose}\) without changing the determinant.
As a result, \((R,C)\) has exactly the same equilibria and indices as
\((S,I_N)\).

Only the payoff matrix derived from \(S\) needs to be perturbed.  To see why,
let \(\mathbf p\) be any mixed strategy of the row player and let
\(K=\supp(\mathbf p)\).  Against \(\mathbf p\), the payoff of column \(j\)
under \(C\) is
\(p_j\).  Since \(\max_j p_j>0\), every best response belongs to \(K\), and
hence
\[
   |\BR_{C}(\mathbf p)|\leq |\supp(\mathbf p)|.
\]
Thus the nondegeneracy condition concerning the column player's best
responses already holds for \(C\), for every row strategy.

For \(0\leq t<1\), we define the sequence $\alpha_i(t):=1+t^i$ with $i\in[N],$ $\alpha(t) = (\alpha_i(t))_{i\in [N]}$ and let
\[
   D(t):=\operatorname{diag}
      (\alpha_1(t),\ldots,\alpha_N(t)).
\]
We perturb multiplicatively only the row payoff matrix
\begin{equation}
  \mathcal G(t):=(R(t),C),
  \qquad
  R(t):=D(t)^{-1}R.
  \label{eq:lex-game-perturbation}
\end{equation}
At \(t=0\), this gives \(\mathcal G(0)=(R,C)\), and
\(\mathcal G(t)\) converges to \(\mathcal G(0)\) as \(t\downarrow0\). The main claim of the section is the following lemma.

\begin{lemmabox}
\begin{lemma}[Effective perturbation]
\label{lem:effective-perturbation}
Let \(L_{\rm enc}\) be the binary encoding length of \((R,C)\).  One can
compute, in time polynomial in \(N+L_{\rm enc}\), an integer
$\kappa=\poly(N,L_{\rm enc})$
such that, for \(\varepsilon=2^{-\kappa}\), the following statements hold.
\begin{enumerate}
  \item The rational game \(\mathcal G(\varepsilon)\) is globally
  nondegenerate and has encoding length polynomial in \(N+L_{\rm enc}\).
  \item There is a one to one correspondence between the equilibria of
  \(\mathcal G(0)\) and those of \(\mathcal G(\varepsilon)\).
  Corresponding equilibria have the same row support, the same column
  support, and the same index.
  \item In particular, \(\mathcal G(\varepsilon)\) has no additional
  equilibria.
\end{enumerate}
\end{lemma}
\end{lemmabox}

\begin{proof}
We first clear the denominators of \(R\).  For a rational number \(a\), we
denote by \(\operatorname{den}(a)\) its positive denominator after all common
factors have been cancelled, and we set
\[
 \delta:=\prod_{i,j\in[N]}\operatorname{den}(R_{ij}),
 \qquad
 \overline R:=\delta R.
\]
Thus \(\delta\) and \(\overline R\) are integral and can be computed in
polynomial time.  We define
\begin{equation}
 h:=\left\lceil
   \log_2\!\left(2+\delta+\max_{i,j}|\overline R_{ij}|\right)
 \right\rceil,
 \qquad
 \ell_N:=\lceil\log_2(N+2)\rceil.
 \label{eq:perturbation-size-parameters}
\end{equation}
The quantity \(h\) bounds the binary length of \(\delta\) and the binary length of every entry of \(\overline R\).  Since
\(L_{\rm enc}\) is the total binary encoding length, we have
\(h\leq 2L_{\rm enc}+3\).

For an integer polynomial
\(f(t)=a_0+a_1t+\cdots+a_st^s\), we define
\[
  \operatorname{ht}(f)
   :=\left\lceil
      \log_2\left(1+\max_{0\leq r\leq s}|a_r|\right)
     \right\rceil.
\]
Thus \(\operatorname{ht}(f)\) measures, up to an additive constant, the
maximum number of bits needed to represent any coefficient of \(f\).

The proof proceeds in three parts, beginning with families of integer
polynomials that record every possible violation of row player nondegeneracy
and the conditions that determine an equilibrium on a given support.  We
then obtain uniform bounds on their degrees and coefficients and choose
\(\varepsilon=2^{-\kappa}\) so that every nonzero polynomial in these
families has a fixed sign on the relevant interval.  Finally, we use these
signs to prove global nondegeneracy and to match the equilibria of
\(\mathcal G(0)\) and \(\mathcal G(\varepsilon)\) while preserving their
supports and indices.

We begin with the condition concerning the row player's best responses.  The
condition concerning the column player's best responses was established
before the lemma.  It therefore remains to ensure that
\begin{equation}
  \bigl|\BR_{R(t)}(\mathbf y)\bigr|
  \leq \bigl|\supp(\mathbf y)\bigr|
  \qquad\text{for every }\mathbf y\in\Delta_N.
  \label{eq:row-best-response-nondegeneracy}
\end{equation}
Since \(R(t)\) is strictly positive, we consider the normalized best response
polytope
\begin{equation}
  \mathcal Q(t)
    =\{\mathbf q\geq\zero:R(t)\mathbf q\leq\one\}
    =\{\mathbf q\geq\zero:R\mathbf q
          \leq\boldsymbol\alpha(t)\}.
  \label{eq:perturbed-q}
\end{equation}
This polytope is bounded and full dimensional.  Indeed, strict positivity of
one row of \(R(t)\) gives an upper bound on every coordinate of
\(\mathbf q\), while every sufficiently small strictly positive vector
satisfies all the defining inequalities strictly.

We fix \(\mathbf y\in\Delta_N\), let
\(\mu=\max_i(R(t)\mathbf y)_i>0\), and set
\(\widehat{\mathbf q}=\mathbf y/\mu\).  Then
\(\widehat{\mathbf q}\in\mathcal Q(t)\).  Exactly
\(N-|\supp(\mathbf y)|\) of its nonnegativity inequalities hold with
equality.  Moreover, the payoff inequality indexed by \(i\) holds with
equality precisely when \(i\in\BR_{R(t)}(\mathbf y)\).  Hence the total
number of defining inequalities that hold with equality at
\(\widehat{\mathbf q}\) is
\[
  N-|\supp(\mathbf y)|+|\BR_{R(t)}(\mathbf y)|.
\]

For the converse, we take a point
\(\widehat{\mathbf q}\in\mathcal Q(t)\) that satisfies
more than \(N\) defining inequalities with equality.  At least one payoff
inequality must then be tight, and therefore
\(\max_i(R(t)\widehat{\mathbf q})_i=1\).  In particular,
\(\widehat{\mathbf q}\neq\zero\).  We define
\(\mathbf y=\widehat{\mathbf q}/(\one^{\transpose}
\widehat{\mathbf q})\).  Then \(\mathbf y\in\Delta_N\) and
\[
 \widehat{\mathbf q}
 =\frac{\mathbf y}{\max_i(R(t)\mathbf y)_i}.
\]
It follows that condition~\eqref{eq:row-best-response-nondegeneracy} fails if
and only if a feasible point of \(\mathcal Q(t)\) satisfies more than
\(N\) defining inequalities with equality.

We now assume that such a feasible point exists.  A vertex of the smallest
face of \(\mathcal Q(t)\) containing that point satisfies the same
equalities.  Since
\(\mathcal Q(t)\) is full dimensional, one can select \(N\) tight
inequalities at that vertex whose normal vectors are linearly independent,
and at least one additional defining inequality is also tight.  We will rule
out this possibility by considering every selection of \(N\) linearly
independent defining inequalities, whether or not the point determined by
the corresponding equalities is feasible.  We call such a selection a
basis.

We now consider a basis of \(\mathcal Q(t)\).  We use \(Z\) to index its
selected nonnegativity inequalities and \(I\) to index its selected payoff
inequalities, and we set \(F=[N]\setminus Z\).  Because the basis contains exactly \(N\)
inequalities,
\[
  |Z|+|I|=N
  \qquad\text{and hence}\qquad
  |F|=|I|.
\]
The selected equalities determine a point, which need not be feasible.  If
\(F=I=\varnothing\), this point is \(\mathbf q=\zero\), where every payoff
inequality is strict.  Thus no additional equality is possible in this case,
and we may assume that \(|F|=|I|\geq1\).

To characterize linear independence of the selected inequalities, order the
coordinates with those in \(Z\) first and those in \(F\) second, and place
the selected nonnegativity inequalities before the selected payoff
inequalities.  Up to row and column permutations and multiplication of rows
by \(-1\), the matrix of their normal vectors is
\[
 \begin{pmatrix}
   I_{|Z|} & 0\\
   \overline R_{I,Z} & \overline R_{I,F}
 \end{pmatrix}.
\]
This matrix is block lower triangular.  Its upper left block is an identity
matrix, so the selected normal vectors are linearly independent if and only
if
$
       B:=\overline R_{I,F}
$
is nonsingular.

The selected nonnegativity equalities give \(\mathbf q_Z=\zero\).  After the
selected payoff equalities are multiplied by \(\delta\), they give
\(\overline R_{I,F}\mathbf q_F
=\delta\boldsymbol\alpha_I(t)\).  Therefore, the point determined by the
basis satisfies
\begin{equation}
      \mathbf q_Z=\zero,
      \qquad
      \mathbf q_F=B^{-1}\delta\boldsymbol\alpha_I(t).
      \label{eq:q-basis-solution}
\end{equation}

We first consider an unselected payoff inequality indexed by \(j\notin I\).  If
this inequality also held with equality at the point in
\eqref{eq:q-basis-solution}, then
\(\overline R_{j,F}B^{-1}\delta\boldsymbol\alpha_I(t)
=\delta\alpha_j(t)\).  Multiplication by \(\det B\) gives
\begin{equation}
  \overline R_{j,F}\operatorname{adj}(B)\,
       \delta\boldsymbol\alpha_I(t)
  -\delta\det(B)\alpha_j(t)=0.
  \label{eq:extra-payoff-polynomial}
\end{equation}
The first term on the left hand side contains only a constant term and powers
\(t^i\) with \(i\in I\).  Since \(j\notin I\), it has no term in
\(t^j\).  The coefficient of \(t^j\) in the second term is
\(-\delta\det(B)\neq0\).  Hence the polynomial on the left hand side is not
identically zero.

We next consider an unselected nonnegativity inequality indexed by \(j\in F\).
It holds with equality precisely when \(q_j=0\).  By
\eqref{eq:q-basis-solution}, this is equivalent to
\[
 \bigl(\operatorname{adj}(B)\delta
       \boldsymbol\alpha_I(t)\bigr)_j=0.
\]
Since \(B\) is nonsingular, the \(j\)th row of
\(\operatorname{adj}(B)\) is nonzero.  Thus
\(\operatorname{adj}(B)_{j,i}\neq0\) for some \(i\in I\).  In the
expansion
\[
 \bigl(\operatorname{adj}(B)\delta
       \boldsymbol\alpha_I(t)\bigr)_j
 =\delta\sum_{i\in I}
   \operatorname{adj}(B)_{j,i}(1+t^i),
\]
the powers \(t^i\), \(i\in I\), are distinct.  Therefore a nonzero
coefficient of one of these powers cannot be cancelled by another term, and
\(\bigl(\operatorname{adj}(B)\delta
\boldsymbol\alpha_I(t)\bigr)_j\) is not identically zero.

The family \(\mathcal I\) consists of the polynomials obtained in these two
ways from all possible bases and all possible additional defining
equalities.  Every member of \(\mathcal I\) is a nonzero integer polynomial.
There are at most \(\binom{2N}{N}\) choices of a basis and at most \(N\)
unselected inequalities for each choice.  Thus
\(|\mathcal I|\leq N\binom{2N}{N}\).  By construction, if \(t>0\) is not a
root of any polynomial in \(\mathcal I\), no basis point satisfies an
additional defining equality.  In particular, no feasible point of
\(\mathcal Q(t)\) satisfies more than \(N\) defining equalities.  Therefore
condition~\eqref{eq:row-best-response-nondegeneracy} holds.

We next construct the polynomials that control the equilibria, leaving the
basis calculation behind.  For now, we fix \(t>0\) that is not a root of any
polynomial in \(\mathcal I\), and we take an equilibrium
\((\mathbf p,\mathbf q)\) of \(\mathcal G(t)\).  We set
\(I=\supp(\mathbf p)\) and \(J=\supp(\mathbf q)\).  Every column in \(J\)
is a best response to \(\mathbf p\), while the direct argument for \(C\)
given before the lemma shows that every column best response belongs to
\(I\).  Hence
\[
       J\subseteq\BR_C(\mathbf p)\subseteq I.
\]
Similarly, because \(\mathbf p\) is a best response to \(\mathbf q\), every
row in \(I\) belongs to \(\BR_{R(t)}(\mathbf q)\).  Applying
\eqref{eq:row-best-response-nondegeneracy} to the column strategy
\(\mathbf q\) gives
\[
       |I|
       \leq |\BR_{R(t)}(\mathbf q)|
       \leq |\supp(\mathbf q)|
       =|J|.
\]
Together with \(J\subseteq I\), these inequalities show that all the
inclusions and cardinality inequalities are equalities.  Therefore
\[
 I=J=:K,\qquad
 \BR_C(\mathbf p)=K,\qquad
 \BR_{R(t)}(\mathbf q)=K.
\]
In particular, every row inequality outside the support is strict.  Moreover, since
\(C=I_N\), all coordinates \(p_i\), \(i\in K\), must be equal.  Thus, with
\(k=|K|\),
\begin{equation}
       \mathbf p_K=\frac1k\one_k,
       \qquad
       \mathbf p_{[N]\setminus K}=\zero.
       \label{eq:uniform-row-strategy}
\end{equation}
Every column in \(K\) then has payoff \(1/k\), whereas every column outside
\(K\) has payoff \(0\).  Hence all column inequalities outside the support
are strict.

It remains to control the column strategy and the row player's payoff
inequalities.  For each nonempty \(K\subseteq[N]\), we set \(k=|K|\) and define 
\begin{align}
 \mathcal R_K(t)
   &:=
   \begin{pmatrix}
      R(t)_{K,K}&-\one_k\\
      \one_k^{\transpose}&0
   \end{pmatrix},
 &
 \widehat{\mathcal R}_K(t)
   &:=
   \begin{pmatrix}
      \overline R_{K,K}&-\delta\boldsymbol\alpha_K(t)\\
      \one_k^{\transpose}&0
   \end{pmatrix}.
 \label{eq:one-sided-support-systems}
\end{align}
The equal payoff equations for the rows in \(K\), together with the
normalization of \(\mathbf q_K\), are
\[
 \mathcal R_K(t)
 \begin{pmatrix}\mathbf q_K\\ \rho\end{pmatrix}
 =\mathbf e_{k+1},
\]
where \(\rho\) is the row player's equilibrium payoff.  Multiplying the
first \(k\) equations by the positive numbers
\(\delta\alpha_i(t)\), \(i\in K\), gives the equivalent system 
\[
 \widehat{\mathcal R}_K(t)
 \begin{pmatrix}\mathbf q_K\\ \rho\end{pmatrix}
 =\mathbf e_{k+1}.
\]

\noindent We define
\begin{equation}
 r_K(t):=\det\widehat{\mathcal R}_K(t)
   =\delta^k\!\prod_{i\in K}\alpha_i(t)\,
      \det\mathcal R_K(t).
 \label{eq:row-support-det}
\end{equation}
Using Cramer's rule, whenever \(r_K(t)\neq0\), we have
\[
       q_j(t)=\frac{q_{j;K}(t)}{r_K(t)},
       \qquad
       \rho(t)=\frac{\rho_K(t)}{r_K(t)}.
\]
Here \(q_{j;K}(t)\) and \(\rho_K(t)\) denote the determinants obtained from
\(\widehat{\mathcal R}_K(t)\) by replacing, respectively, the column
corresponding to \(q_j\) and the last column with \(\mathbf e_{k+1}\). 
For every \(i\notin K\), we define
\begin{equation}
 s_{i;K}(t):=
    \delta\alpha_i(t)\rho_K(t)
      -\sum_{j\in K}\overline R_{ij}q_{j;K}(t).
 \label{eq:cleared-row-slack}
\end{equation}
Then
\[
  \rho(t)-(R(t)\mathbf q(t))_i
   =\frac{s_{i;K}(t)}
          {\delta\alpha_i(t)r_K(t)}.
\]
Since \(\alpha_i(t)>0\), the candidate determined by \(K\) has
positive probabilities and all the row inequalities outside \(K\) are strict
precisely when
\[
       r_K(t)q_{j;K}(t)>0\quad(j\in K)
       \;\textrm{ and }\;
       r_K(t)s_{i;K}(t)>0\quad(i\notin K).
\]

We define
\begin{align*}
 \mathcal D&:=\{r_K:\varnothing\neq K\subseteq[N]\} \textrm{ with }|\mathcal D|
   \leq 2^N-1,\\
 \mathcal P_{\rm prob}
   &:=\{r_Kq_{j;K}:\varnothing\neq K\subseteq[N],\ j\in K\} \textrm{ with } |\mathcal P_{\rm prob}|
   \leq N2^{N-1},\\
 \mathcal S 
   &:=\{r_Ks_{i;K}:\varnothing\neq K\subseteq[N],\ i\notin K\} \textrm{ with } |\mathcal S|
   \leq N\bigl(2^{N-1}-1).
\end{align*}
The four families have distinct roles. The family \(\mathcal I\) 1 global nondegeneracy for positive \(t\).
The remaining three families track an equilibrium with a fixed support
\(K\).  By~\eqref{eq:row-support-det}, the positive factor relating
\(r_K(t)\) to
\(\Delta(R(t)_{K,K})=\det\mathcal R_K(t)\) implies that
\(\mathcal D\) controls both nonsingularity of the support system and the
sign of the row bordered determinant.  The family
\(\mathcal P_{\rm prob}\) controls positivity of the probabilities on \(K\),
while \(\mathcal S\) controls strictness of the row payoff inequalities
outside \(K\).  Preserving these signs ensures that the same support
continues to define an equilibrium.  The sign information in
\(\mathcal D\) will also be used in the final comparison of the equilibrium
indices. Also notice that the column bordered
determinant is independent of \(t\).  Since \(C=I_N\),
\begin{equation}
 \Delta\bigl((C_{K,K})^{\transpose}\bigr)
   =\Delta(I_k)=k>0.
 \label{eq:constant-column-bordered-det}
\end{equation}

We now derive uniform degree and coefficient bounds.  We set
\begin{equation}
 H:=(N+1)(h+\ell_N)+1,
 \qquad
 H_\ast:=2H+h+2\ell_N,
 \qquad
 d_\ast:=2N.
 \label{eq:degree-height-bounds}
\end{equation} 

\noindent First, we bound the degrees of the relevant polynomials.  Each polynomial in
\(\mathcal I\), each \(r_K\), and each Cramer numerator \(q_{j;K}\) and
\(\rho_K\) is the determinant of a square matrix of dimension at most
\(N+1\).  In each such matrix, at most one column depends on \(t\).  If a
column depends on \(t\), expanding the determinant along that column
expresses the determinant as a linear combination of its entries with
coefficients independent of \(t\).  Every entry in this column has degree at
most \(N\).  Hence every such determinant has degree at most \(N\).

The numerator \(\rho_K\) is independent of \(t\), because its Cramer matrix
replaces the only column containing
\(\boldsymbol\alpha_K(t)\) with \(\mathbf e_{k+1}\).  It follows that
\(s_{i;K}\) has degree at most \(N\).  Since every polynomial in
\(\mathcal P_{\rm prob}\) and \(\mathcal S\) is a product of two
polynomials of degree at most \(N\), we have

$$
 \deg f\leq N
 \quad\text{for }f\in\mathcal I\cup\mathcal D,
 \qquad
 \deg f\leq2N
 \quad\text{for }f\in\mathcal P_{\rm prob}\cup\mathcal S.
$$

\noindent Thus \(d_\ast=2N\) is a common degree bound for all four families.

We next bound their coefficients.  Every matrix entry is either constant or
has the form \(c_0+c_1t^i\) for some \(i\in[N]\), where each coefficient has
absolute value less than \(2^h\).  The determinant expansion of an
\(r\times r\) matrix contains \(r!\) terms, one for each permutation.  Each
such term is a product of \(r\) matrix entries.  The absolute value of each
resulting product of coefficients is therefore less than
$
   (2^h)^r=2^{rh}.
$

At most one factor in each product depends on \(t\), and that factor contains
at most two terms.  Hence every term in the determinant expansion produces
at most two polynomial terms.  It follows that every coefficient of the
determinant has absolute value less than $$
   2r!2^{rh} \leq 2^{r(h+\ell_N)+1}.
$$
The definition of \(H\) therefore gives
$
 \operatorname{ht}(f)\leq H
$
for every \(f\in\mathcal I\cup\mathcal D\) and for every Cramer numerator
\(q_{j;K}\) and \(\rho_K\). Moreover, we get that

\[
\begin{aligned}
 \operatorname{ht}(s_{i;K})
   &\leq H+h+\ell_N,\\
 \operatorname{ht}(r_Kq_{j;K})
   &\leq2H+\ell_N,\\
 \operatorname{ht}(r_Ks_{i;K})
   &\leq2H+h+2\ell_N.
\end{aligned}
\]

\noindent By the definition of \(H_\ast\), these inequalities imply

$$
 \operatorname{ht}(f)\leq H_\ast
 \qquad\text{for every }
 f\in\mathcal I\cup\mathcal D\cup
 \mathcal P_{\rm prob}\cup\mathcal S.
$$

\noindent Thus \(d_\ast\) and \(H_\ast\) are common bounds on the degree and on
\(\operatorname{ht}(f)\), respectively.  Both bounds are polynomial in
\(N+L_{\rm enc}\).

We can now choose \(\kappa\).  We consider any nonzero integer polynomial
\(f(t)=\sum_{\nu=0}^{d_0}a_\nu t^\nu\)
with \(d_0\leq d_\ast\) and \(\operatorname{ht}(f)\leq H_\ast\), and we
denote by \(\nu_0\) the least index for which \(a_{\nu_0}\neq0\).  For
\(0<t<1\),
\[
 |a_{\nu_0}|\geq1,
 \qquad
 \left|\sum_{\nu>\nu_0}a_\nu t^\nu\right|
   \leq d_\ast2^{H_\ast}t^{\nu_0+1}.
\]
We define
\begin{equation}
 \kappa:=
   H_\ast+
   \left\lceil\log_2\bigl(4(d_\ast+1)\bigr)\right\rceil+2,
 \qquad
 \varepsilon:=2^{-\kappa}.
 \label{eq:choice-epsilon}
\end{equation}
For \(0<t\leq\varepsilon\), the choice of \(\kappa\) gives
\[
 d_\ast2^{H_\ast}t
 \leq\frac{d_\ast}{16(d_\ast+1)}<1.
\]
Hence the lowest degree nonzero term of \(f\)
has strictly larger absolute value than the sum of all remaining terms.
Therefore \(f(t)\) has the sign of \(a_{\nu_0}\) throughout
\((0,\varepsilon]\), and in particular it has no root in this interval.  If
\(f(0)\neq0\), then \(\nu_0=0\), so the sign of \(f\) is constant on the
closed interval \([0,\varepsilon]\).

Every member of \(\mathcal I\) is nonzero, so no member of
\(\mathcal I\) vanishes for \(0<t\leq\varepsilon\).  The argument leading
to~\eqref{eq:row-best-response-nondegeneracy} now shows that the condition
concerning the row player's best responses holds throughout this interval.
The condition concerning the column player's best responses holds for \(C\)
by the direct argument preceding the lemma.  Therefore \(\mathcal G(t)\) is
globally nondegenerate for every \(0<t\leq\varepsilon\).

It remains to compare the equilibria.  Every equilibrium of
\(\mathcal G(0)\) is regular and quasi strict by
Lemmas~\ref{lem:full-imitation} and~\ref{lem:lcp-game-index}.  We denote the
common row and column support of such an equilibrium by \(K\).  Regularity gives
\(r_K(0)\neq0\).  Positivity of its column probabilities gives
\(r_K(0)q_{j;K}(0)>0\) for every \(j\in K\), and
the quasi strict property gives \(r_K(0)s_{j;K}(0)>0\) for every \(j\notin K\).
These polynomials are therefore nonzero, and the choice of
\(\varepsilon\) preserves all their signs on \([0,\varepsilon]\).  Thus,
for every \(t\in[0,\varepsilon]\), the support system has a unique solution
with positive probabilities and strict row inequalities outside the support.  With
the row strategy in~\eqref{eq:uniform-row-strategy}, this solution is an
equilibrium with support pair \((K,K)\).  Hence no equilibrium of
\(\mathcal G(0)\) is destroyed.

For the converse, we take an equilibrium of
\(\mathcal G(t)\) with support pair \((K,K)\) for some
\(t\in(0,\varepsilon]\).  Since
\(\mathcal G(t)\) is globally nondegenerate, its equilibrium support system
is nonsingular, and therefore \(r_K(t)\neq0\).  Its support probabilities
are positive, and the equality
\(\BR_{R(t)}(\mathbf q)=K\) proved above shows that every row inequality
outside the support is strict.  Thus \(r_K(t)\) is nonzero, while the
corresponding members of \(\mathcal P_{\rm prob}\) and \(\mathcal S\) are
positive.  All these polynomials are nonzero, and their signs are constant
throughout \((0,\varepsilon]\).  It follows that the
same support pair determines an equilibrium of \(\mathcal G(s)\) for every
\(s\in(0,\varepsilon]\).

We choose a sequence \(t_s\to0\).  The corresponding row strategy is the
uniform distribution on \(K\) for every \(s\).  Compactness of the column
player's simplex gives a subsequence of the column strategies converging to
some \(\mathbf q^0\).  Since the payoff matrices and the Nash inequalities
vary continuously with \(t\), the limit
\((\mathbf p^0,\mathbf q^0)\) is an equilibrium of \(\mathcal G(0)\), where
\(\mathbf p^0\) is uniform on \(K\).  Since \(C=I_N\), we have
\(\BR_C(\mathbf p^0)=K\).  The limiting equilibrium is quasi strict, and
therefore
\[
       \supp(\mathbf q^0)=\BR_C(\mathbf p^0)=K.
\]
Thus the limiting equilibrium has support pair \((K,K)\).

Regularity of the limiting equilibrium makes its support system
nonsingular, so there is at most one equilibrium of \(\mathcal G(0)\) with
support pair \((K,K)\).  Likewise, for \(t>0\), the equality
\(r_K(t)\neq0\) and the fact that the row strategy is uniform on \(K\) imply
that there is at most one equilibrium of \(\mathcal G(t)\) with this support
pair.  Matching the common supports therefore gives a bijection between the
equilibria of \(\mathcal G(0)\) and those of
\(\mathcal G(\varepsilon)\).  In particular, no additional equilibrium is
created.

Finally, we consider corresponding equilibria with support pair \((K,K)\).  By
\eqref{eq:constant-column-bordered-det}, the column bordered determinant is
the fixed positive number \(k=|K|\).  By~\eqref{eq:row-support-det}, the row
bordered determinant has the same sign as \(r_K(t)\).  Since
\(r_K(0)\neq0\), this sign is constant on \([0,\varepsilon]\).  The support
size is also unchanged, so Shapley's formula shows that the equilibrium
index is unchanged.

It remains to verify the encoding length.  Since
\(\varepsilon^i=2^{-\kappa i}\) and \(i\leq N\), every entry of
\(R(\varepsilon)\) has binary length
\(\mathcal O(L_{\rm enc}+N\kappa)\), while \(C\) is unchanged.  Summed over all
entries, the encoding length is \(\mathcal O(L_{\rm enc}+N^3\kappa)\), which is
polynomial in \(N+L_{\rm enc}\).  This completes the proof.
\end{proof}

\subsection{Exact decoding and the hardness theorem}
\label{subsec:exact-decoding}

The perturbation changes the probabilities in an equilibrium, so one cannot
decode by simply rescaling the column player's strategy.  The support,
however, contains all the required exact information.

\begin{lemma}[Support decoder]
\label{lem:support-decoder}
Let \((\mathbf p,\mathbf q)\) be an exact Nash equilibrium of
\(\mathcal G(\varepsilon)\).  Its row and column supports are equal to a
set
\[
       K=I\cup\{\star\}.
\]
The corresponding solution of the original LCP is recovered exactly by
\begin{equation}
      \mathbf x_I=M_{I,I}^{-1}\mathbf b_I,
      \qquad
      \mathbf x_{[m]\setminus I}=\zero.
      \label{eq:support-decoder}
\end{equation}
This decoding takes polynomial time.
\end{lemma}

\begin{proof}
By Lemma~\ref{lem:effective-perturbation}, \((\mathbf p,\mathbf q)\) has
the same support pair as a unique equilibrium of \(\mathcal G(0)\).  By
Lemma~\ref{lem:full-imitation}, that support pair is \((K,K)\), where
\(K=I(\mathbf x)\cup\{\star\}\) for a unique LCP solution
\(\mathbf x\).  On the active coordinates, complementarity gives
\(M_{I,I}\mathbf x_I=\mathbf b_I\), and
Lemma~\ref{lem:lcp-fixed-point-determinant} makes \(M_{I,I}\)
nonsingular.  Equation~\eqref{eq:support-decoder} therefore recovers
\(\mathbf x\) uniquely.  The solve is over a rational matrix of polynomial
dimension and encoding length, and hence takes polynomial time.
\end{proof}

We can now complete the reduction.

\begin{theorembox}
\begin{theorem}[\PPADS-hardness]
\label{thm:positive-index-hardness}
There is a polynomial-time reduction from \SoL\ to the following problem:
given a rational nondegenerate bimatrix game, find a Nash equilibrium of
index \(+1\).  In particular, every output instance of the reduction
satisfies the nondegeneracy promise.
\end{theorem}
\end{theorembox}
\begin{proof}
Given a normalized \SoL\ instance,
Theorem~\ref{thm:axis-unconditional-map} provides a rational
max-affine map and a circuit that is strict at every fixed point. Its
positive index fixed points
correspond precisely to the sinks,
and Section~\ref{sec:fixp-to-lcp} constructs the LCP
\eqref{eq:lcp-game-input}.  Its solutions are strictly complementary and
regular, satisfy the output anchoring of Lemma~\ref{lem:output-anchoring},
and are in one to one correspondence with the fixed points of the map.
For a solution with active set \(I\),
Lemma~\ref{lem:lcp-fixed-point-determinant} gives
\[
      \sign\det M_{I,I}
       =\sign\det(I_d-D\Phi(\boldsymbol{\lambda})).
\]

We construct the imitation game \((S,I_N)\), shift the row payoffs to obtain the game
\((R,I_N)\), and apply Lemma~\ref{lem:effective-perturbation}.  The resulting game
\(\mathcal G(\varepsilon)\) is rational, globally
nondegenerate, and has polynomial encoding length.  If
\((\mathbf p,\mathbf q)\) is an equilibrium of index \(+1\),
Lemma~\ref{lem:effective-perturbation} and Lemma \ref{lem:support-decoder}, and
\ref{lem:lcp-game-index} give
\[
 \begin{split}
   +1
      &=\ind(\mathbf p,\mathbf q)\\
      &=\sign\det M_{I,I}\\
       &=\sign\det(I_d-D\Phi(\boldsymbol{\lambda})).
 \end{split}
\]
Thus \(\boldsymbol{\lambda}\) is a positive index fixed point and corresponds
to a sink.  Lemma~\ref{lem:support-decoder} recovers the exact LCP solution from
the equilibrium support, and Theorem~\ref{thm:axis-unconditional-map}
recovers the corresponding \SoL\ output in polynomial time.
\end{proof}

\begin{proof}[Proof of Theorem~\ref{thm:main-intro}]
Membership is Theorem~\ref{thm:membership}, and hardness is
Theorem~\ref{thm:positive-index-hardness}.  The former is a
promise preserving reduction on nondegenerate games, while the latter always
outputs a rational globally nondegenerate game, so together they prove the
stated completeness.
\end{proof}

\subsection{Allowing nonregular equilibria}
\label{subsec:positive-or-nonregular}

We prove the extension to arbitrary bimatrix games stated in
Corollary~\ref{cor:positive-or-nonregular}.  The reduction below constructs
a rational nondegenerate perturbation of the input game and converts
every positive-index equilibrium of that perturbation into an exact
equilibrium of the input game.  The decoded equilibrium may have smaller
supports, but if it is regular, its index is still \(+1\).


\paragraph{Roadmap for the promise-free membership argument.}
The hardness direction needs no modification: the reduction of
Theorem~\ref{thm:positive-index-hardness} already outputs globally
nondegenerate games.  The only new issue is membership for an arbitrary,
possibly degenerate, input game.  The argument has four ingredients.
First, we perturb both payoff matrices by distinct powers of one scalar
parameter~\(t\).  In the complementary formulation this leaves the
constraint matrix fixed and moves only the right-hand side.  Second, we
choose one explicit rational value \(\varepsilon>0\) for which every basis
of the resulting linear system has a fixed sign pattern on
\((0,\varepsilon]\).  This makes every perturbed game in that interval
globally nondegenerate.  Third, we apply the signed Lemke--Howson
membership construction at \(t=\varepsilon\) and retain the complementary
basis of the returned positive-index equilibrium.  Finally, we keep this
same basis while sending \(t\) to zero.  Its basic variables are rational
polynomials in \(t\); their signs cannot change for \(t>0\), and evaluating
them at \(t=0\) gives an exact equilibrium of the original game.

The last step is where allowing a nonregular output becomes essential.  A
basic variable that is positive for every \(t>0\) may become zero exactly
at \(t=0\).  If this causes a supported strategy to disappear while it
remains a best response, the decoded equilibrium is nonregular and is
therefore already an accepted output.  Otherwise the support does not
shrink, and the positive diagonal scalings used in the perturbation
preserve the determinant signs in Shapley's formula, so the decoded
regular equilibrium has index \(+1\).

\begin{proof}[Proof of Corollary~\ref{cor:positive-or-nonregular}]
Hardness follows from Theorem~\ref{thm:positive-index-hardness}: its output games are globally nondegenerate, so all their equilibria are regular.
Therefore, allowing nonregular equilibria introduces no additional
valid outputs on those instances.

For membership, we give a polynomial-time reduction to the
nondegenerate problem covered by Theorem~\ref{thm:membership}.

\paragraph{Positive payoffs.}
Let \(A,B\in\mathbb Q^{m\times n}\) be the input matrices and let
\(d=m+n\).  We add
\[
 c_A=1+\max_{i,j}|A_{ij}|,
 \qquad
 c_B=1+\max_{i,j}|B_{ij}|
\]
to every entry of the respective payoff matrices.  These shifts have
polynomial encoding length and make every payoff strictly positive.
They preserve all best-response comparisons, hence all equilibria and
all strict off-support inequalities.  They also preserve the bordered
determinants, since
$\Delta(M+c\one\one^{\transpose})=\Delta(M)$.
It follows that the shifts
preserve regularity and the indices of regular equilibria.
We continue to denote the shifted, strictly positive matrices by
\(A,B\).

\paragraph{The perturbation and its complementarity formulation.}
For \(0\leq t<1\), define
\[
 \alpha_i(t)=1+t^i\quad(i\in[m]),
 \qquad
 \beta_j(t)=1+t^{m+j}\quad(j\in[n]),
\]
and set
\begin{equation}\label{eq:general-game-perturbation}
 \begin{aligned}
 D_A(t)&=\operatorname{diag}(\alpha_1(t),\ldots,\alpha_m(t)),\\
 D_B(t)&=\operatorname{diag}(\beta_1(t),\ldots,\beta_n(t)),\\
 A(t)&=D_A(t)^{-1}A,
 & B(t)&=BD_B(t)^{-1}.
 \end{aligned}
\end{equation}
Write \(\mathcal G(t)=(A(t),B(t))\).  In particular,
\(\mathcal G(0)=(A,B)\).

Let
\[
 C=\begin{pmatrix}0&A\\ B^{\transpose}&0\end{pmatrix},
 \qquad
 q(t)=\binom{\alpha(t)}{\beta(t)}
     =\one+(t,t^2,\ldots,t^d)^{\transpose},
 \qquad
 z=\binom{u}{v},
\]
where \(u\in\mathbb R^m\) and \(v\in\mathbb R^n\).  The matrix \(C\)
is independent of \(t\).  Now consider
\begin{equation}\label{eq:general-game-complementarity}
 z\geq\zero,\qquad
 s=q(t)-Cz\geq\zero,\qquad
 z_rs_r=0\quad(r\in[d]).
\end{equation}
Its feasible polytope, before imposing the product equalities, is
\[
 \mathcal R(t)
 =\{z\geq\zero:Cz\leq q(t)\}
 =\{u\geq\zero:B^{\transpose}u\leq\beta(t)\}
  \times
  \{v\geq\zero:Av\leq\alpha(t)\}.
\]

We verify directly the correspondence with Nash equilibria.
In any nonzero solution of
\eqref{eq:general-game-complementarity}, both \(u\) and \(v\) are
nonzero.  Indeed, \(u=\zero\) would make every column-player slack
equal to \(\beta_j(t)>0\), forcing \(v=\zero\); the converse follows
in the same way.  Thus
\[
 U=\one^{\transpose}u>0,\qquad
 V=\one^{\transpose}v>0,\qquad
 x=u/U,\qquad y=v/V
\]
are well defined.  Dividing the row inequalities by
\(\alpha_i(t)V\) gives
\[
 (A(t)y)_i\leq 1/V,
 \qquad
 (A(t)y)_i=1/V\quad\text{whenever }x_i>0.
\]
Likewise,
\[
 (B(t)^{\transpose}x)_j\leq 1/U,
 \qquad
 (B(t)^{\transpose}x)_j=1/U
       \quad\text{whenever }y_j>0.
\]
Each player therefore assigns positive probability only to best
responses, so \((x,y)\) is a Nash equilibrium of \(\mathcal G(t)\).

Conversely, let \((x,y)\) be an equilibrium of \(\mathcal G(t)\)
with payoffs
\[
 a=x^{\transpose}A(t)y>0,
 \qquad
 b=x^{\transpose}B(t)y>0.
\]
Then \(u=x/b\) and \(v=y/a\) satisfy
\eqref{eq:general-game-complementarity}.  This proves the claimed
correspondence, including the normalization used by the decoder.

\paragraph{Bases and their basic variables.}
We write the feasibility constraints in standard form as
\[
 W\binom{s}{z}=q(t),\qquad s,z\geq\zero,
 \qquad W=[\,I_d\ \ C\,].
\]
Thus there are \(d\) equations and \(2d\) nonnegative variables.
A \emph{basis} \(\mathcal B\) is a choice of \(d\) linearly independent
columns of \(W\).  The variables corresponding to those columns are
the \emph{basic variables}.  The remaining, \emph{nonbasic variables}
are set to zero, after which the basic variables are uniquely
determined by
\[
 \xi_{\mathcal B}(t)=W_{\mathcal B}^{-1}q(t).
\]
A basis is \textit{feasible} at \(t\) if every coordinate of
\(\xi_{\mathcal B}(t)\) is nonnegative.  Basic variables need not be
positive, or even nonnegative, for an arbitrary basis.

\paragraph{An effective choice of the perturbation size.}
We write every entry \(W_{rh}\) in reduced rational form as
\(p_{rh}/b_{rh}\), where \(b_{rh}\geq1\); integers, including zero,
have denominator \(1\).  Define
\[
 \delta=\prod_{r=1}^{d}\prod_{h=1}^{2d}b_{rh},
 \qquad
 \overline W=\delta W,
 \qquad
 H=\max\{1,\delta,\max_{r,h}|\overline W_{rh}|\},
 \qquad
 H_0=d!\,H^d.
\]
Thus \(\delta\) is a positive integer that clears all denominators,
and \(\overline W\) is an integer matrix.

For every basis \(\mathcal B\), the adjugate formula gives
\begin{equation}\label{eq:general-basis-polynomials}
 W_{\mathcal B}^{-1}q(t)
 =
 \frac{\operatorname{adj}(\overline W_{\mathcal B})\,
             \delta q(t)}
      {\det\overline W_{\mathcal B}}.
\end{equation}
The denominator is a fixed nonzero integer.  Each coordinate of the
numerator is an integer polynomial of degree at most \(d\).
Moreover, none of these polynomials is identically zero.  To see
this, let \((c_1,\ldots,c_d)\) be the corresponding row of
\(\operatorname{adj}(\overline W_{\mathcal B})\).  This row is
nonzero because the adjugate of an invertible matrix is invertible.
Its numerator polynomial is
\[
 f(t)=\delta\sum_{r=1}^{d}c_r
      +\delta\sum_{r=1}^{d}c_rt^r.
\]
If \(f\) were identically zero, the coefficient of \(t^r\) would
give \(\delta c_r=0\) for every \(r\), a contradiction.

Expanding a cofactor as a sum over permutations bounds its absolute
value by \((d-1)!\,H^{d-1}\).  Since \(\delta\leq H\), each
nonconstant coefficient of \(f\) has absolute value at most
\((d-1)!\,H^d\).  The constant coefficient is a sum of \(d\) such
terms, so every coefficient has absolute value at most \(H_0\).
These bounds hold simultaneously for every basis, including
infeasible bases.

Choose
\begin{equation}\label{eq:general-perturbation-size}
 \kappa=\left\lceil\log_2\bigl(4(d+1)H_0\bigr)\right\rceil,
 \qquad
 \varepsilon=2^{-\kappa}.
\end{equation}
For any numerator polynomial
\(f(t)=\sum_{\nu=0}^{d}a_\nu t^\nu\), let \(\nu_0\) be the
smallest index with \(a_{\nu_0}\neq0\).  For
\(0<t\leq\varepsilon<1\),
\[
 f(t)=t^{\nu_0}
 \left(a_{\nu_0}
       +\sum_{\nu>\nu_0}a_\nu t^{\nu-\nu_0}\right),
 \qquad
 \left|\sum_{\nu>\nu_0}a_\nu t^{\nu-\nu_0}\right|
 \leq dH_0t
 \leq\frac{d}{4(d+1)}
 <1\leq|a_{\nu_0}|.
\]
The final inequality uses the fact that \(a_{\nu_0}\) is a
nonzero integer.  Therefore, \(f(t)\) never vanishes on
\((0,\varepsilon]\), and its sign throughout that interval is the
sign of \(a_{\nu_0}\).  By
\eqref{eq:general-basis-polynomials}, every basic variable of every
basis has a nonzero, constant sign on this interval.  It follows that
a basis feasible at one parameter in \((0,\varepsilon]\) is feasible
throughout the interval, with every basic variable strictly positive.

This is the only place where we need a bound that is uniform over
\emph{all} bases.  The reduction never enumerates these bases.  The
uniform coefficient bound above merely lets us choose one value of
\(\varepsilon\) that works simultaneously for every possible basis that
could arise later.  Once the signed Lemke--Howson procedure returns one
particular equilibrium at \(t=\varepsilon\), the continuation argument
will use only the finitely many basic-variable polynomials belonging to
that one complementary basis.  Keeping these two uses of the polynomial
argument separate will be useful below.

All quantities used to compute \(\varepsilon\) have polynomial
binary length.  Indeed, \(\log_2\delta\) is bounded by the sum of
the binary lengths of the denominators of \(W\); the entries of
\(\overline W\) and the integer \(H\) therefore have polynomial
binary length.  Furthermore,
$\log_2 H_0=\mathcal O(d\log(d+1)+d\log H)$,
and  $\kappa=\mathcal O(d\log(d+1)+d\log H)$.
These integers can be computed in polynomial time, without
enumerating the bases.  Since the powers in \(q(t)\) have degree at
most \(d\), the entries of \(q(\varepsilon)\), and hence of
\(A(\varepsilon)\) and \(B(\varepsilon)\), also have polynomial
binary length and can be computed exactly in polynomial time.

\paragraph{The perturbed game is globally nondegenerate.}
Fix \(0<t\leq\varepsilon\).  The polytope \(\mathcal R(t)\) is
bounded: positivity of the payoff entries gives,
\[
 0\leq u_i\leq\frac{\beta_1(t)}{B_{i1}},
 \qquad
 0\leq v_j\leq\frac{\alpha_1(t)}{A_{1j}}.
\]
It is full dimensional because a sufficiently small vector
\(z>\zero\) satisfies \(Cz<q(t)\).

The affine map \(z\mapsto(q(t)-Cz,z)\) identifies
\(\mathcal R(t)\) with the standard-form polytope
\[
 \{\xi\geq\zero:W\xi=q(t)\}.
\]
At a vertex of this latter polytope, the columns of \(W\)
corresponding to positive coordinates of \(\xi\) must be linearly
independent.  Otherwise a nonzero vector \(h\), supported on those
coordinates and satisfying \(Wh=\zero\), would make both
\(\xi+\eta h\) and \(\xi-\eta h\) feasible for sufficiently small
\(\eta>0\), contradicting extremality.  Extend those independent
columns to a basis of \(W\).  The resulting basic solution is
\(\xi\), and the sign argument above says that all its \(d\) basic
coordinates are strictly positive.  Hence every vertex has exactly
\(d\) positive coordinates and exactly \(d\) zero coordinates.

A zero coordinate \(z_r\) means that the corresponding nonnegativity
inequality binds; a zero coordinate \(s_r\) means that the
corresponding payoff inequality binds.  Thus exactly \(d\) defining
inequalities bind at each vertex of \(\mathcal R(t)\).
No feasible point can have more than \(d\) binding inequalities.

We now deduce game nondegeneracy directly.  Suppose a mixed row
strategy \(x\) with support size \(h\) had \(b>h\) pure column best
responses in \(\mathcal G(t)\).  We set
\[
 p=\max_j(B(t)^{\transpose}x)_j>0,\qquad u=x/p,\qquad v=\zero.
\]
Then \((u,v)\in\mathcal R(t)\).  At this point, \(m-h\)
nonnegativity inequalities for \(u\), all \(n\) nonnegativity
inequalities for \(v\), and \(b\) column-payoff inequalities bind.
Their total is \(d-h+b>d\), a contradiction.  Interchanging the
players gives the same conclusion for every mixed column strategy.
Thus no mixed strategy has more pure best responses than the size
of its support, which is precisely global nondegeneracy.

\paragraph{The signed Lemke--Howson step at \(t=\varepsilon\).}
We recall the precise consequence of the nondegenerate membership theorem
that is used here.  For a rational nondegenerate bimatrix game, shift the
payoffs to be positive and consider the usual labeled product of the two
best-response polytopes.  Fixing one missing label gives a complementary
pivot graph whose components are paths and cycles.  Nondegeneracy makes
every internal vertex incident to exactly two pivot edges and every
completely labeled vertex an endpoint.  The Shapley--Todd determinant
orientation directs each path from its negative endpoint to its positive
endpoint; the artificial completely labeled origin has sign \(-1\), while
a genuine equilibrium endpoint has sign equal to its Shapley index.
Consequently, orienting each pivot edge by exact determinant signs gives a
polynomial-time \(\SoL\) instance in which every sink decodes to a Nash
equilibrium of index \(+1\).  A node is represented by its complementary
basis, and its predecessor and successor are obtained by the exact ratio
test, so all computations use rational arithmetic of polynomial bit
length.

We apply this construction only to the nondegenerate game
\(\mathcal G(\varepsilon)\).  Nothing in the argument below requires an
orientation of the possibly degenerate game at \(t=0\); instead, the output
basis at \(t=\varepsilon\) will be continued algebraically to the original
game.

\paragraph{Obtaining a complementary basis from the perturbed equilibrium.}
We apply Theorem~\ref{thm:membership} to the rational nondegenerate
game \(\mathcal G(\varepsilon)\).  It produces, in polynomial time,
a \SoL\ instance whose every solution decodes to an equilibrium
\((x^\varepsilon,y^\varepsilon)\) of index \(+1\).
Let
\[
 I=\supp(x^\varepsilon),\qquad
 J=\supp(y^\varepsilon).
\]
At equilibrium, every supported strategy is a best response.
Nondegeneracy therefore gives \(|I|=|J|=k\).  It also shows that there are no unused best
responses.

Let \((s^\varepsilon,z^\varepsilon)\) be the associated nonzero
complementarity solution.  Complementarity makes at least one
coordinate in each pair \((s_r^\varepsilon,z_r^\varepsilon)\)
equal to zero.  Hence at least \(d\) inequalities bind at
\(z^\varepsilon\), and the preceding argument shows that exactly
\(d\) bind.  This point must be a vertex: otherwise its minimal
face would have positive dimension, and a vertex of that face
would satisfy an additional binding inequality.  This would give
more than \(d\) binding inequalities, which is impossible.

Define
\[
 K=I\cup\{m+j:j\in J\}\subseteq[d],
 \qquad K^c=[d]\setminus K.
\]
The positive variables at this vertex are exactly
$z_r$ for $r\in K$ and
$s_r$ for $r\in K^c$.
Their columns form a basis \(\mathcal B\) of \(W\), called here
the complementary basis.  To see explicitly what its
nonsingularity implies, order its columns as \(z_K,s_{K^c}\)
and its rows as \(K,K^c\).  Its matrix then has the form
\[
 W_{\mathcal{B}}=\begin{pmatrix}
  C_{K,K}&0\\
  C_{K^c,K}&I_{d-2k}
 \end{pmatrix},
 \qquad
 C_{K,K}=
 \begin{pmatrix}
  0&A_{I,J}\\
  (B_{I,J})^{\transpose}&0
 \end{pmatrix}.
\]
Since $W_{\mathcal{B}}$ is nonsingular, it follows that \(C_{K,K}\) is nonsingular, and consequently both
\(A_{I,J}\) and \((B_{I,J})^{\transpose}\) are nonsingular.

\paragraph{Exact continuation and decoding.}
Keep the complementary basis \(\mathcal B\) obtained above fixed and vary
only the right-hand side \(q(t)\).  Its nonbasic variables are
\(s_K\) and \(z_{K^c}\), and hence they remain identically zero.  Solving
the basic equalities gives
\begin{equation}\label{eq:general-support-continuation}
 \begin{aligned}
 u_I(t)&=((B_{I,J})^{\transpose})^{-1}\beta_J(t),
 &u_{[m]\setminus I}(t)&=\zero,\\
 v_J(t)&=A_{I,J}^{-1}\alpha_I(t),
 &v_{[n]\setminus J}(t)&=\zero.
 \end{aligned}
\end{equation}
The remaining basic variables are the off-support slacks (i.e., corresponding to a pure strategy that is not used with positive probability in the equilibrium).  To make the
continuation explicit, for \(i\notin I\) and \(j\notin J\) define
\begin{equation}\label{eq:general-continuation-slacks}
 \begin{aligned}
 r_i(t)&:=\alpha_i(t)-A_{i,J}v_J(t),\\
 c_j(t)&:=\beta_j(t)-B_{I,j}^{\transpose}u_I(t).
 \end{aligned}
\end{equation}
Here \(r_i(t)\) is the slack of the row-payoff inequality indexed by
\(i\), and \(c_j(t)\) is the slack of the column-payoff inequality indexed
by \(j\).  On the support, the complementary slacks vanish identically:
\begin{equation}\label{eq:general-supported-slacks-zero}
 \alpha_I(t)-A_{I,J}v_J(t)\equiv\zero,
 \qquad
 \beta_J(t)-(B_{I,J})^{\transpose}u_I(t)\equiv\zero.
\end{equation}
Thus the \(d\) basic variables of \(\mathcal B\) are precisely
\[
 \{u_i(t):i\in I\}\cup\{v_j(t):j\in J\}
 \cup\{r_i(t):i\notin I\}\cup\{c_j(t):j\notin J\}.
\]

The matrices \(A_{I,J}^{-1}\) and
\(((B_{I,J})^{\transpose})^{-1}\) are fixed rational matrices, while
each coordinate of \(\alpha(t)\) and \(\beta(t)\) has the form
\(1+t^h\).  Therefore, every coordinate of \(u_I(t)\) and
\(v_J(t)\) is a polynomial in \(t\) with rational coefficients.  The same
is then true of every slack \(r_i(t)\) and \(c_j(t)\) in
\eqref{eq:general-continuation-slacks}.  Equivalently, these are exactly
the basic-variable polynomials obtained from
\eqref{eq:general-basis-polynomials} for the returned basis
\(\mathcal B\).

At \(t=\varepsilon\), the equilibrium of the nondegenerate game has
support pair \((I,J)\) and no unused best responses.  Hence
\[
 u_i(\varepsilon)>0\quad(i\in I),\qquad
 v_j(\varepsilon)>0\quad(j\in J),
\]
and
\[
 r_i(\varepsilon)>0\quad(i\notin I),\qquad
 c_j(\varepsilon)>0\quad(j\notin J).
\]
The uniform sign-stability argument above applies to every one of these
basic-variable polynomials.  None can vanish on \((0,\varepsilon]\), and
therefore
\begin{equation}\label{eq:general-positive-continuation}
 \begin{aligned}
 u_i(t)&>0 &&(i\in I),&
 v_j(t)&>0 &&(j\in J),\\
 r_i(t)&>0 &&(i\notin I),&
 c_j(t)&>0 &&(j\notin J)
 \end{aligned}
 \qquad (0<t\leq\varepsilon).
\end{equation}
Together with \eqref{eq:general-supported-slacks-zero}, this proves
feasibility and complementarity for the same basis throughout
\((0,\varepsilon]\).  In particular, after normalization we obtain the mixed strategies $x$ and $y$ whose support pair
is exactly \((I,J)\) for every positive parameter in this interval.

The basis
is fixed once at \(t=\varepsilon\), and all of its variables are explicit
polynomials whose values at the original game are obtained by direct
substitution \(t=0\).

\paragraph{Evaluation at the original game.}
Substituting \(t=0\) in
\eqref{eq:general-support-continuation} gives
\begin{equation}\label{eq:general-exact-decoder}
 \begin{aligned}
 u_I^0&=((B_{I,J})^{\transpose})^{-1}\one,
 &u_{[m]\setminus I}^0&=\zero,\\
 v_J^0&=A_{I,J}^{-1}\one,
 &v_{[n]\setminus J}^0&=\zero.
 \end{aligned}
\end{equation}
Since all the functions in \eqref{eq:general-positive-continuation} are
polynomials, they are continuous at zero.  Taking \(t\downarrow0\) in the
strict inequalities therefore gives
\[
 u_i^0\geq0\quad(i\in I),\qquad
 v_j^0\geq0\quad(j\in J),
\]
and, with
\[
 r_i^0:=1-A_{i,J}v_J^0,
 \qquad
 c_j^0:=1-B_{I,j}^{\transpose}u_I^0,
\]
we obtain
\[
 r_i^0\geq0\quad(i\notin I),\qquad
 c_j^0\geq0\quad(j\notin J).
\]
The supported slack equalities
\eqref{eq:general-supported-slacks-zero} remain exact at zero.  Hence, if
\(z^0=(u^0,v^0)\) and \(s^0=\one-Cz^0\), then
\[
 z^0\geq\zero,\qquad s^0\geq\zero,
 \qquad z_r^0s_r^0=0\quad(r\in[d]).
\]
Thus \((s^0,z^0)\) is an exact complementary solution for the original
shifted game.

Some of the basic variables in
\eqref{eq:general-positive-continuation} may become zero at exactly
\(t=0\); the nonvanishing statement was only for positive \(t\).  This is
precisely the phenomenon that can make the decoded equilibrium degenerate.
Nevertheless,
\[
 A_{I,J}v_J^0=\one,
 \qquad
 (B_{I,J})^{\transpose}u_I^0=\one,
\]
so neither \(u^0\) nor \(v^0\) is zero.  Consequently
\[
 x^0=\frac{u^0}{\one^{\transpose}u^0},
 \qquad
 y^0=\frac{v^0}{\one^{\transpose}v^0}
\]
is an exact Nash equilibrium of \(\mathcal G(0)\), and hence of the
original, unshifted input game.

The decoder only needs the supports \(I,J\) and the two rational linear
systems in \eqref{eq:general-exact-decoder}.  Their matrices have
polynomial dimension and binary length.  After clearing denominators,
Cramer's rule and the standard determinant bound give polynomial binary
length for every coordinate of \(u^0,v^0\), and therefore for the
normalized strategies \(x^0,y^0\).  All linear solves and comparisons are
exact rational computations, so the decoder runs in polynomial time.

\paragraph{What can change at \(t=0\).}
For every positive \(t\leq\varepsilon\), the support pair is exactly
\((I,J)\).  At \(t=0\), however, one or more of the supported-coordinate
polynomials \(u_i(t)\) or \(v_j(t)\) may evaluate to zero, and one or more
off-support slacks may also evaluate to zero.  Both events are harmless for
the search problem.  If a supported probability disappears, the
corresponding pure strategy still satisfies the equal-payoff equation in
\eqref{eq:general-supported-slacks-zero}; it therefore remains a best
response while being unused.  Likewise, a vanishing off-support slack
creates an unused best response directly.  Either event makes the decoded
equilibrium nonregular.  Thus the only case in which further work is
needed is when \((x^0,y^0)\) is regular; regularity then forces the support
and strict off-support inequalities to be exactly the ones present for
\(t>0\).

\paragraph{A regular decoded equilibrium has index \(+1\).}
If \((x^0,y^0)\) is nonregular, it is an accepted output.
Suppose instead that it is regular.  Its supports are contained
in \(I,J\).  For every \(i\in I\), the equality
\((Av^0)_i=1\) still holds, and every row payoff is at most \(1\).
Thus all strategies in \(I\) remain best responses after
normalization.  Similarly, all strategies in \(J\) remain best
responses.  Since a regular equilibrium has no unused best
response, no supported probability can have disappeared:
\[
 \supp(x^0)=I,\qquad \supp(y^0)=J.
\]

For completeness, the index comparison can be made directly with
bordered determinants.  For every nonsingular square matrix \(M\),
the block determinant identity gives
\[
 \Delta(M)=\det(M)\,\one^{\transpose}M^{-1}\one.
\]
For \(0\leq t\leq\varepsilon\), the support matrices are
nonsingular and
\[
 \begin{aligned}
 (A(t)_{I,J})^{-1}\one&=v_J(t),\\
 ((B(t)_{I,J})^{\transpose})^{-1}\one&=u_I(t).
 \end{aligned}
\]
Both coordinate sums are positive, including at \(t=0\).
Moreover, the scalings in
\eqref{eq:general-game-perturbation} give
\[
 \det A(t)_{I,J}
   =\frac{\det A_{I,J}}{\prod_{i\in I}\alpha_i(t)},
 \qquad
 \det B(t)_{I,J}
   =\frac{\det B_{I,J}}{\prod_{j\in J}\beta_j(t)}.
\]
All factors in these denominators are positive.  Then it follows that
\[
 \begin{aligned}
 \ind_{A,B}(x^0,y^0)
 &=(-1)^{k+1}\sign\!\left(\det A_{I,J}\det B_{I,J}\right)\\
 &=\ind_{A(\varepsilon),B(\varepsilon)}
            (x^\varepsilon,y^\varepsilon)
 =+1.
 \end{aligned}
\]
The initial payoff shifts preserve this index.  Thus every
solution of the constructed \SoL\ instance decodes to a valid
output of \PINNR.

\paragraph{Verification and polynomial bounds.}

The construction of \(\mathcal G(\varepsilon)\), the reduction
from that game to \SoL, and the exact decoder are all polynomial
time.  The decoder also proves the existence of an accepted
rational output of polynomial encoding length for every input
game.  This establishes membership in \PPADS\ without a promise.
Together with the hardness argument, it proves the corollary.
\end{proof}
\section{Conclusions and open questions}\label{sec:conclusion}


We proved that computing an exact Nash equilibrium of index $+1$ in a rational bimatrix game promised to be nondegenerate is $\PPADS$-complete. We also proved that finding either a regular equilibrium of index $+1$ or a nonregular equilibrium in an arbitrary rational bimatrix game is $\PPADS$-complete. The hardness reduction always produces a game satisfying the nondegeneracy promise. On nondegenerate inputs, membership uses the ordinary oriented Lemke--Howson graph: for each fixed missing label, nondegeneracy makes the corresponding complementary pivot graph a disjoint union of paths and cycles, and its determinant orientation directs every path from the endpoint of incidence sign $-1$ to the endpoint of incidence sign $+1$. In particular, membership for the promised problem requires no perturbation, while the promise-free extension uses an explicit rational perturbation and exact decoding.

For hardness, the edges of the original \SoL\ graph are
routed on a two-dimensional grid, with crossings replaced by planar switches.
Interpolation gives a map with one regular fixed point at each sink and
noncanonical source, with indices \(+1\) and \(-1\), respectively.
The canonical source contributes none, and the boundary coloring excludes
additional fixed points from clipping. The exact interpolation circuit is
strict at every fixed point. The subsequent LCP and imitation game
reductions preserve this correspondence and its indices. A final rational
perturbation gives global nondegeneracy while preserving supports and
indices, allowing every positive index equilibrium to decode a sink.

An immediate question is whether
there is a natural approximate version for which \(\PPADS\)-hardness holds at
inverse polynomial accuracy. Finally, it would also be interesting to seek analogous sink only completeness results for standard equilibrium refinements, or for restricted classes of bimatrix games in which index has additional economic or combinatorial structure.

\paragraph{Acknowledgments.} Ioannis Panageas and Jingming Yan were supported by NSF CCF-2454115. 
Andreas Kontogiannis and Vasilis Pollatos acknowledge the financial support of the Archimedes Research Unit (Athena RC) during their PhD. 
The authors acknowledge the use of ChatGPT 5.6 and 6 for assistance with auditing the paper, Lemmas from Section 4, Lemma 6.6, and Section 6.5. The authors have rewritten the relevant portions and take full responsibility for the content and correctness of the final manuscript.

\begingroup
\small
\bibliographystyle{plainnat}
\bibliography{main}

@article{Nash1951,
  author  = {Nash, John F.},
  title   = {Non-Cooperative Games},
  journal = {Annals of Mathematics},
  volume  = {54},
  number  = {2},
  pages   = {286--295},
  year    = {1951},
  doi     = {10.2307/1969529}
}

@inproceedings{yarotsky2018optimal,
  title={Optimal approximation of continuous functions by very deep ReLU networks},
  author={Yarotsky, Dmitry},
  booktitle={Conference on learning theory},
  pages={639--649},
  year={2018},
  organization={PMLR}
}

@article{fearnley2025complexity,
  title={The complexity of computing KKT solutions of quadratic programs},
  author={Fearnley, John and Goldberg, Paul and Hollender, Alexandros and Savani, Rahul},
  journal={Journal of the ACM},
  volume={72},
  number={5},
  pages={1--48},
  year={2025},
  publisher={ACM New York, NY}
}

@article{MegiddoPapadimitriou1991,
  author  = {Megiddo, Nimrod and Papadimitriou, Christos H.},
  title   = {On Total Functions, Existence Theorems and Computational Complexity},
  journal = {Theoretical Computer Science},
  volume  = {81},
  number  = {2},
  pages   = {317--324},
  year    = {1991},
  doi     = {10.1016/0304-3975(91)90200-L}
}

@article{Papadimitriou1994,
  author  = {Papadimitriou, Christos H.},
  title   = {On the Complexity of the Parity Argument and Other Inefficient Proofs of Existence},
  journal = {Journal of Computer and System Sciences},
  volume  = {48},
  number  = {3},
  pages   = {498--532},
  year    = {1994},
  doi     = {10.1016/S0022-0000(05)80063-7}
}

@article{BeameEtAl1998,
  author  = {Beame, Paul and Cook, Stephen A. and Edmonds, Jeff and Impagliazzo, Russell and Pitassi, Toniann},
  title   = {The Relative Complexity of {NP} Search Problems},
  journal = {Journal of Computer and System Sciences},
  volume  = {57},
  number  = {1},
  pages   = {3--19},
  year    = {1998},
  doi     = {10.1006/jcss.1998.1575}
}

@article{GoldbergPapadimitriou2018,
  author  = {Goldberg, Paul W. and Papadimitriou, Christos H.},
  title   = {Towards a Unified Complexity Theory of Total Functions},
  journal = {Journal of Computer and System Sciences},
  volume  = {94},
  pages   = {167--192},
  year    = {2018},
  doi     = {10.1016/j.jcss.2017.12.003}
}

@article{LemkeHowson1964,
  author  = {Lemke, Carlton E. and Howson, Jr., Joseph T.},
  title   = {Equilibrium Points of Bimatrix Games},
  journal = {Journal of the Society for Industrial and Applied Mathematics},
  volume  = {12},
  number  = {2},
  pages   = {413--423},
  year    = {1964},
  doi     = {10.1137/0112033}
}

@incollection{Shapley1974,
  author    = {Shapley, Lloyd S.},
  title     = {A Note on the {Lemke--Howson} Algorithm},
  booktitle = {Pivoting and Extensions},
  editor    = {Balinski, Michel L.},
  series    = {Mathematical Programming Studies},
  volume    = {1},
  pages     = {175--189},
  publisher = {North-Holland},
  year      = {1974},
  doi       = {10.1007/BFb0121248}
}

@article{Harsanyi1973,
  author  = {Harsanyi, John C.},
  title   = {Oddness of the Number of Equilibrium Points: A New Proof},
  journal = {International Journal of Game Theory},
  volume  = {2},
  pages   = {235--250},
  year    = {1973},
  doi     = {10.1007/BF01737572}
}

@article{Jansen1981,
  author  = {Jansen, M. J. M.},
  title   = {Regularity and Stability of Equilibrium Points of Bimatrix Games},
  journal = {Mathematics of Operations Research},
  volume  = {6},
  number  = {4},
  pages   = {530--550},
  year    = {1981},
  doi     = {10.1287/moor.6.4.530}
}

@article{Todd1976,
  author  = {Todd, Michael J.},
  title   = {Orientation in Complementary Pivot Algorithms},
  journal = {Mathematics of Operations Research},
  volume  = {1},
  number  = {1},
  pages   = {54--66},
  year    = {1976},
  doi     = {10.1287/moor.1.1.54}
}

@book{CottlePangStone1992,
  author    = {Cottle, Richard W. and Pang, Jong-Shi and Stone, Richard E.},
  title     = {The Linear Complementarity Problem},
  publisher = {Academic Press},
  address   = {San Diego},
  year      = {1992}
}

@book{Schrijver1986,
  author    = {Schrijver, Alexander},
  title     = {Theory of Linear and Integer Programming},
  publisher = {John Wiley \& Sons},
  address   = {Chichester},
  year      = {1986}
}

@article{DaskalakisGoldbergPapadimitriou2009,
  author  = {Daskalakis, Constantinos and Goldberg, Paul W. and Papadimitriou, Christos H.},
  title   = {The Complexity of Computing a {Nash} Equilibrium},
  journal = {SIAM Journal on Computing},
  volume  = {39},
  number  = {1},
  pages   = {195--259},
  year    = {2009},
  doi     = {10.1137/070699652}
}

@article{ChenDengTeng2009,
  author  = {Chen, Xi and Deng, Xiaotie and Teng, Shang-Hua},
  title   = {Settling the Complexity of Computing Two-Player {Nash} Equilibria},
  journal = {Journal of the ACM},
  volume  = {56},
  number  = {3},
  pages   = {14:1--14:57},
  year    = {2009},
  doi     = {10.1145/1516512.1516516}
}

@article{SavaniVonStengel2006,
  author  = {Savani, Rahul and von Stengel, Bernhard},
  title   = {Hard-to-Solve Bimatrix Games},
  journal = {Econometrica},
  volume  = {74},
  number  = {2},
  pages   = {397--429},
  year    = {2006},
  doi     = {10.1111/j.1468-0262.2006.00667.x}
}

@article{VeghVonStengel2015,
  author  = {V{\'e}gh, L{\'a}szl{\'o} A. and von Stengel, Bernhard},
  title   = {Oriented {Euler} Complexes and Signed Perfect Matchings},
  journal = {Mathematical Programming},
  volume  = {150},
  pages   = {153--178},
  year    = {2015},
  doi     = {10.1007/s10107-014-0770-4}
}

@article{vonStengel2021,
  author  = {von Stengel, Bernhard},
  title   = {Finding {Nash} Equilibria of Two-Player Games},
  journal = {arXiv preprint arXiv:2102.04580},
  year    = {2021},
  doi     = {10.48550/arXiv.2102.04580}
}

@article{EtessamiYannakakis2010,
  author  = {Etessami, Kousha and Yannakakis, Mihalis},
  title   = {On the Complexity of {Nash} Equilibria and Other Fixed Points},
  journal = {SIAM Journal on Computing},
  volume  = {39},
  number  = {6},
  pages   = {2531--2597},
  year    = {2010},
  doi     = {10.1137/080720826}
}

@article{Mehta2018,
  author  = {Mehta, Ruta},
  title   = {Constant Rank Two-Player Games are {PPAD}-Hard},
  journal = {SIAM Journal on Computing},
  volume  = {47},
  number  = {5},
  pages   = {1858--1887},
  year    = {2018},
  doi     = {10.1137/15M1032338}
}

@article{McLennanTourky2010,
  author  = {McLennan, Andrew and Tourky, Rabee},
  title   = {Imitation Games and Computation},
  journal = {Games and Economic Behavior},
  volume  = {70},
  number  = {1},
  pages   = {4--11},
  year    = {2010},
  doi     = {10.1016/j.geb.2009.08.003}
}

@inproceedings{FriedlEtAl2006,
  author    = {Friedl, Katalin and Ivanyos, G{\'a}bor and Santha, Mikl{\'o}s and Verhoeven, Yves F.},
  title     = {Locally 2-Dimensional {Sperner} Problems Complete for the Polynomial Parity Argument Classes},
  booktitle = {Algorithms and Complexity (CIAC 2006)},
  series    = {Lecture Notes in Computer Science},
  volume    = {3998},
  pages     = {380--391},
  publisher = {Springer},
  year      = {2006},
  doi       = {10.1007/11758471_36}
}

@techreport{HollenderGoldberg2018,
  author      = {Hollender, Alexandros and Goldberg, Paul W.},
  title       = {The Complexity of Multi-source Variants of the {End-of-Line} Problem, and the Concise Mutilated Chessboard},
  institution = {Electronic Colloquium on Computational Complexity},
  number      = {TR18-120},
  year        = {2018},
  url         = {https://eccc.weizmann.ac.il/report/2018/120/}
}

@incollection{Daskalakis2019,
  author    = {Daskalakis, Constantinos},
  title     = {Equilibria, Fixed Points, and Computational Complexity---{Nevanlinna Prize Lecture}},
  booktitle = {Proceedings of the International Congress of Mathematicians},
  volume    = {1},
  pages     = {147--210},
  publisher = {World Scientific},
  year      = {2019},
  doi       = {10.1142/9789813272880_0009}
}

@article{GoosEtAl2024,
  author  = {G{\"o}{\"o}s, Mika and Hollender, Alexandros and Jain, Siddhartha and Maystre, Gilbert and Pires, William and Robere, Robert and Tao, Ran},
  title   = {Further Collapses in {TFNP}},
  journal = {SIAM Journal on Computing},
  volume  = {53},
  number  = {3},
  pages   = {573--587},
  year    = {2024},
  doi     = {10.1137/22M1498346}
}

@article{Du2013,
  author  = {Du, Ye},
  title   = {On the Complexity of Deciding Degeneracy in a Bimatrix Game with Sparse Payoff Matrix},
  journal = {Theoretical Computer Science},
  volume  = {472},
  pages   = {104--109},
  year    = {2013},
  doi     = {10.1016/j.tcs.2012.10.053}
}

@article{IckstadtTheobaldVonStengel2025,
  author  = {Ickstadt, Constantin and Theobald, Thorsten and von Stengel, Bernhard},
  title   = {A Stable-Set Bound and Maximal Numbers of {Nash} Equilibria in Bimatrix Games},
  journal = {Mathematics of Operations Research},
  year    = {2025},
  doi     = {10.1287/moor.2024.0809}
}

@inproceedings{VonSchemdeVonStengel2008,
  author    = {{von Schemde}, Arndt and {von Stengel}, Bernhard},
  title     = {Strategic Characterization of the Index of an Equilibrium},
  booktitle = {Algorithmic Game Theory},
  series    = {Lecture Notes in Computer Science},
  volume    = {4997},
  pages     = {242--254},
  publisher = {Springer},
  year      = {2008},
  doi       = {10.1007/978-3-540-79309-0_22}
}

@article{GovindanLarakiPahl2023,
  author  = {Govindan, Srihari and Laraki, Rida and Pahl, Lucas},
  title   = {On Sustainable Equilibria},
  journal = {Journal of Economic Theory},
  volume  = {213},
  pages   = {105736},
  year    = {2023},
  doi     = {10.1016/j.jet.2023.105736}
}

@article{Balthasar2010,
  author  = {Balthasar, Anne},
  title   = {Equilibrium Tracing in Strategic-Form Games},
  journal = {Economic Theory},
  volume  = {42},
  pages   = {39--54},
  year    = {2010},
  doi     = {10.1007/s00199-009-0442-4}
}

@article{ChenDeng2009,
  author  = {Chen, Xi and Deng, Xiaotie},
  title   = {On the Complexity of {2D} Discrete Fixed Point Problem},
  journal = {Theoretical Computer Science},
  volume  = {410},
  number  = {44},
  pages   = {4448--4456},
  year    = {2009},
  doi     = {10.1016/j.tcs.2009.07.052}
}

@inproceedings{GhoshGoldbergHollender2026,
  author    = {Ghosh, Abheek and Goldberg, Paul W. and Hollender, Alexandros},
  title     = {The Complexity of Computing a Unique {Nash} Equilibrium},
  booktitle = {Proceedings of the 67th Annual IEEE Symposium on Foundations of
               Computer Science (FOCS)},
  year      = {2026},
  publisher = {IEEE}
}

@article{GoldbergPapadimitriouSavani2013,
  author    = {Goldberg, Paul W. and Papadimitriou, Christos H. and Savani, Rahul},
  title     = {The Complexity of the Homotopy Method, Equilibrium Selection,
               and {Lemke--Howson} Solutions},
  journal   = {ACM Transactions on Economics and Computation},
  volume    = {1},
  number    = {2},
  pages     = {1--25},
  year      = {2013},
  doi       = {10.1145/2465769.2465774}
}

@inproceedings{DaskalakisPapadimitriou2011,
  author    = {Daskalakis, Constantinos and Papadimitriou, Christos H.},
  title     = {Continuous Local Search},
  booktitle = {Proceedings of the Twenty-Second Annual ACM-SIAM Symposium on Discrete Algorithms},
  pages     = {790--804},
  publisher = {SIAM},
  year      = {2011},
  doi       = {10.1137/1.9781611973082.62}
}

@article{FearnleyGordonMehtaSavani2020,
  author  = {Fearnley, John and Gordon, Spencer and Mehta, Ruta and Savani, Rahul},
  title   = {Unique End of Potential Line},
  journal = {Journal of Computer and System Sciences},
  volume  = {114},
  pages   = {1--35},
  year    = {2020},
  doi     = {10.1016/j.jcss.2020.05.007}
}

@article{FearnleyGoldbergHollenderSavani2023,
  author  = {Fearnley, John and Goldberg, Paul W. and Hollender, Alexandros and Savani, Rahul},
  title   = {The Complexity of Gradient Descent: {CLS} = {PPAD} $\cap$ {PLS}},
  journal = {Journal of the ACM},
  volume  = {70},
  number  = {1},
  pages   = {7:1--7:74},
  year    = {2023},
  doi     = {10.1145/3568163}
}

@article{KohlbergMertens1986,
  author  = {Kohlberg, Elon and Mertens, Jean-Fran{\c{c}}ois},
  title   = {On the Strategic Stability of Equilibria},
  journal = {Econometrica},
  volume  = {54},
  number  = {5},
  pages   = {1003--1037},
  year    = {1986}
}

@book{HarsanyiSelten1988,
  author    = {Harsanyi, John C. and Selten, Reinhard},
  title     = {A General Theory of Equilibrium Selection in Games},
  publisher = {MIT Press},
  year      = {1988}
}

@article{Ritzberger1994,
  author  = {Ritzberger, Klaus},
  title   = {The Theory of Normal Form Games from the
             Differentiable Viewpoint},
  journal = {International Journal of Game Theory},
  volume  = {23},
  number  = {3},
  pages   = {207--236},
  year    = {1994}
}

@article{RitzbergerWeibull1995,
  author  = {Ritzberger, Klaus and Weibull, J{\"o}rgen W.},
  title   = {Evolutionary Selection in Normal-Form Games},
  journal = {Econometrica},
  volume  = {63},
  number  = {6},
  pages   = {1371--1399},
  year    = {1995}
}

@article{DeMichelisGermano2000,
  author  = {DeMichelis, Stefano and Germano, Fabrizio},
  title   = {On the Indices of Zeros of {Nash} Fields},
  journal = {Journal of Economic Theory},
  volume  = {94},
  number  = {2},
  pages   = {192--217},
  year    = {2000}
}

@article{McKelveyPalfrey1995,
  author  = {McKelvey, Richard D. and Palfrey, Thomas R.},
  title   = {Quantal Response Equilibria for Normal Form Games},
  journal = {Games and Economic Behavior},
  volume  = {10},
  number  = {1},
  pages   = {6--38},
  year    = {1995}
}

@article{Turocy2005,
  author  = {Turocy, Theodore L.},
  title   = {A Dynamic Homotopy Interpretation of the Logistic
             Quantal Response Equilibrium Correspondence},
  journal = {Games and Economic Behavior},
  volume  = {51},
  number  = {2},
  pages   = {243--263},
  year    = {2005}
}
\endgroup

\end{document}